\documentclass{theoretics}

\newcommand{\dclosp}[1]{\ensuremath{\mathord{\downarrow_{#1}}}\xspace}
\newcommand{\dclosr}{\dclosp{R}}

\newcommand{\canoc}{\ensuremath{\preceq_\Cs}\xspace}

\newcommand{\canec}{\ensuremath{\sim_\Cs}\xspace}

\newcommand{\canep}{\ensuremath{\sim_{\pol{\Cs}}}\xspace}

\newcommand{\canpol}{\ensuremath{\preceq_{\copol{\Cs}}}\xspace}
\newcommand{\canppol}{\ensuremath{\preceq_{\pol{\Cs}}}\xspace}

\newcommand{\qanoc}{\ensuremath{\leq_\Cs}\xspace}
\newcommand{\qanop}{\ensuremath{\leq_{\pol{\Cs}}}\xspace}
\newcommand{\canoac}{\ensuremath{\preceq_{\Cs,\alpha}}\xspace}

\newcommand{\canaeq}{\ensuremath{\sim_{\Cs,\alpha}}\xspace}

\newcommand{\etac}{\ensuremath{\eta_\Cs}\xspace}
\newcommand{\canc}{\ensuremath{N_\Cs}\xspace}

\newcommand{\fo}{\ensuremath{\textup{FO}}\xspace}

\newcommand{\fod}{\ensuremath{\fo^2}\xspace}

\newcommand{\fodw}{\ensuremath{\fod(<)}\xspace}
\newcommand{\fodws}{\ensuremath{\fod(<,+1)}\xspace}

\newcommand{\fodwm}{\ensuremath{\fod(<,\mathit{MOD})}\xspace}
\newcommand{\fodwsm}{\ensuremath{\fod(<,+1, \mathit{MOD})}\xspace}
\newcommand{\fodwam}{\ensuremath{\fod(<, \mathit{AMOD})}\xspace}
\newcommand{\fodwsam}{\ensuremath{\fod(<,+1, \mathit{AMOD})}\xspace}
\newcommand{\fodwgm}{\ensuremath{\fod(<,\prefsig{\grp})}\xspace}
\newcommand{\fodwsgm}{\ensuremath{\fod(<,+1,\prefsig{\grp})}\xspace}

\newcommand{\sic}[1]{\ensuremath{\Sigma_{#1}}\xspace}
\newcommand{\siw}[1]{\ensuremath{\Sigma_{#1}(<)}\xspace}

\newcommand{\siwd}{\siw{2}}

\newcommand{\pic}[1]{\ensuremath{\Pi_{#1}}\xspace}
\newcommand{\piw}[1]{\ensuremath{\Pi_{#1}(<)}\xspace}

\newcommand{\piwd}{\piw{2}}

\newcommand{\dec}[1]{\ensuremath{\Delta_{#1}}\xspace}
\newcommand{\dew}[1]{\ensuremath{\Delta_{#1}(<)}\xspace}
\newcommand{\dewm}[1]{\ensuremath{\Delta_{#1}(<,MOD)}\xspace}
\newcommand{\dews}[1]{\ensuremath{\Delta_{#1}(<,+1)}\xspace}

\newcommand{\dewd}{\dew{2}}
\newcommand{\dewmd}{\dewm{2}}
\newcommand{\dewsd}{\dews{2}}

\newcommand{\bsc}[1]{\ensuremath{\Bs\Sigma_{#1}}\xspace}

\newcommand{\fpfm}{\ensuremath{\textup{F} + \textup{P}}\xspace}
\newcommand{\fpfmx}{\ensuremath{\textup{F} + \textup{X} + \textup{P} + \textup{Y}}\xspace}

\newcommand{\at}{\ensuremath{\textup{AT}}\xspace}
\newcommand{\md}{\ensuremath{\textup{MOD}}\xspace}
\newcommand{\abg}{\ensuremath{\textup{AMT}}\xspace}
\newcommand{\grp}{\ensuremath{\textup{GR}}\xspace}
\newcommand{\ul}{\ensuremath{\textup{UL}}\xspace}

\newcommand{\stzer}{\textup{ST}\xspace}

\newcommand{\etaat}{\ensuremath{\eta_\at}\xspace}
\newcommand{\canat}{\ensuremath{N_\at}\xspace}

\newcommand{\bool}[1]{\ensuremath{\mathit{Bool}(#1)}\xspace}
\newcommand{\pol}[1]{\ensuremath{\mathit{Pol}(#1)}\xspace}
\newcommand{\bpol}[1]{\ensuremath{\mathit{BPol}(#1)}\xspace}

\newcommand{\upol}[1]{\ensuremath{\mathit{UPol}(#1)}\xspace}
\newcommand{\copol}[1]{\ensuremath{\mathit{co\textup{-}}\!\pol{#1}}\xspace}
\newcommand{\capol}[1]{\ensuremath{\pol{#1} \cap \copol{#1}}\xspace}
\newcommand{\adet}[1]{\ensuremath{\mathit{APol}(#1)}\xspace}
\newcommand{\wadet}[1]{\ensuremath{\mathit{WAPol}(#1)}\xspace}

\newcommand{\boolo}{\ensuremath{\mathit{Bool}}\xspace}
\newcommand{\polo}{\ensuremath{\mathit{Pol}}\xspace}
\newcommand{\bpolo}{\ensuremath{\mathit{BPol}}\xspace}
\newcommand{\upolo}{\ensuremath{\mathit{UPol}}\xspace}
\newcommand{\adeto}{\ensuremath{\mathit{APol}}\xspace}
\newcommand{\wadeto}{\ensuremath{\mathit{WAPol}}\xspace}

\newcommand{\davar}{\ensuremath{\mathbf{DA}}\xspace}
\newcommand{\ldavar}{\ensuremath{\mathbf{LDA}}\xspace}

\newcommand{\imprint}{imprint\xspace}
\newcommand{\imprints}{imprints\xspace}

\newcommand{\tame}{multiplicative\xspace}

\newcommand{\ratms}{rating maps\xspace}
\newcommand{\ratm}{rating map\xspace}

\newcommand{\Nice}{Finitary\xspace}
\newcommand{\nice}{finitary\xspace}

\newcommand{\Full}{Full\xspace}
\newcommand{\full}{full\xspace}

\newcommand{\fratm}{full \ratm}

\newcommand{\fratms}{\full \ratms}
\newcommand{\Fratms}{\Full \ratms}

\newcommand{\mratm}{multiplicative rating map\xspace}
\newcommand{\mratms}{multiplicative rating maps\xspace}

\newcommand{\Mratms}{Multiplicative rating maps\xspace}

\newcommand{\prin}[2]{\ensuremath{\Is[#1](#2)}\xspace}

\newcommand{\opti}[2]{\ensuremath{\Is_{#1}[#2]}\xspace}

\newcommand{\popti}[3]{\ensuremath{\Ps_{#1}[#2,#3]}\xspace}

\newcommand{\upolopti}{\opti{\upol{\Cs}}{\rho}}

\newcommand{\pupolopti}{\popti{\upol{\Cs}}{\etac}{\rho}}

\newcommand{\typ}[2]{\ensuremath{[#1]_{#2}}\xspace}
\newcommand{\ctype}[1]{\typ{#1}{\Cs}}

\newcommand{\pctype}[1]{\typ{#1}{\pol{\Cs}}}

\usepackage{stmaryrd}

\newcommand{\Bs}{\ensuremath{\mathcal{B}}\xspace}
\newcommand{\Cs}{\ensuremath{\mathcal{C}}\xspace}
\newcommand{\Ds}{\ensuremath{\mathcal{D}}\xspace}

\newcommand{\Fs}{\ensuremath{\mathcal{F}}\xspace}
\newcommand{\Gs}{\ensuremath{\mathcal{G}}\xspace}

\newcommand{\Is}{\ensuremath{\mathcal{I}}\xspace}

\newcommand{\Ps}{\ensuremath{\mathcal{P}}\xspace}

\newcommand{\Hb}{\ensuremath{\mathbf{H}}\xspace}

\newcommand{\Kb}{\ensuremath{\mathbf{K}}\xspace}
\newcommand{\Lb}{\ensuremath{\mathbf{L}}\xspace}

\newcommand{\Ub}{\ensuremath{\mathbf{U}}\xspace}
\newcommand{\Vb}{\ensuremath{\mathbf{V}}\xspace}

\newcommand{\frI}{\ensuremath{\mathbbm{I}}\xspace}

\newcommand{\frP}{\ensuremath{\mathbbm{P}}\xspace}

\newcommand{\frS}{\ensuremath{\mathbbm{S}}\xspace}

\newcommand{\wsuit}{well-suited\xspace}
\newcommand{\Wsuit}{Well-suited\xspace}

\newcommand{\vari}{prevariety\xspace}

\newcommand{\varis}{prevarieties\xspace}

\newcommand{\pvari}{positive prevariety\xspace}

\newcommand{\pvaris}{positive prevarieties\xspace}

\def\inv{^{-1}}

\newcommand{\Jrel}{\ensuremath{\mathrel{\mathscr{J}}}\xspace}

\newcommand{\Rrel}{\ensuremath{\mathrel{\mathscr{R}}}\xspace}
\newcommand{\Lrel}{\ensuremath{\mathrel{\mathscr{L}}}\xspace}

\newcommand{\Jord}{\ensuremath{\leqslant_{\mathscr{J}}}\xspace}

\newcommand{\Rord}{\ensuremath{\leqslant_{\mathscr{R}}}\xspace}
\newcommand{\Lord}{\ensuremath{\leqslant_{\mathscr{L}}}\xspace}

\newcommand{\Jords}{\ensuremath{<_{\mathscr{J}}}\xspace}

\newcommand{\Rords}{\ensuremath{<_{\mathscr{R}}}\xspace}
\newcommand{\Lords}{\ensuremath{<_{\mathscr{L}}}\xspace}

\newcommand{\veps}{\ensuremath{\varepsilon}\xspace}

\DeclareMathOperator{\uclos}{\uparrow}

\newcommand{\nat}{\ensuremath{\mathbb{N}}\xspace}

\newcommand{\tls}{\ensuremath{\textup{TL}}\xspace}
\newcommand{\tlxs}{\ensuremath{\textup{TLX}}\xspace}
\newcommand{\tla}[1]{\ensuremath{\tls[#1]}\xspace}
\newcommand{\tlxa}[1]{\ensuremath{\tlxs[#1]}\xspace}
\newcommand{\tlc}[1]{\ensuremath{\tls(#1)}\xspace}
\newcommand{\tlxc}[1]{\ensuremath{\tlxs(#1)}\xspace}

\newcommand{\finally}[1]{\ensuremath{\textup{F}~#1}\xspace}
\newcommand{\nex}[1]{\ensuremath{\textup{X}~#1}\xspace}

\newcommand{\finallym}[1]{\ensuremath{\textup{P}~#1}\xspace}
\newcommand{\nexm}[1]{\ensuremath{\textup{Y}~#1}\xspace}

\newcommand{\finallyp}[2]{\ensuremath{\textup{F}_{#1}~#2}\xspace}
\newcommand{\finallymp}[2]{\ensuremath{\textup{P}_{#1}~#2}\xspace}
\newcommand{\finallyl}[1]{\finallyp{L}{#1}}
\newcommand{\finallyml}[1]{\finallymp{L}{#1}}

\newcommand{\tleqp}[1]{\ensuremath{\cong_{\eta,#1}}\xspace}
\newcommand{\tleqk}{\tleqp{k}}
\newcommand{\tldp}[1]{\ensuremath{\cong_{\delta,#1}}\xspace}
\newcommand{\tldk}{\tldp{k}}

\newcommand{\prefsig}[1]{\ensuremath{\frP_{#1}}\xspace}

\newcommand{\prefsigg}{\prefsig{\Gs}}

\newcommand{\infsig}[1]{\ensuremath{\frI_{#1}}\xspace}

\newcommand{\infsigc}{\infsig{\Cs}}

\newcommand{\infsigg}{\infsig{\Gs}}

\newcommand{\infsiggp}{\infsig{\Gs^+}}

\newcommand{\infix}[3]{\ensuremath{#1(#2,#3)}\xspace}
\newcommand{\suffix}[2]{\infix{#1}{#2}{|#1|+1}}
\newcommand{\prefix}[2]{\infix{#1}{0}{#2}}
\newcommand{\wpos}[2]{\ensuremath{#1[#2]}\xspace}

\newcommand{\poschar}{\textup{\sffamily{Pos}}}

\newcommand{\pos}[1]{\ensuremath{\poschar(#1)\xspace}}

\usepackage{amsmath,amsthm,amssymb,xspace,pdfsync,verbatim}

\usepackage[T1]{fontenc}
\usepackage{stmaryrd}
\usepackage{mathrsfs}
\usepackage{bbm}
\usepackage{csquotes}

\usepackage{tikz}
\usetikzlibrary{decorations.text,arrows,decorations.pathreplacing,shapes,fit,shadows,calc,patterns,intersections}

\usepackage[english]{babel}

\declaretheorem[style=standard,sibling=theorem,name=Corollary]{cor}
\declaretheorem[style=standard,sibling=theorem,name=Proposition]{prop}
\declaretheorem[style=standard,sibling=theorem,name=Example]{exa}
\declaretheorem[style=standard,sibling=theorem,name=Lemma]{lem}
\declaretheorem[style=standard,sibling=theorem,name=Fact]{fct}
\declaretheorem[style=standard,sibling=theorem,name=Remark]{rem}

\title{All about unambiguous polynomial closure}

\ThCSauthor[bordeaux]{Thomas~Place}{tplace@labri.fr}[https://orcid.org/0009-0000-2840-9586]
\ThCSauthor[bordeaux]{Marc~Zeitoun}{mz@labri.fr}[https://orcid.org/0000-0003-4101-8437]

\ThCSaffil[bordeaux]{Univ. Bordeaux, CNRS,  Bordeaux INP, LaBRI, UMR 5800, F-33400, Talence, France}

\ThCSyear{2023}
\ThCSarticlenum{11}
\ThCSreceived{Aug 17, 2022}
\ThCSrevised{Oct 30, 2022}
\ThCSaccepted{Nov 6, 2023}
\ThCSpublished{Dec 23, 2023}
\ThCSkeywords{Words, regular languages, concatenation hierarchies,
  first-order logic, quantifier alternation, membership, separation}
\ThCSdoi{10.46298/theoretics.23.11}

\ThCSshortnames{T.\ Place, M.\ Zeitoun}
\ThCSshorttitle{All about unambiguous polynomial closure}

\ThCSthanks{Work supported by the DeLTA project (ANR-16-CE40-0007), and, for the first author, by the Institut Universitaire de France.}

\begin{document}
\maketitle

\begin{abstract}
  We study a standard operator on classes of languages: unambiguous polynomial closure. We prove that for every class \Cs of regular languages satisfying mild properties, the membership problem for its unambiguous polynomial closure $\upol{\Cs}$ reduces to the same problem for~\Cs. We also show that unambiguous polynomial closure coincides with alternating left and right deterministic closure. Moreover, we prove that if additionally~\Cs is finite, the separation and covering problems are decidable for~\upol{\Cs}. Finally, we present an overview of the generic logical characterizations of the classes built using unambiguous polynomial closure.
\end{abstract}

\section{Introduction}
\label{sec:intro}
Regularity is arguably one of the most robust notions in computer science: indeed, regular languages can be equivalently defined using deterministic or nondeterministic automata, monoids, regular expressions or monadic second-order logic. The motivation and context of this paper is the investigation of \emph{subclasses} of regular languages. Such classes arise naturally when weakening or restricting these formalisms. This active research track started in the~1960s with an emblematic example: a series of results by Schützenberger~\cite{schutzsf}, McNaughton, Papert~\cite{mnpfosf} and Kamp~\cite{kltl} established several equivalent characterizations of the class of languages defined by a \emph{first-order} sentence (instead of a monadic second-order one). At that time, the main question that emerged was the \emph{membership problem}, which simply asks whether the class under study is recursive: in order to solve it, one has to design an algorithm testing membership of an input language in the class. For instance, the above results provide a membership algorithm for first-order definable~languages.

Of course, the design of a membership algorithm for a class~\Cs heavily depends on~\Cs. This means that for each class, one has to design a different algorithm. However,
most of the interesting classes of regular languages are built using a restricted set of operators: given a class \Cs, one can consider its Boolean closure \bool\Cs, its polynomial closure \pol\Cs, and deterministic variants thereof, which usually yield more elaborate classes than \Cs. It is therefore desirable to investigate the operators~themselves rather than individual classes.

The \emph{polynomial closure} \pol\Cs of a class \Cs is its closure under finite union and marked concatenation (a \emph{marked concatenation} of $K$ and $L$ is a language of the form $KaL$, where $a$ is some letter). Together with the Boolean closure, it is used to define concatenation hierarchies, which consist in a sequence of classes of languages indexed by integers and half integers. Starting from a given class (\emph{level 0} in the hierarchy), level~$n+\frac12$ is the polynomial closure of level~$n$, and level $n+1$ is the Boolean closure of level~$n+\frac12$. The importance of these hierarchies stems from the fact that they are the combinatorial counterpart of quantifier alternation hierarchies in logic, which count the number of $\forall/\exists$ alternations needed to define a language~\cite{ThomEqu,PZ:generic_csr_tocs:18}.

As explained above, the main question when investigating a class of languages is the \emph{membership} problem: can we decide whether an input regular language belongs to the class? Despite decades of research on concatenation hierarchies, we know little about~them. The state of the art is that when level~0 is either \emph{finite} or contains only \emph{group languages} (and satisfies mild properties), membership is decidable for levels 1/2, one, 3/2, and 5/2~\cite{pzbpol,pseps3j,pzqalt,pzconcagroup} (on the other hand, the problem remains open for the second full level). These results encompass those that were obtained previously~\cite{arfi91,arfi87,simonthm,pwdelta2} and even go beyond by investigating the \emph{separation problem}, a generalization of membership. This problem for a class~\Cs takes \emph{two arbitrary regular} languages as input (unlike membership, which takes a single one). It asks whether there exists a third language from~\Cs, containing the first language and disjoint from the second. Membership is the special case of separation when the input consists of a language and its complement (as the only possible separator is the first input language). Although more difficult than membership, separation is also more rewarding. This is witnessed by a transfer theorem~\cite{PZ:generic_csr_tocs:18}: membership for \pol\Cs reduces to separation~for~\Cs. The above results on membership come from this theorem and the fact that separation is decidable for \pol\Cs, \bpol{\Cs} (\emph{i.e.}, \bool{\pol\Cs}) and \pol{\bpol{\Cs}} when \Cs is either \emph{finite} or contains only \emph{group languages}. See~\cite{jep-dd45,PZ:generic_csr_tocs:18} for detailed surveys on concatenation hierarchies.

\subparagraph{Unambiguous closure.} A natural way to weaken the polynomial closure operator is to specify \emph{semantic conditions} restricting the situations in which using marked concatenation or union is allowed. In the paper, we investigate a prominent operator that can be defined in this way: \emph{unambiguous polynomial closure}. A marked concatenation $KaL$ is \emph{unambiguous} if every word $w$ belonging to $KaL$ has a \emph{unique} factorization $w = u a v$ with $u \in K$ and $v \in L$. The \emph{unambiguous closure} of a class \Cs, denoted by \upol{\Cs}, is the least class containing \Cs that is closed under \emph{disjoint} union and \emph{unambiguous} marked concatenation. Observe that it is not immediate from the definition that \upol{\Cs} has robust properties (such as closure under Boolean operations) even when \Cs does. A prominent example of a class built using unambiguous concatenation is that of \emph{unambiguous languages}~\cite{schul} (\ul). It is the unambiguous polynomial closure of the class \bpol{\stzer} (\emph{i.e.}, \upol{\bpol{\stzer}}) where \stzer is the trivial class $\stzer = \{\emptyset,A^*\}$\footnote{The notation \stzer comes from the fact that it is level 0 in the \textbf{S}traubing-\textbf{T}hérien concatenation hierarchy.}. The robustness of \ul makes it one of the most investigated classes: it enjoys a number of equivalent definitions (we refer the reader to~\cite{Tesson02diamondsare,small_fragments08} for an overview). In particular, this class has several logical characterizations. First, \ul corresponds to the intermediary level \dewd in the quantifier alternation hierarchy of first-order logic equipped with the linear order predicate~\cite{pwdelta2}. This level consists of all languages that can be simultaneously defined by a sentence of \siwd (\emph{i.e.}, whose prenex normal form has a quantifier prefix of the form $\exists^* \forall^*$) and a sentence of \piwd (\emph{i.e.}, whose prenex normal form has a quantifier prefix of the form $\forall^* \exists^*$). It is also known that \ul corresponds to the two-variable fragment \fodw of first-order logic~\cite{twfo2}. Finally, \ul is also characterized~\cite{evwutl} by the variant \fpfm of \emph{unary temporal logic}, where $\textup{F}$ stands for ``sometimes in the future'' and $\textup{P}$ stands for ``sometimes in the past''.

Historically, the operator \upolo was first investigated by Pin~\cite{Pin80} who describes it in algebraic terms. Note however that~\cite{Pin80} starts from an alternative definition that assumes closure under Boolean operations already. It was later shown by Pin, Straubing and Thérien~\cite{Pinambigu} that the operator $\Cs \mapsto \upol{\Cs}$ preserves closure under Boolean operations (provided that the input class \Cs satisfies mild hypotheses). Let us point out that both the formulation and the proof of these results rely on elaborate mathematical tools (categories, bilateral kernel, relational morphisms,\ldots) as well as on results by Schützenberger~\cite{schul} and Rhodes~\cite{Rhodes1967AHT}, used as black boxes. Unambiguous polynomial closure also appears in concatenation hierarchies: assuming suitable hypotheses on \Cs, Pin and Weil~\cite{pwdelta2} proved that $\upol\Cs$ is the \emph{intermediate level} $\pol\Cs\cap\copol\Cs$ where \copol\Cs is the class consisting of all complements of languages in~\pol\Cs. Finally, a reduction from \upol\Cs-membership to \Cs-membership was obtained in~\cite{AlmeidaBKK15}. This proof is indirect: it relies on the nontrivial equality $\upol\Cs=\capol\Cs$, which itself depends on the algebraic characterizations of \upol\Cs and \pol{\Cs} obtained in~\cite{Pinambigu,pwdelta2,pin-branco-polc}.

\subparagraph{Contributions.} Unambiguous polynomial closure was not yet investigated with respect to separation, aside from the particular case of unambiguous languages~\cite{pvzptsep}. This is the starting point of this paper: we look for generic separation results applying to \upol\Cs, similar to the ones obtained for \pol\Cs in~\cite{pseps3j}. We prove that when \Cs is a finite Boolean algebra closed under quotients, separation is decidable for \upol{\Cs}. However, as is usual with separation, we also obtain several extra results as a byproduct of our work, improving our understanding of the $\mathit{UPol}$ operator:

\begin{itemize}
  \item We had to rethink the way membership is classically handled for \upolo in order to lift the techniques to separation. This yields self-contained, direct and elementary proofs that membership for \upol\Cs reduces to membership for \Cs provided that \Cs is closed under Boolean operations and quotients. The proof is new and independent from the one of~\cite{AlmeidaBKK15}. We strengthen the result of~\cite{AlmeidaBKK15} as we require fewer hypotheses on the class \Cs. Our proof also precisely pinpoints why this result holds for \upol\Cs but not \pol\Cs: the languages from~\Cs needed to construct a~\upol\Cs expression for a language $L$ are all recognized by any recognizer of~$L$.

  \item We prove that when \Cs is closed under Boolean operations and quotients, then so is \upol{\Cs}. Again, this strengthens earlier results~\cite{Pinambigu}, which involved stronger hypotheses on the input class \Cs.

  \item Using our results on~\pol\Cs~\cite{PZ:generic_csr_tocs:18}, we obtain a new proof of the above mentioned theorem of~\cite{pwdelta2} and generalize it to more classes \Cs. We prove that $\upol\Cs=\pol\Cs\cap \copol\Cs$ when \Cs is closed under Boolean operations and quotients (whereas the result of~\cite{pwdelta2} only applies to classes that are additionally closed under inverse morphic images).

  \item Finally, we obtain a previously unknown characterization of \upol\Cs in terms of alternating left and right deterministic concatenations, which are restricted forms of unambiguous concatenation.  A marked concatenation $KaL$ is \emph{left (resp.~right) deterministic} when $KaA^*\cap K=\emptyset$ (resp.\ $A^*aL\cap L=\emptyset$). We prove that \upol\Cs coincides with \adet\Cs, the closure of \Cs under disjoint union and left and right deterministic concatenation\footnote{The letter ``A'' in \adet\Cs stands for ``alternating'' (left and right) deterministic polynomial closure.}.
\end{itemize}

Moreover, we investigate the logical characterizations of unambiguous polynomial closure. As mentioned above, it is well known that \mbox{$\upol{\bpol{\stzer}} = \dewd = \fodw = \fpfm$}. Such correspondences also hold for variants of \dewd and \fodw equipped with additional predicates such as the successor~\cite{twfo2} ``$+1$'' and the modular predicates~\cite{DartoisP13,KufleitnerW15}. We use our results to generalize these connections. For each class \Cs, we define two objects:
\begin{itemize}
  \item A set of predicates \prefsig\Cs that can be used in the sentences of first-order logic. Each language~$L$ in \Cs gives rise to a predicate $P_L(x)$, which selects all positions $x$ in a word $w$ such that the prefix of $w$ up to position $x$ (excluded) belongs to $L$.
  \item Two variants of unary temporal logic, which we denote by \tlc{\Cs} and \tlxc{\Cs}.
\end{itemize}
We consider classes of group languages. \emph{Group languages} are those  recognized by a finite group, or equivalently by a permutation automaton~\cite{permauto} (\emph{i.e.}, a complete, deterministic \emph{and} co-deterministic automaton). Our results apply to classes of group languages \Gs closed under Boolean operations and quotients, and their \wsuit extensions $\Gs^+$ (roughly, $\Gs^+$ is the smallest Boolean algebra containing~\Gs and the singleton $\{\veps\}$). In such cases, we prove that,
\[
  \begin{array}{ccccccc}
    \upol{\bpol{\Gs}} & = & \dec{2}(<,\prefsigg) & = & \fod(<,\prefsigg) & =&  \tlc{\Gs}. \\
    \upol{\bpol{\Gs^+}} & = & \dec{2}(<,+1,\prefsigg) & = & \fod(<,+1,\prefsigg) & =&  \tlxc{\Gs}.
  \end{array}
\]
This generalizes all aforementioned results (each of them corresponds to a particular class of group languages \Gs). Thus, we obtain a generic proof of these results.

\subparagraph{Organization of the paper.} Section~\ref{sec:prelims} sets up the notation and the terminology. In Section~\ref{sec:frag}, we recall the definition of first-order logic over words and introduce the fragments that we shall consider. Section~\ref{sec:opers} is devoted to presenting the operators that we investigate in the paper. In Section~\ref{sec:cano}, we introduce notions that we shall need to present the algebraic characterizations of \polo and \upolo. Section~\ref{sec:caracpol} presents the algebraic characterization of \pol{\Cs} (which is taken from~\cite{PZ:generic_csr_tocs:18}). We use it to obtain a characterization of the classes \capol{\Cs}. Then, in Section~\ref{sec:caracupol}, we turn to the classes \upol{\Cs} and present their generic characterization. We investigate separation and covering for the classes \upol{\Cs} when~\Cs is finite in Section~\ref{sec:upolcov}. Finally, Sections~\ref{sec:logic} and~\ref{sec:logcar} are devoted to the logical characterizations of \upolo. In the former, we define our generalized notion of ``unary temporal logic'' and prove the correspondence with two-variable first-order logic. Finally, we establish the connection with unambiguous polynomial closure in Section~\ref{sec:logcar}.

\smallskip

This paper is the journal version of~\cite{pzupolicalp}. The presentation and the proof arguments have been reworked. Moreover, the logical part of the paper is new (the conference version only considered the language theoretic point of view).

\section{Preliminaries}
\label{sec:prelims}
We introduce the terminology used in the paper. In particular, we present the membership, separation and covering problems, and key mathematical tools designed to handle them.

\subsection{Classes of regular languages}

For the whole paper, we fix a finite alphabet $A$. As usual, $A^*$ denotes the set of all \emph{finite} words over $A$, including the empty word \veps. Moreover, we let $A^+ = A^* \setminus \{\veps\}$. Given a word $w \in A^*$, we write $|w| \in \nat$ for its length (the number of letters in $w$). For example, $|\veps| = 0$ and $|abcbaa| = 6$. Moreover, for every $a \in A$ and $w \in A^*$, we write $|w|_a \in \nat$ for the number of copies of the letter ``$a$'' in $w$. For example, $|abcbaa|_a = 3$, $|abcbaa|_b = 2$ and $|abcbaa|_c = 1$. Given $u,v \in A^*$, we write $u v$ for their~concatenation.

A \emph{language} is a subset of $A^*$. It is standard to extend concatenation to languages: given $K,L \subseteq A^*$, we write~$KL = \{uv \mid u \in K \text{ and } v \in L\}$.

\subparagraph{Classes.} A class of languages \Cs is a set of languages. Such a class \Cs is a \emph{lattice} when $\emptyset\in\Cs$, $A^* \in \Cs$ and \Cs is closed under both union and intersection: for every $K,L \in \Cs$, we have $K \cup L \in \Cs$ and $K \cap L \in \Cs$. Moreover, a \emph{Boolean algebra} is a lattice \Cs which is additionally closed under complement: for every $L \in \Cs$, we have $A^* \setminus L \in \Cs$. Finally, a class~\Cs is \emph{quotient-closed} if for every $L \in \Cs$ and $u \in A^*$, the following properties~hold:
\[
  u^{-1}L \stackrel{\text{def}}{=}\{w\in A^*\mid uw\in L\} \text{\quad and\quad} Lu^{-1} \stackrel{\text{def}}{=}\{w\in A^*\mid wu\in L\}\text{\quad both belong to \Cs}.
\]
We call a \pvari (resp. \vari) a quotient-closed lattice (resp. Boolean algebra) containing only \emph{regular languages}. Regular languages are those which can be equivalently defined by nondeterministic finite automata, finite monoids or monadic second-order logic. We work with the definition by monoids, which we recall below.

In the paper, we write \at for the class of \emph{alphabet testable languages}: the Boolean combinations of languages $A^*aA^*$ for $a\in A$. It is straightforward to verify that \at is a \vari. It will be important for providing examples.

\subparagraph{Finite monoids and morphisms.} A \emph{semigroup} is a set $S$  endowed with an associative multiplication $(s,t)\mapsto s\cdot t$ (also denoted by~$st$). A \emph{monoid} is a semigroup $M$ whose multiplication has an identity element $1_M$, \emph{i.e.}, such that ${1_M}\cdot s=s\cdot {1_M}=s$ for every~$s \in M$.

An \emph{idempotent} of a semigroup $S$ is an element $e \in S$ such that $ee = e$. We write $E(S) \subseteq S$ for the set of all idempotents in $S$. It is folklore that for every \emph{finite} semigroup~$S$, there exists a natural number $\omega(S)$ (denoted by $\omega$ when $S$ is understood) such that for every $s \in S$, the element $s^\omega$ is an idempotent.

We also consider ordered semigroups. An \emph{ordered semigroup} is a pair $(S,\leq)$ where $S$ is a semigroup and $\leq$ is a partial order defined on $S$, which is compatible with multiplication: for every $s,s',t,t' \in S$ such that $s \leq t$ and $s' \leq t'$, we have $ss'\leq tt'$. In particular, we say that a subset $F \subseteq S$ is an upper set (for $\leq$) to indicate that it is upward closed for the ordering $\leq$: for every $s,t\in M$, if $s \in F$ and $s \leq t$, then $t \in F$. In particular, for every $s \in S$, we write $\uclos s \subseteq S$ for the least upper set containing $s$. That is, $\uclos s = \{t \in S \mid s \leq t\}$. Finally, an \emph{ordered monoid} is an ordered semigroup $(M,\leq)$ such that $M$ is a monoid.

We shall use the following convention throughout the paper: we view every \emph{unordered} semigroup $S$ as the ordered semigroup $(S,=)$ whose ordering is \emph{equality}. Observe that in this case, \emph{every} subset of $S$ is an upper set. Thanks to this convention, every definition involving ordered semigroups that we present will also make sense for unordered ones.

Finally, our proofs make use of the Green relations~\cite{green}, which are defined on monoids. We briefly recall them. Given a monoid $M$ and $s,t \in M$,
\[
  \begin{array}{ll@{\ }l}
    s \Jord t & \text{when} & \text{there exist $x,y\in M$ such that $s=xty$}, \\
    s \Lord t & \text{when} & \text{there exists $x \in M$ such that $s = xt$}, \\
    s \Rord t & \text{when} & \text{there exists $y \in M$ such that $s = ty$}. \\
  \end{array}
\]
Clearly, \Jord, \Lord and \Rord are preorders (\emph{i.e.}, they are reflexive and transitive). We write \Jords, \Lords and \Rords for their strict variants (for example, $s \Jords t$ when $s \Jord t$ but $t \not\Jord s$). Finally, we write \Jrel, \Lrel and \Rrel for the corresponding equivalence relations (for example, $s \Jrel t$ when $s \Jord t$ and $t \Jord s$). There are many technical results about Green relations. We use the following standard lemma (see \emph{e.g.} \cite{pingoodref} for a proof), which applies to \emph{finite} monoids.

\begin{lem} \label{lem:jlr}
  Consider a finite monoid $M$ and $s,t \in M$ such that $s \Jrel t$. Then, $s \Rord t$ implies $s \Rrel t$. Symmetrically, $s \Lord t$ implies $s \Lrel t$.
\end{lem}

\subparagraph{Regular languages and syntactic morphisms.}
Clearly, $A^*$ is a monoid whose multiplication is concatenation (the identity element is the empty word \veps). We shall consider monoid morphisms $\alpha: A^* \to (M,\leq)$ where $(M,\leq)$ is an arbitrary ordered monoid. That is, $\alpha:A^*\to M$ is a map satisfying $\alpha(\varepsilon)=1_M$ and $\alpha(uv)=\alpha(u)\alpha(v)$ for all $u,v\in A^*$. Given such a morphism and some language $L \subseteq A^*$, we say that $L$ is \emph{recognized} by $\alpha$ when there exists an accepting set $F \subseteq M$ which is an upper set for $\leq$ and such that $L= \alpha\inv(F)$. Let us emphasize that the definition depends on the ordering $\leq$ since accepting sets must be upper sets.

\begin{rem}
  Recall that by convention, we view an unordered monoid $M$ as the ordered monoid $(M,=)$. Hence, the definition also makes sense when $\alpha: A^*\to M$ is a morphism into an unordered monoid. Since every subset of $M$ is an upper set for $=$, a language $L \subseteq A^*$ is recognized by $\alpha$ if and only if there exists an arbitrary set $F \subseteq M$ such that $L = \alpha\inv(F)$.
\end{rem}

It is standard and well known that regular languages are those which can be recognized by a morphism into a \emph{finite} monoid. Moreover, every language $L$ is recognized by a canonical morphism. Let us briefly recall its definition. One may associate to $L$ a preorder relation $\preceq_L$ over $L$: the \emph{syntactic precongruence of~$L$}. Given $u,v \in A^*$, we let,
\[
  \text{$u \preceq_L v$ if and only if $xuy \in L \Rightarrow xvy \in L$ for every $x,y \in A^*$}.
\]
As the name suggests, it is known and simple to verify that ``$\preceq_L$'' is a precongruence on $A^*$: it is reflexive and transitive, and for every $u,u',v,v' \in A^*$ such that $u \preceq_L v$ and $u' \preceq_L v'$, we have $uu' \preceq_L vv'$. Additionally, we write $\equiv_L$ for the equivalence generated by this preorder. For $u,v \in A^*$, $u \equiv_L v$ if and only if $u \preceq_L v$ and $v \preceq_L u$ (\emph{i.e.}, $xuy \in L \Leftrightarrow xvy \in L$ for every $x,y \in A^*$). The equivalence $\equiv_L$ is a congruence called the \emph{syntactic congruence of~$L$}. Thus, the set of equivalence classes $M_L = {A^*}/{\equiv_L}$ is a monoid and the preorder $\preceq_L$ defines an ordering $\leq_L$ on $M_L$ such that $(M_L,\leq_L)$ is an ordered monoid. It is called the \emph{syntactic ordered monoid of $L$}. Finally, the map $\alpha_L: A^* \to (M_L,\leq_L)$ sending every word to its equivalence class is a morphism recognizing~$L$, called the \emph{syntactic morphism of~$L$}. Another characterization of the regular languages is that $L$ is regular if and only if $M_L$ is finite (\emph{i.e.}, $\equiv_L$ has finite index): this is the Myhill-Nerode theorem. In this case, one can compute the syntactic morphism $\alpha_L: A^* \to (M_L,\leq_L)$ from any representation of $L$ (such as a finite automaton or an arbitrary monoid morphism).

\subsection{Decision problems} We consider three decision problems. They all depend on an arbitrary class of languages \Cs: they serve as mathematical tools for analyzing \Cs. Intuitively, obtaining an algorithm for one of these three problems requires a solid understanding of~\Cs.

The \emph{\Cs-membership problem} is the simplest one. It takes as input a single regular language~$L$ and asks whether $L\in \Cs$. The second problem, \emph{\Cs-separation}, is more general. Given three languages $K,L_1,L_2$, we say that $K$ \emph{separates} $L_1$ from $L_2$ if we have $L_1 \subseteq K$ and $L_2 \cap K = \emptyset$. Given a class of languages \Cs, we say that $L_1$ is \emph{\Cs-separable} from $L_2$ if some language in \Cs separates $L_1$ from $L_2$. Observe that when \Cs is not closed under complement, the definition is not symmetrical: it is possible for $L_1$ to be \Cs-separable from $L_2$ while $L_2$ is not \Cs-separable from~$L_1$. The separation problem associated to a given class \Cs takes two regular languages $L_1$ and $L_2$ as input and asks whether $L_1$ is \Cs-separable from $L_2$.

\begin{rem}
  The \Cs-separation problem generalizes the \Cs-membership problem. Indeed, a regular language belongs to $\Cs$ if and only if it is \Cs-separable from its complement, which is also regular.
\end{rem}

In the paper, we do not consider separation directly. Instead, we work with a third, even more general problem: \Cs-covering. It was introduced in~\cite{pzcovering2} and takes as input a single regular language $L_1$ and a \emph{finite set of regular languages} $\Lb_2$. It asks whether there exists a ``\Cs-cover of $L_1$ which is separating for $\Lb_2$''.

Given a language $L$, a \emph{cover of $L$} is a \emph{finite} set of languages \Kb such that $L \subseteq \bigcup_{K \in \Kb} K$. A cover \Kb is a \emph{\Cs-cover} if all languages $K \in \Kb$ belong to \Cs. Moreover, given two finite sets of languages \Kb and \Lb, we say that \Kb is \emph{separating} for \Lb if for every $K\in\Kb$, there exists $L\in\Lb$ such that $K \cap L = \emptyset$. Finally, given a language $L_1$ and a finite set of languages~$\Lb_2$, we say that the pair $(L_1,\Lb_2)$ is \emph{\Cs-coverable} if there exists a \Cs-cover of $L_1$ which is separating~for~$\Lb_2$.

The \Cs-covering problem is now defined as follows. Given as input a regular language $L_1$ and a finite set of regular languages $\Lb_2$, it asks whether the pair $(L_1,\Lb_2)$ is \Cs-coverable. It is straightforward to prove that covering generalizes separation if the class~\Cs is a lattice, as stated in the following lemma (see~\cite[Theorem~3.5]{pzcovering2} for the proof, which is easy).

\begin{lem} \label{lem:covsep}
  Let \Cs be a lattice and $L_1,L_2$ be two languages. Then $L_1$ is \Cs-separable from $L_2$ if and only if $(L_1,\{L_2\})$ is \Cs-coverable.
\end{lem}

\subsection{Group languages}\label{sec:group-languages}

We now present an important special kind of class. As we explained in the introduction, all classes investigated in the paper are built from basic ones using generic operators. Here, we introduce the basic classes used in such constructions: the \emph{classes of group languages} and their \emph{\wsuit extensions}.

\begin{definition}
A \emph{group} is a monoid $G$ such that every element $g \in G$ has an inverse $g\inv \in G$, \emph{i.e.}, $gg\inv = g\inv g = 1_G$. A language $L$ is a \emph{group language} if it is recognized by a morphism $\alpha: A^* \to G$ into a \emph{finite group $G$} (\emph{i.e.}, there exists $F \subseteq G$ such that $L = \alpha\inv(F)$). We write \grp for the class of \emph{all} group languages. One can verify that \grp is a \vari.
\end{definition}
\begin{rem}
  No language theoretic definition of \grp is known. There is however a definition based on automata: the group languages are those which can be recognized by a permutation automaton (\emph{i.e.}, which is simultaneously deterministic, co-deterministic and complete).
\end{rem}

A \emph{class of group languages} is a class consisting of group languages only, \emph{i.e.}, a subclass of \grp. The results of this paper apply to arbitrary \emph{\varis of group languages}.

\begin{rem}\label{rem:orderedgroups}
  Let us explain why we do not explicitly mention \emph{\pvaris} of group languages in the paper. This is because they are all closed under complement (hence, they are actually \emph{\varis}). This follows from the simple  fact that the only ordering compatible with multiplication in a finite group is the equality relation.
\end{rem}

While our results apply in a generic way to all \varis of group languages, there are \emph{four} main classes of this kind that we shall use for providing examples. One of them is \grp itself. Let us present the other three. First, we write $\stzer = \{\emptyset,A^*\}$, which is clearly a \vari of group languages (the name comes from the fact that this class is the base level of the \emph{Straubing-Th\'erien} hierarchy). While trivial, we shall see that this class has important applications. Moreover, we look at the class \md of \emph{modulo languages}. For every $q,r \in \nat$ with $r < q$, we write $L_{q,r} = \{w \in A^* \mid |w| \equiv r \bmod q\}$. The class \md contains all \emph{finite unions} of languages~$L_{q,r}$. One can verify that \md is a \vari of group languages. Finally, we shall consider~the class \abg of \emph{alphabet modulo testable languages}. For all $q,r \in \nat$ with $r < q$ and all $a \in A$, let $L^a_{q,r} = \{w \in A^* \mid |w|_a \equiv r \bmod q\}$. We define \abg as the least class consisting of all languages $L^a_{q,r}$ and closed under union and intersection. It is again straightforward to verify that \abg is a \vari of group languages.

We do not investigate classes of group languages themselves in the paper: we only use them as input classes for our operators.  In particular, we shall use \stzer, \md, \abg and \grp in order to illustrate our results. In this context, it will be important that \emph{separation} is decidable for these four classes. The techniques involved in proving this are independent of what we do in the paper. Actually, this can be difficult. On one hand, the decidability of \stzer-separation is immediate (two languages are \stzer-separable if and only if one of them is empty). On the other hand, the decidability of \grp-separation is equivalent to a difficult algebraic question, which remained open for several years before it was solved by Ash~\cite{Ash91}. Recent automata-based proofs that separation is decidable for \md, \abg and \grp are available in~\cite{groups}.

\subparagraph{\Wsuit extensions.} It can be verified from the definition that $\{\veps\}$ and $A^+$ are \emph{not} group languages. This motivates the next definition: for a class \Cs, the \emph{\wsuit extension of \Cs}, denoted by $\Cs^+$, consists of all languages of the form $L \cap A^+$ and $L \cup \{\veps\}$ where $L \in \Cs$ (in particular, $\Cs\subseteq\Cs^{+}$). The next lemma follows immediately from the definition.

\begin{lem} \label{lem:wsuit}
  Let \Cs be a \vari. Then, $\Cs^+$ is a \vari containing $\{\veps\}$ and $A^+$.
\end{lem}

While the definition of \wsuit extensions makes sense for arbitrary classes, it will only be useful when applied to a \vari of group languages. Indeed, the \wsuit extensions $\Gs^+$, where \Gs is a \vari of group languages, are also important input classes for our operators. In this case,   all group languages in $\Gs^+$ already belong to~$\Gs$.

\begin{lem} \label{lem:gplus1}
  Let \Gs be a \vari of group languages. For every group language $L \in  \Gs^+$, we have $L \in \Gs$.
\end{lem}

\begin{proof}
  By definition of $\Gs^+$, there exists a group language $K \in \Gs$ such that either $L = \{\veps\} \cup K$ or $L = A^+ \cap K$. In particular, this implies that for every $w \in A^+$, we have $w \in K \Leftrightarrow w \in L$. We prove that the fact that $K$ and $L$ are both group languages implies $K = L$, which yields $L \in \Gs$, as desired. Since we already know that $w\in K\Leftrightarrow w\in L$ for every $w \in A^+$, it suffices to prove that $\veps \in K \Leftrightarrow \veps \in L$. Since $K$ and $L$ are group languages, there exist two morphisms $\alpha_K: A^* \to G_K$ and $\alpha_L: A^*\to G_L$ into finite groups recognizing $K$ and $L$ respectively. Let $p = \omega(G_K) \times \omega(G_L)$ and consider an arbitrary letter $a \in A$. By definition of $p$, we have $\alpha_K(a^p) = 1_{G_K}$ and $\alpha_L(a^p) = 1_{G_L}$. Since $\alpha_K$ and $\alpha_L$ recognize both $K$ and $L$, we have $\veps \in K \Leftrightarrow a^p \in K$ and $a^p \in L \Leftrightarrow \veps \in L$. Finally, since $a^p \in A^+$, we know that $a^p \in K \Leftrightarrow a^p \in L$ by definition. Altogether, it follows that $\veps \in K \Leftrightarrow \veps \in L$, as desired.
\end{proof}

\subsection{\Cs-morphisms}

We now present a central mathematical tool. Consider an arbitrary \pvari \Cs. A \emph{\Cs-morphism} is a \emph{surjective} morphism $\eta: A^*\to (N,\leq)$ into a finite ordered monoid $(N,\leq)$ such that every language recognized by $\eta$ belongs to \Cs. We complete the definition with a key remark when \Cs is a \vari (\emph{i.e.}, \Cs is additionally closed under complement): in this case, it suffices to consider \emph{unordered} monoids (recall that by convention, we view them as monoids ordered by equality).

\begin{lem} \label{lem:cmorphbool}
  Let \Cs be a \vari and let $\eta: A^* \to (N,\leq)$ be a morphism into a finite ordered monoid. The two following conditions are equivalent:
  \begin{enumerate}
    \item $\eta: A^* \to (N,\leq)$ is a \Cs-morphism.
    \item $\eta: A^* \to (N,=)$ is a \Cs-morphism.
  \end{enumerate}
\end{lem}

\begin{proof}
  The implication $2) \Rightarrow 1)$ is trivial since all languages recognized by $\eta: A^* \to (N,\leq)$ are also recognized by $\eta: A^* \to (N,=)$. We prove that $1) \Rightarrow 2)$. Assume that $\eta: A^* \to (N,\leq)$ is a \Cs-morphism. We have to prove that  $\eta: A^* \to (N,=)$ is a \Cs-morphism, that is, that for all $F\subseteq N$, we have $\eta\inv(F)\in\Cs$. Since $\eta\inv(F)=\bigcup_{s\in F}\eta\inv(s)$ and \Cs is closed under union, it suffices to prove that $\eta\inv(s)\in\Cs$ for all $s\in N$. Therefore, we fix an arbitrary element $s\in N$. Let $T\subseteq M$ be the set of all elements $t\in M$ such that $s\leq t$ and $s \neq t$. One can verify that we have $\{s\} = (\uclos s) \setminus \left(\bigcup_{t \in T} \uclos t\right)$. Consequently, we obtain $\eta\inv(s) = \eta\inv(\uclos s) \setminus  \left(\bigcup_{t \in T} \eta\inv(\uclos t)\right)$. By hypothesis, this implies that $\eta\inv(s)$ is a Boolean combination of languages in \Cs. We conclude that $\eta\inv(s) \in \Cs$, since~\Cs is a Boolean~algebra.
\end{proof}

\begin{rem} \label{rem:cmbool}
  Another formulation of Lemma~\ref{lem:cmorphbool} is that when \Cs is a \vari, whether a morphism $\eta: A^* \to (N,\leq)$ is a \Cs-morphism does not depend on the ordering $\leq$ on $N$. This is important, as most classes that we consider in the paper are indeed \varis.
\end{rem}

While simple, the notion of \Cs-morphism is a central tool in the paper. First, it is connected to the membership problem via the following simple, yet crucial proposition.

\begin{prop} \label{prop:synmemb}
  Let \Cs be a \pvari. A regular language~$L$ belongs to~\Cs if and only if its syntactic morphism $\alpha_L: A^* \to (M_L,\leq_L)$ is a \Cs-morphism.
\end{prop}

\begin{proof}
  The ``if'' implication is immediate since $L$ is recognized by its syntactic morphism. We prove the converse direction: assuming that $L \in \Cs$, we prove that every language recognized by $\alpha_L: A^* \to (M_L,\leq_L)$ belongs to \Cs (recall that syntactic morphisms are surjective by definition). Hence, we fix an upper set $F \subseteq M_L$ and prove that $\alpha_L\inv(F) \in \Cs$. For every $s \in M_L$, we fix a word $u_s \in A^*$ such that $\alpha_L(u_s) = s$. Moreover, we let $P_s \subseteq M_L^2$ be the~set,
  \[
    P_s = \big\{(q,r) \in M_L^2 \mid qsr \in \alpha_L(L)\big\}.
  \]
  We now prove that,
  \[
    \alpha_L\inv(F) = \bigcup_{s \in F} \bigcap_{(q,r) \in P_s}\left(u_q\inv L u_r\inv\right).
  \]
  Since $L \in \Cs$ and \Cs is a \pvari, this yields $\alpha_L\inv(F) \in \Cs$, as desired. We start with the left to right inclusion. Consider $w \in \alpha_L\inv(F)$ and let $s = \alpha_L(w) \in F$. We prove that $w \in u_q\inv L u_r\inv$ for every $(q,r) \in P_s$. Clearly, $\alpha_L(u_qwu_r) = qsr$ and by definition of $P_s$, we have $qsr \in \alpha_L(L)$. Since $L$ is recognized by $\alpha_L$, we get $u_qwu_r\in L$, which exactly says that $w \in u_q\inv L u_r\inv$, as desired.

  We turn to the converse inclusion. Fix $s \in F$ and consider $w \in \bigcap_{(q,r) \in P_s}\left(u_q\inv L u_r\inv\right)$. We show that $w \in \alpha_L\inv(F)$, \emph{i.e.}, $\alpha_L(w) \in F$. Since $s \in F$ and $F$ is an upper set, it suffices to prove that $s \leq_L \alpha_L(w)$. Recall that $u_s \in A^*$ is such that $\alpha_L(u_s)=s$. Therefore, by definition of the syntactic morphism $\alpha_L$, it suffices to show that $u_s \preceq_L w$, where $\preceq_L$ is the syntactic precongruence of $L$. Consider $x,y \in A^*$ such that $xu_sy \in L$. We have to prove that $xwy \in L$. Let $q = \alpha_L(x)$ and $r = \alpha_L(y)$. By hypothesis, we have $qsr = \alpha_L(xu_sy) \in \alpha_L(L)$ which yields $(q,r) \in P_s$. Hence, we get $w \in u_q\inv L u_r\inv$ by hypothesis, which yields $u_qwu_r \in L$. Since $\alpha_L(u_qwu_r) = q\alpha_L(w)r = \alpha_L(xwy)$ and $L$ is recognized by $\alpha_L$, it then follows that $xwy \in L$, as desired. This completes the proof.
\end{proof}

In view of Proposition~\ref{prop:synmemb}, getting an algorithm for \Cs-membership boils down to finding a procedure to decide whether an input morphism $\alpha: A^* \to (M,\leq)$ is a \Cs-morphism. This is how we approach the question in the paper.

\smallskip

Additionally, we shall use \Cs-morphisms as mathematical tools in proof arguments. They are convenient when manipulating arbitrary classes. We present a few properties that we shall need in this context. First, we have the following simple corollary of Proposition~\ref{prop:synmemb}.

\begin{prop}\label{prop:genocm}
  Let \Cs be a \pvari and consider finitely many languages $L_1,\dots,L_k \in \Cs$. There exists a \Cs-morphism $\eta: A^* \to (N,\leq)$ such that $L_1,\dots,L_k$ are recognized by $\eta$.
\end{prop}

\begin{proof}
  For every $i \leq k$, we let $\alpha_i: A^* \to (M_i,\leq_i)$ be the syntactic morphism of $L_i$. We know from Proposition~\ref{prop:synmemb} that $\alpha_i$ is a \Cs-morphism. Consider the monoid $M = M_1 \times \cdots \times M_k$ equipped with the componentwise multiplication. Moreover, consider the ordering $\leq$  on $M$ defined by $(s_1,\dots,s_k) \leq (t_1,\dots,t_k)$ if and only if  $s_i \leq_i t_i$ for every $i \leq k$. Clearly, $(M,\leq)$ is an ordered monoid. Moreover, let $\alpha: A^* \to (M,\leq)$ be the morphism defined by $\alpha(w) = (\alpha_1(w_1),\dots,\alpha_k(w))$ for every $w \in A^*$. One can verify from the definition that all languages recognized by $\alpha$ are finite intersections of languages recognized by $\alpha_1,\dots,\alpha_k$ (in particular, $\alpha$ recognizes each $L_i$). Hence, all languages recognized by $\alpha$ belong to \Cs. It now suffices to let $\eta: A^* \to (N,\leq)$ be the surjective morphism obtained by restricting the codomain of  $\alpha$ to its image.
\end{proof}

We now consider the particular case of \Gs-morphisms when \Gs is a \vari of group languages. In this case, the codomain of any \Gs-morphism is a group.

\begin{lem} \label{lem:gmorph}
  Let \Gs be a \vari of group languages and let $\eta: A^* \to G$ be a \Gs-morphism. Then, $G$ is a group.
\end{lem}

\begin{proof}
  Let $p = \omega(G)$. We prove that for every $g \in G$, the element $g^{p-1}$ is an inverse of $g$, \emph{i.e.}, that $g^p = 1_G$. Since $\eta$ is a \Gs-morphism and \Gs is a \vari of group languages,  $\eta\inv(1_G)\in\Gs$ is recognized by a finite group. Let $\beta: A^*\to H$ be a morphism into a finite group $H$ recognizing $\eta\inv(1_G)$. Since $\veps \in \eta\inv(1_G)$, it follows that $\beta\inv(1_H) \subseteq \eta\inv(1_G)$. Let $w \in \eta\inv(g)$ (recall that $\eta$ is surjective, since it is a $\Gs$-morphism), and let $k = \omega(H)$. Since $H$ is a group, we have $(\beta(w))^{kp}=1_H$ (the unique idempotent in $H$). Therefore, $\beta(w^{kp}) = 1_H$, which implies that $w^{kp} \in \eta\inv(1_G)$. Hence, $\eta(w^{kp}) = 1_G$ which exactly says that $g^p = 1_G$ by definition of $p$.
\end{proof}

We now look at $\Gs^+$-morphisms when \Gs is a \vari of group languages.

\begin{lem} \label{lem:gpmorph}
  Let \Gs be a \vari of group languages and let $\eta: A^* \to N$ be a $\Gs^+$-morphism. Then, the set $G = \eta(A^+)$ is a group and the map $\beta: A^* \to G$ defined by $\beta(\veps) = 1_G$ and $\beta(w) = \eta(w)$ for $w \in A^+$ is a \Gs-morphism.
\end{lem}

\begin{proof}
  We first prove that $G = \eta(A^+)$ is a group. Let $p = \omega(N)$. It suffices to show that for every $g \in G$, we have $g^p=1_G$ \emph{i.e.}, that for every $g,h \in G$, we have $g^p h = hg^p = h$. By definition of $G$, there exist $u,v \in A^+$ such that $\eta(u) = g$ and $\eta(v) = h$. By definition of $\Gs^+$, since $\eta: A^* \to N$ is a $\Gs^+$-morphism, we know that there exists a language $L_h \in \Gs$ such that either $\eta\inv(h) = \{\veps\} \cup L_h$ or $\eta\inv(h) = A^+ \cap L_h$. Since  $v \in A^+$ and $\eta(v) = h$, we have $v \in L_h$. Since $L_h \in \Gs$ is a group language, there exists a number $n \geq 1$ such that $u^{np}v \in L_h$ and $vu^{np} \in L_h$: it suffices to choose $n$ as the idempotent power of a finite group recognizing $L_h$. Since $u^{np}v,vu^{np} \in A^+$, it follows that $\eta(u^{np}v) =\eta(vu^{np}) = h$. Since $p = \omega(N)$, this exactly says that $g^p h = hg^p = h$, concluding the proof that $G$ is a group. Finally, consider the morphism $\beta: A^* \to G$ defined by $\beta(\veps) = 1_G$ and $\beta(w) = \eta(w)$ for every $w \in A^+$. Let $L$ be a language recognized by $\beta$. By definition $L$ is a group language and it is also recognized by $\eta$, which implies that $L \in \Gs^+$ by hypothesis on $\eta$. Hence, it follows from Lemma~\ref{lem:gplus1} that $L \in \Gs$, and we conclude that $\beta$ is a \Gs-morphism.
\end{proof}

Finally, we consider the special case where \Cs is a \emph{finite} \vari (\emph{i.e.}, \Cs contains finitely many languages). First, recall from Lemma~\ref{lem:cmorphbool} that being a \Cs-morphism does not depend on the ordering of the finite monoid, which we can therefore assume to be unordered. Moreover, since \Cs is finite, Proposition~\ref{prop:genocm} implies that there exists a \Cs-morphism recognizing \emph{all} languages in \Cs. The following lemma implies that it is unique (up to renaming).

\begin{lem} \label{lem:cmdiv}
  Let \Cs be a finite \vari and let $\alpha: A^* \to M$ and $\eta: A^* \to N$ be two \Cs-morphisms. If $\alpha$ recognizes all languages in \Cs, then there exists a morphism $\gamma: M \to N$ such that $\eta = \gamma \circ \alpha$.
\end{lem}

\begin{proof}
  We assume that  $\alpha$ recognizes all languages in \Cs and define $\gamma: M \to N$. For every $s \in M$, we fix a word $w_s \in \alpha\inv(s)$ (recall again that \Cs-morphisms are surjective by definition) and define $\gamma(s) = \eta(w_s)$. It remains to prove that $\gamma$ is a morphism and that $\eta = \gamma \circ \alpha$. It suffices to prove the latter: since $\alpha$ is surjective, the former is an immediate consequence. Let $v \in A^*$. We show that $\eta(v) = \gamma(\alpha(v))$. Let $s = \alpha(v)$. By definition, $\gamma(s) = \eta(w_s)$. Hence, we need to prove that $\eta(v) = \eta(w_s)$. Since $\eta$ is a \Cs-morphism, we have $\eta\inv(\eta(w_{s})) \in \Cs$. Hence, our hypothesis implies that $\eta\inv(\eta(w_{s}))$ is recognized by $\alpha$. Since it is clear that $w_s \in \eta\inv(\eta(w_{s}))$ and $\alpha(v) = \alpha(w_s) = s$, it follows that $v \in \eta\inv(\eta(w_{s}))$ which exactly says that $\eta(v) = \eta(w_s)$, completing the proof.
\end{proof}

By Lemma~\ref{lem:cmdiv}, if \Cs is a finite \vari and $\alpha: A^* \to M$ and $\eta: A^* \to N$ are two \Cs-morphisms which both recognize \emph{all} languages in \Cs, there exist two morphisms $\gamma: M \to N$ and $\beta: N \to M$ such that $\eta=\gamma \circ \alpha$ and $\alpha = \beta \circ \eta$. Since $\alpha$ and $\eta$ are surjective, it follows that $\beta \circ \gamma: M \to M$ is the identity morphism. Hence, $\beta$ and $\gamma$ are both isomorphisms which means that $\alpha$ and $\eta$ are the same object up to renaming. We call it the \emph{canonical \Cs-morphism} and denote it by $\etac: A^* \to \canc$. Let us emphasize that this object is only defined when \Cs is a \emph{finite \vari}.

\begin{exa}
  An important example of finite \vari is the class \at of alphabet testable languages defined above: the Boolean combinations of languages $A^*aA^*$ for $a \in A$. It can be verified that \canat is the monoid $2^A$ of subalphabets whose multiplication is union (up to isomorphism). Moreover, $\etaat: A^* \to 2^A$ is the morphism such that for every $w \in A^*$, $\etaat(w) \in 2^A$ is the set of letters occurring in $w$ (\emph{i.e.}, the least $B \subseteq A$ such that~$w \in B^*$).
\end{exa}

\section{Fragments of first-order logic}
\label{sec:frag}
We recall the definition of first-order logic overs words and its fragments. In the paper, we are interested in two particular fragments which we define here: two-variable first-order logic \fod and level \dec{2} in the quantifier alternation hierarchy of first-order logic.

\subparagraph{Positions.} We view a word $w \in A^*$ as a logical structure whose domain is a set of positions. More precisely, a word $w =a_1 \cdots a_{|w|} \in A^*$ for the letters $a_1,\dots,a_{|w|} \in A$, is viewed as an \emph{ordered set of $|w|+2$ positions} $\pos{w} = \{0,1,\dots,|w|,|w|+1\}$. Each position $i$ such that $1 \leq i \leq |w|$ carries the label $a_i \in A$. On the other hand, positions $0$ and $|w|+1$ are \emph{artificial} leftmost and rightmost positions which carry \emph{no label}.

\begin{rem}
  With this convention, every word has at least two positions. In particular, the empty word \veps contains only the two unlabeled positions: we have $\pos{\veps}= \{0,1\}$.
\end{rem}

Given a word $w= a_1\cdots a_{|w|} \in A^*$ and a position $i \in \pos{w}$, we associate an element $\wpos{w}{i}\in A\cup \{min,max\}$. If $i=0$, then $\wpos{w}{i}=min$ and if $i = |w|+1$, then $\wpos{w}{i} = max$. Otherwise, $1\leq i \leq |w|$ and we let $\wpos{w}{i} = a_i$. Additionally, if  $i,j \in \pos{w}$ are \emph{two} positions such that $i < j$, we write $\infix{w}{i}{j} = a_{i+1} \cdots a_{j-1} \in A^*$ (\emph{i.e.}, the infix obtained by keeping the letters carried by positions that are \emph{strictly} between $i$ and $j$). Note that $\infix{w}{0}{|w|+1} = w$.

\subparagraph{First-order logic.} We use first-order logic (\fo) to express properties of words $w$ in~$A^*$. A particular formula can quantify over the positions in $w$ and use a predetermined set of predicates to test properties of these positions. We also allow two constants ``$min$'' and ``$max$''  interpreted as the artificial unlabeled positions $0$ and $|w|+1$. For a formula $\varphi(x_1,\dots,x_n)$ with free variables $x_1,\dots,x_n$, a word $w \in A^*$ and positions $i_1,\dots,i_n \in \pos{w}$, we write $w \models \varphi(i_1,\dots,i_n)$ to indicate that $w$ satisfies $\varphi$ when $x_1,\dots,x_n$ are interpreted as the positions $i_1,\dots,i_n$. A sentence $\varphi$ is a formula without free variables. It defines the language $L(\varphi) = \{w \in A^* \mid w \models \varphi\}$.

We use two kinds of predicates. The first kind is standard. For each $a \in A$, we have a unary predicate (also denoted by~$a$)  selecting all positions labeled by ``$a$''. We also use three binary predicates: equality ``$=$'', the (strict) linear order ``$<$'' and the successor ``$+1$''.

\begin{exa}
  The language $A^*abA^*cA^*d$ is defined by the following \fo sentence:
  \[
    \exists x_1\exists x_2\exists x_3\exists x_4\ (x_1+1 = x_2) \wedge (x_2 < x_3) \wedge  (x_4+1 = max) \wedge a(x_1) \wedge b(x_2) \wedge c(x_3) \wedge d(x_4).
  \]
\end{exa}

A \emph{fragment} of first-order logic consists in the specification of a (possibly finite) set $V$ of variables and a set \Fs of first-order formulas using only the variables of $V$. This set \Fs must contain all quantifier-free formulas and must be closed under disjunction, conjunction and  quantifier-free substitution (if $\varphi \in \Fs$, replacing a quantifier-free subformula in $\varphi$ by another quantifier-free formula produces a new formula in \Fs). If \frS is a set of predicates and~\Fs is a fragment of first-order logic, we write $\Fs(\frS)$ for the class containing all languages $L(\varphi)$ where $\varphi$ is a sentence of \Fs using only predicates that belong to $\frS \cup A \cup \{=\}$ (the label predicates and the equality predicate are \emph{always} implicitly allowed). We provide examples at the end of the section (\emph{``Particular fragments''}).

We consider generic \emph{signatures}, \emph{i.e.}, possibly infinite sets of predicates built from an arbitrary input class \Cs. There are of two kinds:
\begin{itemize}
  \item  First, we consider the set $\frI_{\Cs}$ containing a binary ``infix'' predicate $I_L(x,y)$ for every language $L \in \Cs$. Given a word $w \in A^*$ and two positions $i,j \in \pos{w}$, we have $w \models I_L(i,j)$ if and only if $i < j$ \emph{and} $\infix{w}{i}{j} \in L$.
  \item Second, we consider the set $\frP_{\Cs}$ containing a unary ``prefix'' predicate $P_L(x)$ for every language $L \in \Cs$. Given a word $w \in A^*$ and a position $i \in \pos{w}$, we have $w \models P_L(i)$ if and only if $0 < i$ \emph{and} $\infix{w}{0}{i} \in L$.
\end{itemize}
Observe that the predicates available in $\frP_{\Cs}$ are easily expressed from those in $\frI_{\Cs}$: clearly, the formula $P_L(x)$ is equivalent to $I_L(min,x)$.

\smallskip

The sets of predicates $\frP_{\Cs}$ and $\frI_{\Cs}$ are defined for every class \Cs. Yet, we shall mostly be interested in the cases when \Cs is either a \vari of group languages \Gs or its \wsuit extension $\Gs^+$. This is motivated by the following lemma.

\begin{lem}\label{lem:gensig}
  Let \Gs be a \vari of group languages and let \Fs be a fragment of first-order logic. Then, $\Fs(\frI_{\Gs}) = \Fs(<,\frP_{\Gs})$ and $\Fs(\frI_{\Gs^+}) = \Fs(<,+1,\frP_{\Gs})$.
\end{lem}

\begin{proof}
  We first prove that $\Fs(<,\prefsigg) \subseteq \Fs(\infsigg)$ and $\Fs(<,+1,\prefsigg) \subseteq \Fs(\infsiggp)$. It suffices to prove that we may express all atomic formulas of $\Fs(<,\prefsigg)$ and $\Fs(<,+1,\prefsigg)$ using atomic formulas of $\Fs(\infsigg)$ and $ \Fs(\infsiggp)$ respectively. The linear order $x < y$ is expressed by $I_{A^*}(x,y)$. For every $L \in \Gs$,  $P_L(x)$ is expressed by $I_L(min,x)$. Finally, the successor relation $x+1 = y$ is expressed by $I_{\{\veps\}}(x,y)$ (by definition of $\Gs^+$, we know that $I_{\{\veps\}}$ is a predicate of \infsiggp).

  We now prove that $\Fs(\infsigg) \subseteq \Fs(<,\prefsigg)$. By definition of fragments, it suffices to prove that for every $L \in \Gs$, the atomic formula $I_L(x,y)$ is equivalent to a quantifier-free formula of $\Fs(<,\prefsigg)$. Proposition~\ref{prop:genocm} yields a \Gs-morphism $\alpha: A^* \to G$ recognizing~$L$. Since \Gs is a \vari of group languages, $G$ is a group by Lemma~\ref{lem:gmorph}. Since $\alpha $ is a \Gs-morphism, the language $\alpha\inv(g)$ belongs to \Gs for every $g\in G$, whence $P_{\alpha\inv(g)}$ is a predicate available in $\prefsigg$. Let $F \subseteq G$ be the set such that $\alpha\inv(F) = L$. Since $G$ is a group, we have $\alpha(v) = (\alpha(ua))\inv \alpha(uav)$ for all $u,v \in A^*$ and $a \in A$. We define $T = \{(g,a,h) \in G \times A \times G \mid (g\alpha(a))\inv h \in F\}$. Consider the following quantifier-free formula of~$\Fs(<,\prefsigg)$:
  \[
    \varphi(x,y) = (x < y) \wedge \Big(\bigvee_{(g,a,h) \in T} \big(P_{\alpha\inv(g)}(x) \wedge a(x) \wedge P_{\alpha\inv(h)}(y)\big)\Big).
  \]
  One now may verify that the atomic formula $I_L(x,y)$ is equivalent to the following quantifier-free formula of~$\Fs(<,\prefsigg)$:
  \[
    (x = min \wedge P_L(y)) \vee  \varphi(x,y).
  \]
  This concludes the proof of the inclusion $\Fs(\infsigg) \subseteq \Fs(<,\prefsigg)$.

  Finally, we prove that $\Fs(\infsiggp) \subseteq \Fs(<,+1,\prefsigg)$. By definition, it suffices to show that for every language $K \in \Gs^+$, the atomic formula $I_K(x,y)$ is equivalent to a quantifier-free formula of $\Fs(<,+1,\prefsigg)$. By definition of $\Gs^+$, there exists $L \in \Gs$ such that either $K = \{\veps\} \cup L$ or $K = A^+ \cap L$. Consequently, $I_K(x,y)$ is equivalent to either $I_{\{\veps\}}(x,y) \vee I_L(x,y)$ or $I_{A^+}(x,y) \wedge I_L(x,y)$. Since, $L \in \Gs$, we already proved above that $I_L(x,y)$ is equivalent to a quantifier-free formula of $\Fs(<,\prefsigg) \subseteq \Fs(<,+1,\prefsigg)$. Moreover, $I_{\{\veps\}}(x,y)$ is equivalent to $x+1 = y$ and $I_{A^+}$ is equivalent to $x < y \wedge \neg (x+1 = y)$. This concludes the proof.
\end{proof}

Lemma~\ref{lem:gensig} covers many natural sets of predicates. We present three important cases. If \Gs is the trivial \vari $\stzer = \{\emptyset,A^*\}$, all predicates in $\frP_{\stzer}$ are trivial. Hence, we get the classes $\Fs(<)$ and $\Fs(<,+1)$. For the class \md of \emph{modulo languages} (see Section~\ref{sec:group-languages} for its definition), one can check that in this case, we obtain the classes $\Fs(<, \mathit{MOD})$ and $\Fs(<,+1,\mathit{MOD})$ where ``$\mathit{MOD}$'' is the set of \emph{modular predicates}, defined as follows. For all $k,m \in \nat$ such that $k < m$, the set $\mathit{MOD}$ contains a unary predicate $M_{k,m}$ selecting all positions $i$ such that $i \equiv k \bmod m$. Finally, we may consider the class \abg of \emph{alphabet modulo testable languages} (see Section~\ref{sec:group-languages} for its definition). In this case, we get the classes \mbox{$\Fs(<,\mathit{AMOD})$} and $\Fs(<,+1, \mathit{AMOD})$ where ``$\mathit{AMOD}$'' is the set of \emph{alphabetic modular predicates}: for all $a \in A$ and $k,m \in \nat$ such that $k < m$, the set $\mathit{AMOD}$ contains a unary predicate $M^a_{k,m}$ selecting all positions $i$ such $|\prefix{w}{i}|_a \equiv k \bmod m$.

\subparagraph{Particular fragments.} In the paper, we consider \emph{two} special fragments. The first one is \emph{two-variable first-order logic} (\fod). It consists of all first-order formulas that use at most \emph{two} distinct variables (which can be reused). In the formal definition, introducing the fragment \fod boils down to using a set $V$ of variables which has size two. For example, the sentence
\[
  \exists x \exists y\  x < y \wedge M_{1,2}(x) \wedge a(x) \wedge b(y)  \wedge (\exists
  x\ y < x \wedge c(x))
\]
\noindent
is an \fod sentence using the label predicates, the linear order predicate and the modular predicate ``$M_{1,2}$''. It defines the language $(AA)^*aA^*bA^*cA^*$, which therefore belongs to the class \fodwm.

\smallskip

We also consider individual levels in the quantifier alternation hierarchy of full first-order logic. One may classify first-order sentences by counting the alternations between $\exists$ and $\forall$ quantifiers in their prenex normal form. For $n \geq 1$, an \fo sentence is \sic{n} (resp.\ \pic{n}) when its prenex normal form has $(n -1)$ quantifier alternations (that is, $n$ blocks of quantifiers) and starts with an $\exists$ (resp.~a $\forall$) quantifier. For example, a sentence whose prenex normal form is,
\[
  \exists x_1 \exists x_2  \forall x_3 \exists x_4
  \ \varphi(x_1,x_2,x_3,x_4) \quad \text{(with $\varphi$ quantifier-free)}
\]
\noindent
is {\sic 3}. Observe that the sets of \sic{n} and \pic{n} sentences are not closed under negation: negating a \sic{n} sentence yields a \pic{n} sentence and vice versa. Thus, one also considers \bsc{n} sentences: Boolean combinations of \sic{n} sentences.

We are interested in intermediary levels that are not defined directly by a set of formulas.  If \frS is a set of predicates and $n \geq 1$, we may consider the classes $\sic{n}(\frS)$ and $\pic{n}(\frS)$. It is also standard to consider a third class, denoted by $\dec{n}(\frS)$: it consists exactly of the languages that can be defined by both a \sic{n} formula \emph{and} a \pic{n} formula (using the predicates in $\frS \cup A \cup \{=\}$). In other words, $\dec{n}(\frS) = \sic{n}(\frS) \cap \pic{n}(\frS)$. In the paper, we look at the classes $\dec{2}(\infsigc)$.

\section{Operators}
\label{sec:opers}
In this section, we introduce the operators investigated in the paper. They are all variants of the \emph{polynomial closure operator} \polo, which we first define. Let us point out that the paper is not about the operator \polo itself: all results about \polo that we use are taken from previous work. Then, we introduce \emph{unambiguous polynomial closure} (\upolo), which is a semantic restriction of \polo. This is the main operator of the paper. Finally, we present another variant of \polo: \emph{alternated polynomial closure} (\adeto). We use it as a tool for investigating \upolo. In particular, we shall prove that $\upol{\Cs}=\adet{\Cs}$ when \Cs is a \vari.

\subsection{Polynomial closure}

The definition is based on \emph{marked} concatenation. Abusing terminology, given a word $u \in A^*$, we denote by $u$ the singleton language~$\{u\}$. Given two languages $K,L \subseteq A^*$, \emph{a marked concatenation} of $K$ with $L$ is a language of the form $KaL$, for some letter $a \in A$.

\begin{definition}
Consider a class \Cs. The \emph{polynomial closure of \Cs}, denoted by \pol{\Cs} is the least class containing \Cs closed under union and marked concatenation. That is, for every $K,L \in \pol{\Cs}$ and $a \in A$, we have $K \cup L \in \pol{\Cs}$ and $KaL \in \pol{\Cs}$.
\end{definition}

\begin{exa}\label{ex:polst}
  Consider the class $\stzer = \{\emptyset, A^*\}$. Then, using the fact that language concatenation distributes over union, it is easy to check that \pol{\stzer} consists of all finite unions of languages $A^*a_1A^* \cdots a_n A^*$ with $a_1,\dots,a_n \in A$.
\end{exa}

It is not clear from the definition whether the class \pol{\Cs} has robust properties, even when \Cs does. It was shown by Arfi~\cite{arfi87,arfi91} that if \Cs is a \vari, then \pol{\Cs} is a \pvari (the key issue is closure under intersection, which is not required by the definition). This result was later strengthened by Pin~\cite{jep-intersectPOL}: if \Cs is a \pvari, then so is \pol{\Cs} (see also~\cite{PZ:generic_csr_tocs:18} for a more recent proof).

\begin{theorem}[Arfi~\cite{arfi87,arfi91}, Pin~\cite{jep-intersectPOL}]\label{thm:polc}
  Let \Cs be a \pvari. Then, \pol{\Cs} is a \pvari as well.
\end{theorem}

In general, the classes \pol{\Cs} built with polynomial closure are \emph{not} closed under complement, even when the class \Cs is (see \pol{\stzer} in Example~\ref{ex:polst}). Consequently, it is natural to combine polynomial closure with two other independent operators. The first one is \emph{Boolean closure}. Given an input class \Ds, we write \bool{\Ds} for the \emph{smallest} Boolean algebra containing \Ds. Moreover, given an input class \Cs, we write \bpol{\Cs} for \bool{\pol{\Cs}}.

\begin{exa}
  In view of Example~\ref{ex:polst}, \bpol{\stzer} consists of all Boolean combinations of languages $A^*a_1A^* \cdots a_n A^*$ with $a_1,\dots,a_n \in A$. This is the class of~piecewise testable languages~\cite{simonthm}, which is prominent in the literature.
\end{exa}

\begin{exa}\label{ex:polat}
  It can be verified that \bpol{\at} consists exactly of all Boolean combinations of marked products $B_0^*a_1B_1^* \cdots a_{n}B_{n}^*$ with $a_1,\dots,a_n \in A$ and $B_0, \dots,B_n \subseteq A$. It is known~\cite{pin-straubing:upper}, that this is exactly the class \bpol{\bpol{\stzer}}, as we shall see in Example~\ref{rem:upolat} below.
\end{exa}

Since one quotients commute with Boolean operations, the following statement is an immediate corollary of Theorem~\ref{thm:polc}.

\begin{cor}\label{cor:bpolc}
  Let \Cs be a \pvari. Then, \bpol{\Cs} is a \vari.
\end{cor}

The second operator is defined as follows. Given a class \Cs, we write \copol{\Cs} for the class containing all complements of languages in \pol{\Cs}. That is,
\[
  \copol{\Cs} = \big\{A^* \setminus L \mid L \in \pol{\Cs}\big\}.
\]
We shall consider the class \capol{\Cs} consisting of all languages that belong to both \pol{\Cs} and \copol{\Cs}. Consequently,
\[
  \capol{\Cs} = \big\{L \mid L \in \pol{\Cs} \text{ and } A^* \setminus L \in \pol{\Cs}\big\}.
\]
In other words, since \pol{\Cs} is closed under union and intersection, \capol{\Cs} is the \emph{greatest} Boolean algebra contained in \pol{\Cs}. We have the following immediate corollary of Theorem~\ref{thm:polc}.

\begin{cor}\label{cor:polc}
  Let \Cs be a \pvari. Then, \copol{\Cs} is a \pvari and \capol{\Cs} is a \vari.
\end{cor}

Finally, we have the following useful result about classes of the form \pol{\bpol{\Cs}} which is taken from~\cite{PZ:generic_csr_tocs:18}. We recall the proof, which is straightforward.

\begin{lem}\label{lem:elimbool}
  Let \Cs be a \vari. Then, $\pol{\bpol{\Cs}} = \pol{\copol{\Cs}}$.
\end{lem}

\begin{proof}
  It is clear that $\copol{\Cs} \subseteq \bpol{\Cs}$, whence $\pol{\copol{\Cs}} \subseteq \pol{\bpol{\Cs}}$. We show that $\bpol{\Cs} \subseteq \pol{\copol{\Cs}}$. It will follow that,
  \[
    \pol{\bpol{\Cs}} \subseteq \pol{\pol{\copol{\Cs}}} = \pol{\copol{\Cs}}.
  \]
  Let $L\in \bpol{\Cs}$. By definition, $L$ is a Boolean combination of languages in \pol{\Cs}. Hence, De Morgan's laws show that $L$ is built by applying unions and intersections to languages that are elements of either \pol{\Cs} or \copol{\Cs}. Clearly, \pol{\copol{\Cs}} contains both \pol{\Cs} and \copol{\Cs}. Moreover, it follows from Theorem~\ref{thm:polc} and Corollary~\ref{cor:polc} that \pol{\copol{\Cs}} is a \pvari. Altogether, we obtain $L \in \pol{\copol{\Cs}}$ as~desired.
\end{proof}

\subparagraph{Logical characterizations.} It is well known that the quantifier alternation hierarchies of first-order logic can be characterized using the operators \polo and \boolo. Historically, this was first proved by Thomas~\cite{ThomEqu}.

Given a single input class \Cs, the \emph{concatenation hierarchy of basis \Cs} is built from by applying \polo and \boolo to \Cs in alternation. More precisely, the hierarchy is made of two kinds of levels. The \emph{full levels} are denoted by natural numbers $0,1, 2,3,\dots$ The half levels are denoted by ${1}/{2},{3}/{2},{5}/{2},\dots$ Level $0$ is the basis \Cs. Then, for every number $n \in \nat$:
\begin{itemize}
  \item the half level $n + {1}/{2}$ is the polynomial closure (\polo) of level $n$.
  \item the full level $n+1$ is the Boolean closure (\boolo) of level $n+{1}/{2}$.
\end{itemize}
When \Cs is a \vari, this hierarchy exactly characterizes the quantifier alternation hierarchy of first-order logic associated to the set of predicates \infsigc (see~\cite{PZ:generic_csr_tocs:18}). More precisely, for all $n \geq 1$, level $n - {1}/{2}$ is exactly the class $\sic{n}(\infsigc)$ and level $n$ is exactly the class $\bsc{n}(\infsigc)$.

We are mainly interested in the intermediary levels $\dec{n}(\infsigc) = \sic{n}(\infsigc) \cap \pic{n}(\infsigc)$ and more precisely in the special case where $n= 2$. As explained above, $\sic{2}(\infsigc) = \pol{\bpol{\Cs}}$ and one can check that this implies $\pic{2}(\infsigc) = \copol{\bpol{\Cs}}$. We get the following theorem.

\begin{theorem}\label{thm:capoldel2}
  For every \vari \Cs, we have $\dec{2}(\infsigc) = \capol{\bpol{\Cs}}$.
\end{theorem}

\subsection{Unambiguous polynomial closure}

We turn to the main operator of the paper. It is defined as a variant of polynomial closure obtained by restricting marked concatenations to \emph{unambiguous} ones and unions to \emph{disjoint ones}.  We first define what it means for some marked concatenation to be unambiguous. In fact, we introduce three particular kinds of marked concatenations.

\subparagraph{Deterministic marked concatenations.} We define three properties. Let $K,L \subseteq A^*$ and $a \in A$. Consider the marked concatenation $KaL$. We say that,
\begin{itemize}
  \item $KaL$ is \emph{left deterministic} when $K \cap KaA^* = \emptyset$.
  \item $KaL$ is \emph{right deterministic} when $L \cap A^*aL = \emptyset$.
  \item $KaL$ is \emph{unambiguous} when for every $w \in KaL$, there exists a \emph{unique} decomposition of $w$ witnessing this membership. In other words, $KaL$ is unambiguous if for every $u,u' \in K$ and $v,v' \in L$ such that $uav = u'av'$, we have $u= u'$ and $v= v'$.
\end{itemize}

\begin{exa}\label{exa:unamb-vs-det}
  Let $A = \{a,b,c\}$. Then $b^*aA^*$ is left deterministic and $A^*ab^*$ is right deterministic. Moreover, they are both unambiguous. On the other hand, $(ab)^+a(ca)^+$ is unambiguous but is neither left, nor right deterministic. Finally, $A ^*aA^*$ is ambiguous: $aa \in A^*aA^*$ and there are two decompositions of $aa$ witnessing this membership.
\end{exa}

\begin{rem}\label{rem:uconcagroup}
  Marked concatenations of nonempty group languages are ambiguous. Indeed, let $K,L\subseteq A^*$ be two nonempty group languages and let $a\in A$. Let $\alpha:A^*\to G$ be a morphism into a finite group $G$ recognizing both $K$ and $L$, and let $p=\omega(G)>0$. Moreover, let $u\in K$ and $v\in L$. Then $\alpha(ua^{p})=\alpha(u)$, whence $ua^p\in K$. Similarly, $a^{p}v\in L$. Therefore, the word $(ua^p)av=ua(a^pv)$ has two decompositions witnessing its membership in $KaL$.
\end{rem}

\begin{rem}
  Being deterministic or unambiguous is a semantic property: whether $KaL$ satisfies it may not be apparent on the definitions of $K$ and $L$. Moreover, this depends on $K,L$ and $a$, not just on the concatenated language $KaL$. There may exist two marked concatenations representing the same language and such that one is unambiguous while the other is not. For example, if $A = \{a,b\}$, it is clear that $b^*aA^* = A^*aA^*$. However, $b^*aA^*$ is left deterministic and unambiguous while $A^*aA^*$ is neither.
\end{rem}

\begin{rem}\label{rem:sublang-unamb}
  It is immediate from the definitions that if the marked concatenation $KaL$ is unambiguous (resp.\ left deterministic, right deterministic), then for every $K' \subseteq K$ and $L' \subseteq L$, the marked concatenation $K'aL'$ is also unambiguous (resp.\ left deterministic, right deterministic). We shall often use this simple fact implicitly.
\end{rem}

An alternative way to introduce left/right deterministic marked concatenations is to present them as special kinds of unambiguous marked concatenations. We detail this point in the following lemma.

\begin{lem}\label{lem:detuna}
  Let $K,L \subseteq A^*$ and $a \in A$. The two following properties hold:
  \begin{itemize}
    \item $KaL$ is left deterministic if and only if $KaA^*$ is unambiguous.
    \item $KaL$ is right deterministic if and only if $A^*aL$ is unambiguous
  \end{itemize}
\end{lem}

\begin{proof}
  By symmetry, we only prove the first assertion. We prove its contrapositive: $KaL$ is \emph{not} left deterministic if and only if $KaA^*$ is \emph{not} unambiguous. Assume first that $KaL$ is not left deterministic. This yields $u \in K \cap KaA^*$. Hence, we get $u' \in K$ and $v \in A^*$ such that $u = u'av$. Hence, $ua = u'ava \in KaA^*$ and there are two decompositions witnessing this fact: $(u)a(\veps)$ and $(u') a (va)$. We conclude that $KaA^*$ is ambiguous. Conversely, assume that $KaA^*$ is ambiguous. We get $u,u' \in K$ and $v,v' \in A^*$ such that $u \neq u'$ and $uav = u'av'$. By hypothesis, either $ua$ is a prefix of $u'$ or $u'a$ is a prefix of $u$. By symmetry, we consider the case where $ua$ is a prefix of $u'$. This yields $x \in A^*$ such that $u'=uax$. Hence,  $u' \in K \cap KaA^*$ and we conclude that $KaL$ is not left deterministic, as desired.
\end{proof}

Note that Lemma~\ref{lem:detuna} and Remark~\ref{rem:sublang-unamb} imply that left (resp.\ right) deterministic marked concatenations are unambiguous. The converse is not true, as shown by Example~\ref{exa:unamb-vs-det}. We complete the definition with a lemma connecting left (resp. right) deterministic marked concatenations to the Green relation \Rrel (resp.\ \Lrel) defined on monoids. We shall use it in order to prove that marked concatenations are left or right deterministic.

\begin{lem}\label{lem:greendet}
  Let $\alpha: A^* \to N$ be a morphism into a finite monoid. Let $s \in N$, $a \in A$ and $L \subseteq \alpha\inv(s)$. The two following properties hold:
  \begin{itemize}
    \item if $s\alpha(a) \Rords s$, then $LaA^*$ is left deterministic.
    \item if $\alpha(a)s \Lords s$, then $A^*aL$ is right deterministic.
  \end{itemize}
\end{lem}

\begin{proof}
  By symmetry, we only prove the first assertion. We assume that $s\alpha(a)\Rords s$ and show that $LaA^*$ is left deterministic. By contradiction, assume that there exists a word $w \in L \cap LaA^*$. Since $w \in LaA^*$, we get $u\in L$ and $v \in A^*$ such that $w = uav$. Since $u,w \in L$ and $L \subseteq \alpha\inv(s)$, we have $\alpha(u)=\alpha(w) = s$. Therefore, $s = s \alpha(a)\alpha(v)$ which yields $s \Rord s \alpha(a)$, contradicting the hypothesis that $s\alpha(a) \Rords s$.
\end{proof}

\begin{definition}[Unambiguous polynomial closure]
Consider a class~\Cs. Its \emph{unambiguous polynomial closure}, denoted by \upol{\Cs} is the least class containing~\Cs and closed under both \emph{disjoint} union and \emph{unambiguous} marked concatenation. That is given $K,L \in \upol{\Cs}$ and $a \in A$, if $K \cap L = \emptyset$, then $K \uplus L \in \upol{\Cs}$ and if $KaL$ is unambiguous, then $KaL \in \upol{\Cs}$. Note that here, we choose to use the symbol ``$\uplus$'' for union in order to emphasize disjointedness.
\end{definition}

\begin{rem}\label{rem:upolinc}
  It is clear from the definition that $\upol{\Cs}\subseteq\pol{\Cs}$ for every class \Cs. In general, this inclusion is strict. This will be the case for all examples considered in the paper. For example, consider the class $\stzer = \{\emptyset,A^*\}$. We have $\upol{\stzer}=\stzer$  (this is because marked concatenations $A^*aA^*$ are ambiguous). Hence, $\upol{\stzer} \subsetneq \pol{\stzer}$.
\end{rem}

\begin{exa}\label{rem:upolgroup}
  As seen in Remark~\ref{rem:upolinc}, we have $\upol{\stzer} = \stzer$. This is actually a generic property of \varis of group languages: for every such class \Gs, we have $\upol{\Gs} = \Gs$. This is because the marked concatenation of two nonempty group languages can never be unambiguous, see Remark~\ref{rem:uconcagroup}. This illustrates a particularity of \upolo: it is meant to be applied to ``complex enough'' classes.
\end{exa}

\begin{exa}\label{rem:upolat}
  An important example is \upol{\at}. It is simple to verify that this class contains exactly the finite disjoint unions of unambiguous marked products $B_0^*a_1 B_1^* \cdots a_nB_n^*$ where $a_1,\dots,a_n \in A$ and $B_0,\dots,B_n \subseteq A$ (unambiguous marked products are defined in the natural way by generalizing the definition to products of arbitrarily many languages). It is known as the class of \emph{unambiguous languages}~\cite{schul}. We prove below that $\upol{\at} = \upol{\bpol{\stzer}}$. Note that it follows from this equality that  $\pol{\at} = \pol{\bpol{\stzer}}$ and that $\bpol{\at} = \bpol{\bpol{\stzer}}$, as stated in Example~\ref{ex:polat}.
\end{exa}

\begin{rem}
  In the literature, a less restrictive definition, which allows arbitrary finite unions, is sometimes used (see~\cite{jep-intersectPOL} for example). It turns out that this is equivalent when the input class \Cs is a \vari. This is the only case that we consider in the paper. Yet, allowing arbitrary finite unions is important when the input class \Cs is not closed under complement.
\end{rem}

We now present properties of unambiguous polynomial closure that we shall prove in the paper. We start with an important theorem, which generalizes a result by Pin and Weil~\cite{pwdelta2}. It turns out that when the input class \Cs is a \vari\footnote{The original result of Pin and Weil only applies to input classes satisfying stronger properties involving closure under inverse images of morphisms.}, the classes \upol{\Cs} and \capol{\Cs} coincide.

\begin{theorem}\label{thm:polcopol}
  For every \vari \Cs, we have $\upol{\Cs} = \capol{\Cs}$.
\end{theorem}

We postpone the proof of Theorem~\ref{thm:polcopol} to Section~\ref{sec:caracupol}. It is based on algebraic characterizations: we prove independently that for every \vari \Cs, the classes \upol{\Cs} and \capol{\Cs} have the same algebraic characterization. Hence, they are equal.

\begin{rem}\label{rem:eqfails}
  Theorem~\ref{thm:polcopol} fails if \Cs is only a \pvari (\emph{i.e.}, when \Cs is not closed under complement). For example, let $\Cs = \pol{\stzer}$. Clearly, we have $\upol{\Cs} = \Cs$. On the other hand, one can verify that $\capol{\Cs} = \capol{\stzer} = \stzer$. Hence, we have the strict inclusion $\capol{\Cs} \subsetneq \upol{\Cs}$ in this case.
\end{rem}

\begin{rem}
  In view of Theorem~\ref{thm:capoldel2}, a consequence of Theorem~\ref{thm:polcopol} is that for every \vari \Cs, we have $\upol{\bpol{\Cs}}= \dec{2}(\infsigc)$. Hence, we get a first logical characterization of \upolo. We come back to this point in Section~\ref{sec:logcar} where we present a second characterization of \upolo in terms of two-variable first-order logic (note however that in this case, we need stronger hypotheses on the input class \Cs).
\end{rem}

With Theorem~\ref{thm:polcopol} in hand, it is simple to describe the closure properties of \upol{\Cs}. By Corollary~\ref{cor:polc}, we know that when \Cs is a \vari, then so is \capol{\Cs}. In view of Theorem~\ref{thm:polcopol}, we obtain the same result for \upol{\Cs}. This is surprising since closure under Boolean operations is not apparent on the definition of \upol{\Cs}. This generalizes a result\footnote{As for Theorem~\ref{thm:polcopol}, the original result of Pin, Straubing and Thérien only applies to input classes satisfying stronger properties involving closure under inverse images by morphisms.} by Pin, Straubing and Thérien~\cite{Pinambigu}.

\begin{theorem}\label{thm:comp}
  For every \vari \Cs, the class \upol{\Cs} is a \vari as well.
\end{theorem}

We may use Theorem~\ref{thm:comp} to prove a property for a particular input class \Cs that will be important for the applications of our results. As mentioned in Example~\ref{rem:upolat}, we have the equality $\upol{\bpol{\stzer}} = \upol{\at}$.

\begin{lem}\label{lem:upolat}
  We have $\upol{\bpol{\stzer}} = \upol{\at}$.
\end{lem}

\begin{proof}
  It is immediate from the definitions that $\at \subseteq \bpol{\stzer}$. Consequently, we have $\upol{\at} \subseteq \upol{\bpol{\stzer}}$. For the converse inclusion, we show that $\pol{\stzer} \subseteq \upol{\at}$. Since \upol{\at} is a Boolean algebra by Theorem~\ref{thm:comp}, this yields $\bpol{\stzer}\subseteq\upol{\at}$ and therefore $\upol{\bpol{\stzer}} \subseteq \upol{\upol{\at}} = \upol{\at}$.   We now concentrate on proving that $\pol{\stzer} \subseteq \upol{\at}$. One may verify from the definition that $\pol{\stzer}$ contains exactly the finite unions of languages $A^*a_1 \cdots A^* a_nA^*$ with $a_1,\dots,a_n \in A$. Hence, since \upol{\at} is a \vari by Theorem~\ref{thm:comp}, it suffices to prove that $A^*a_1 \cdots A^*a_nA^* \in \upol{\at}$ for every $a_1,\dots,a_n \in A$. For $0 \leq i<  n$, we let $B_i = A \setminus \{a_{i+1}\}$. Clearly, $B_i^* \in \at$ for every $i < n$. Moreover, one can check that,
  \[
    A^*a_1 \cdots A^{*}a_{n}A^* = B_0^* a_1 \cdots B_{n-1}^{*}a_{n}A^{*}.
  \]
  It can then be verified that the marked concatenations $B_i^*a_{i+1}(B^*_{i+1}a_{i+2} \cdots B_{n-1}^*a_nA^*)$ are all unambiguous for $0 \leq i < n$ (they are actually left deterministic). Hence, an immediate induction yields $B_0^* a_1 \cdots B_{n-1}^*a_nA^* \in \upol{\at}$ and we get $A^*a_1 \cdots A^*a_nA^*  \in \upol{\at}$, completing the proof.
\end{proof}

We conclude the presentation with a characteristic property of the classes \upol{\Cs} when~\Cs is an arbitrary \vari. We use it in Section~\ref{sec:caracupol} to prove the algebraic characterization of unambiguous polynomial closure.

\begin{prop}\label{prop:carprop}
  Let \Cs be a \vari and $L \in \upol{\Cs}$. There exists a \Cs-morphism $\eta: A^* \to N$ and $k \in \nat$ satisfying the following property:
  \begin{equation}\label{eq:upolp}
    \text{for all $u,v,v',x,y \in A^*$, if $\eta(u) = \eta(v) = \eta(v')$, then
      $xu^{k}vu^{k}y \in L \Leftrightarrow xu^{k} v' u^{k} y \in L$.}
  \end{equation}
\end{prop}

\begin{proof}
  By definition, there exists a \emph{finite} set of languages $\Hb \subseteq \Cs$ such that $L$ is built from the languages in \Hb using only unambiguous marked concatenations and disjoint unions. Proposition~\ref{prop:genocm} yields a \Cs-morphism $\eta: A^*\to N$ recognizing every $H \in \Hb$. We now use induction on the construction of $L$ from the languages in \Hb to prove that there exists $k \in \nat$ such that~\eqref{eq:upolp} holds. Assume first that $L \in \Hb$ (which means that $L$ is recognized by $\eta$). We show that~\eqref{eq:upolp} holds for $k = 0$. Indeed, let $u,v,v',x,y \in A^*$ such that $\eta(u) = \eta(v) = \eta(v')$. It follows that $\eta(xu^{0}vu^{0}y) = \eta(xu^{0}v'u^{0}y)$. Since $L$ is recognized by $\eta$, this yields $xu^{0}vu^{0}y \in L \Leftrightarrow xu^{0} v'u^{0} y \in L$ as desired.

  We now assume that $L \not\in \Hb$. In this case, there exist $L_1,L_2$, built from the languages in \Hb using unambiguous marked concatenations and disjoint unions such that either $L = L_1 \uplus L_2$, or $L = L_1aL_2$ for some $a \in A$ and $L_1aL_2$ is unambiguous. Using induction, for $i = 1,2$, we get $k_i \in \nat$ such that $L_i$ satisfies~\eqref{eq:upolp} for $k = k_i$. If $L = L_1 \uplus L_2$, it is immediate that $L$ satisfies~\eqref{eq:upolp} for $k = max(k_1,k_2)$. We focus on the case $L = L_1aL_2$. Since $L_1$ and $L_2$ are regular (as they belong to \upol{\Cs}), there exists a morphism $\alpha: A^* \to M$ into a finite monoid $M$ recognizing them both. Let $p = \omega(M) \geq 1$ be the idempotent power of this monoid. Moreover, let $h = max(k_1,k_2)$ and $k = p+h$. We prove that $L = L_1aL_2$ satisfies~\eqref{eq:upolp} for this number $k$. Let $u,v,v',x,y\in A^*$ such that $\eta(u) = \eta(v) = \eta(v')$. Assume that $xu^{k}vu^{k}y \in L$, we prove that $xu^{k}v'u^{k}y \in L$ (the converse implication is symmetrical).

  Since $L = L_1aL_2$, we have $w_1 \in L_1$ and $w_2 \in L_2$ such that $xu^{k}vu^{k}y = w_1aw_2$. Using the hypothesis that $L_1aL_2$ is unambiguous, we prove the following fact.

  \begin{lem}\label{lem:carprop}
    Either $xu^{k} vu^{h}$ is a prefix of $w_1$ or $u^{h} vu^{k} y$ is a suffix of $w_2$.
  \end{lem}

  \begin{proof}
    By contradiction, we assume that $xu^{k} vu^{h}$ is not a prefix of $w_1$ and $u^{h} vu^{k} y$ is not suffix of $w_2$. Since $k = p+h$, we have $w_1aw_2 = xu^{k} vu^{k} y= xu^{p}u^{h} vu^{h}u^{p}y$. Therefore, our hypothesis yields that $u^{p}y$ is a suffix of $w_2$ and $xu^{p}$ is a prefix of $w_1$. We get $z_1,z_2 \in A^*$ such that $w_1 = xu^{p}z_1$ and $w_2 = z_2u^{p} y$. Since $w_1aw_2 = xu^{p}u^{h} vu^{h}u^{p}y$, this also implies $u^{h} vu^{h} = z_1az_2$. Let $\ell = (p-1)(2h+1)$. We prove that $w_1az_2u^{\ell}z_1 \in L_1$ and $z_2u^{\ell}z_1aw_2 \in L_2$. Since we have $w_1 \in L_1$ and $w_2 \in L_2$, this yields that $w_1az_2u^{\ell}z_1aw_2 \in L_1aL_2$ and that this membership is witnessed by \emph{two} decompositions: $w_1az_2u^\ell z_1\cdot a\cdot w_2$ and $w_1\cdot a\cdot z_2u^\ell z_1aw_{2}$. Hence, this contradicts the hypothesis that $L_1aL_2$ is unambiguous, finishing the proof.

    We only prove that $w_1az_2u^{\ell}z_1 \in L_1$ (that $z_2u^{\ell}z_1aw_2 \in L_2$ is symmetrical and left to the reader). Since $w_1 = xu^{p}z_1$ and $u^{h} vu^{h} = z_1az_2$, we have,
    \[
      w_1az_2u^{\ell}z_1 = xu^{p}u^hvu^hu^{\ell}z_1.
    \]
    Thus, it remains to show that $xu^{p}u^hvu^hu^{\ell}z_1 \in L_1$. By definition, $xu^{p}z_1 = w_1 \in L_1$. Since $p$ is the idempotent power of a finite monoid recognizing $L_1$, this yields $xu^{p}u^{p(2h+1)}z_1 \in L_1$. This can be reformulated as $xu^{p}u^huu^hu^{\ell}z_1\in L_1$. By definition, $\eta(v) = \eta(u)$ and $L_1$ satisfies~\eqref{eq:upolp} for $k_1 \leq h$. Hence, we obtain $xu^{p}u^hvu^hu^{\ell}z_1\in L_1$, concluding the proof of the lemma.
  \end{proof}

  We may now prove that $xu^{k}v'u^{k}y \in L$. In view of Lemma~\ref{lem:carprop}, there are two cases. First, assume that $xu^{k} vu^{h}$ is a prefix of $w_1$. This yields $z \in A^*$ such that $w_1 = xu^{k} vu^{h}z$ and $u^{k-h}y = zaw_2$. Since $w_1 \in L_1$, $\eta(u) = \eta(v) = \eta(v')$ and $L_1$ satisfies\eqref{eq:upolp} for $k_1 \leq h$, it follows that, $xu^{k} v'u^{h}z \in L_1$. Thus, $xu^{k} v'u^{k} y = xu^{k} v'u^{h}zaw_2 \in L_1aL_2$ as desired. In the second case, $u^{h} vu^{k} y$ is a suffix of $w_2$. This yields $z \in A^*$ such that $w_2 = zu^{h} vu^{k} y$ and $xu^{k-h} = w_1az$. Since $w_2 \in L_2$,  $\eta(u) = \eta(v) = \eta(v')$ and $L_2$ satisfies~\eqref{eq:upolp} for $k_2 \leq h$, it follows that, $zu^{h} v'u^{k} y \in L_2$. We obtain $xu^{k} v'u^{k} y = w_1azu^{h} v'u^{k} y \in L_1aL_2$, as desired.
\end{proof}

\subsection{Alternated polynomial closure}

We present a third and final operator. For every input class \Cs, we define \adet{\Cs} as the least class containing \Cs closed under \emph{disjoint union}, \emph{left deterministic marked concatenation} and \emph{right deterministic marked concatenation}.

We now consider a technical restriction of \adeto, defined by further restricting the situations in which using marked concatenation is allowed. While less natural, it will be useful in proof arguments. Let \Cs be a class of languages. Given two languages $K,L \subseteq A^*$ and a letter $a \in A$, we say that the marked concatenation $KaL$ is left (resp.~right) \mbox{\emph{\Cs-deterministic}} when there exist $K',L' \in \Cs$ such that $K \subseteq K'$, $L \subseteq L'$ and $K'aL'$ is left (resp.~right)~deterministic.

\begin{rem}
  Clearly, $KaL$ being left (resp. right) \Cs-deterministic implies that it is also left (resp. right) deterministic. The converse need not be true. For example, when $\Cs = \stzer = \{\emptyset,A^*\}$, no nonempty marked concatenation can be left or right \stzer-deterministic.
\end{rem}

The \emph{weak alternated polynomial closure} of a class \Cs, which we denote by \wadet{\Cs}, is the least class containing \Cs and closed under \emph{disjoint union} and \emph{left and right \Cs-deterministic marked concatenation}. Clearly, we have $\wadet{\Cs} \subseteq \adet{\Cs}$. Moreover, we shall prove in Section~\ref{sec:caracupol} that when \Cs is a \vari, the converse inclusion holds as well, so that $\wadet{\Cs} = \adet{\Cs}$. This is useful to prove that \adet{\Cs} is included in another class: it suffices to prove that \wadet{\Cs} is included in this class, which is often simpler.

While \adeto and \wadeto are less prominent than \upolo, they serve as key tools in the paper. This is because of the following theorem that we shall prove in Section~\ref{sec:caracupol}.

\begin{theorem}\label{thm:apol}
  Let \Cs be a \vari. Then, $\upol{\Cs} = \adet{\Cs} = \wadet{\Cs}$.
\end{theorem}

Recall that by Lemma~\ref{lem:detuna}, a left/right deterministic marked concatenation is necessarily unambiguous. Hence, the inclusions $\wadet{\Cs} \subseteq \adet{\Cs}\subseteq\upol{\Cs}$ in Theorem~\ref{thm:apol} are immediate from the definitions. On the other hand, the proof argument for the inclusion $\upol{\Cs} \subseteq \wadet{\Cs}$ is intertwined with that of the generic algebraic characterization of \upol{\Cs}. More precisely, we prove that every language satisfying the characterization belongs to \wadet{\Cs} (which also implies membership in \upol{\Cs}). This is why \adet{\Cs} and \wadet{\Cs} are important. In practice, left/right deterministic marked concatenations are simpler to handle than unambiguous ones. By Theorem~\ref{thm:apol}, they suffice to build \emph{all} languages in~\upol{\Cs}.

\section{Canonical relations associated to morphisms}
\label{sec:cano}
Given a class \Cs and a morphism $\alpha: A^* \to M$ into a finite monoid, we associate three distinct objects: a submonoid of $M$ and two relations.  They are central ingredients for formulating the algebraic characterizations of \pol{\Cs} and \upol{\Cs} that we present in Sections~\ref{sec:caracpol},~\ref{sec:caracupol} and~\ref{sec:logcar}. Our characterization algorithms rely on these notions for~\Cs: the \Cs-kernel, the \Cs-pair and the \Cs-preorder relations.

\subsection{\Gs-kernels}

Consider a class \Gs and a morphism $\alpha: A^* \to M$ into a finite monoid $M$. We associate to $\alpha$ a subset $N \subseteq M$ that we call the \emph{\Gs-kernel} of $\alpha$. It consists of all elements $s \in M$ such that $\{\veps\}$ is \emph{not} \Gs-separable from $\alpha\inv(s)$. Additionally, we shall consider a slightly more restrictive notion: the \emph{strict \Gs-kernel of $\alpha$} is the set $S = N \cap \alpha(A^+)$.

\begin{rem}
  While the definition makes sense for an arbitrary class \Gs, it is meant to be used in the special case where \Gs is a \vari of group languages (hence the notation~\Gs).
\end{rem}

\begin{rem}
  When \Gs is the class \md of modulo languages, it can be shown that the \md-kernel of a morphism corresponds to a standard notion: the stable monoid, defined in~\cite{stablemono}. Given a morphism $\alpha: A^*\to M$ into a finite monoid, it can be verified that there exists a number $d \geq 1$ such that $\alpha(A^{2d}) = \alpha(A^d)$. The least such number $d \geq 1$ is called the stability index of $\alpha$. The stable semigroup of $\alpha$ is $S = \alpha(A^d)$ and the stable monoid of $\alpha$ is $N = \{1_M\} \cup S$. One may verify that $S$ and $N$ are respectively the strict \md-kernel and the \md-kernel of $\alpha$ (this follows from a simple analysis of \md-separation).
\end{rem}

Clearly, having a \Gs-separation algorithm in hand suffices to compute the \Gs-kernel of an input morphism $\alpha$. This yields the following lemma.

\begin{lem}\label{lem:kercomp}
  Let \Gs be a class with decidable separation. Given a morphism $\alpha: A^*\to M$ into a finite monoid as input, one can compute both the \Gs-kernel and the strict \Gs-kernel of $\alpha$.
\end{lem}

We now characterize \Gs-kernels in terms of \Gs-morphisms when \Gs is a \vari.

\begin{lem}\label{lem:morker}
  Let \Gs be a \vari, let $\alpha: A^* \to M$ be a morphism into a finite monoid and let $N$ be the \Gs-kernel of $\alpha$. The following properties hold:
  \begin{itemize}
	\item Let $\eta: A^* \to G$ be a \Gs-morphism. For every $s \in N$, there exists $w \in A^*$ such that $\alpha(w) = s$ and $\eta(w) =1_G$.
	\item There exists a \Gs-morphism $\eta: A^* \to G$ such that for every $w \in A^*$, if $\eta(w) = 1_G$, then $\alpha(w) \in N$.
  \end{itemize}
\end{lem}

\begin{proof}
  For the first assertion, let $\eta: A^* \to G$ be a \Gs-morphism and consider $s \in N$. By hypothesis, $\eta\inv(1_G) \in \Gs$ and $\{\veps\} \subseteq \eta\inv(1_G)$. Since $s \in N$ and $N$ is the \Gs-kernel of $\alpha$, the language $\eta\inv(1_G)\in\Gs$ cannot separate $\alpha\inv(s)$ from $\{\veps\}$. Since it includes $\{\veps\}$, we have $\alpha\inv(s) \cap \eta\inv(1_G) \neq \emptyset$. This yields $w \in A^*$ such that $\alpha(w) = s$ and $\eta(w) =1_G$.

  We turn to the second assertion. By definition of the \Gs-kernel $N$ of $\alpha$, we know that for every $t \in M \setminus N$, there exists $K_t \in \Gs$ such that $\veps \in K_t$ and $\alpha\inv(t) \cap K_t = \emptyset$. Proposition~\ref{prop:genocm} yields a \Gs-morphism $\eta: A^* \to G$ recognizing all  languages $K_t$ for $t \in M \setminus N$. It remains to prove that given $w \in A^*$ such that $\eta(w) = 1_G$, we have $\alpha(w) \in N$. Assume that $\eta(w)=1_G=\eta(\veps)$. Since $\eta$ recognizes $K_t$ for every $t \in M \setminus N$ and $\veps \in K_t$ by definition, we get $w \in K_t$ for every $t \in M \setminus N$. Since  $\alpha\inv(t) \cap K_t = \emptyset$, we deduce that $\alpha(w) \neq t$ for every $t \in M \setminus N$. This means that $\alpha(w) \in N$,  completing the~proof.
\end{proof}

An important property is that the \Gs-kernel $N$ of a morphism $\alpha: A^* \to M$ is a submonoid of~$M$. Observe that since the \emph{strict} \Gs-kernel of $\alpha$ is $S = N \cap \alpha(A^+)$ by definition, this also implies that $S$ is a subsemigroup of $M$. We state these properties in the following fact, which we shall use implicitly from now on.

\begin{fct}\label{fct:kernelmono}
  Let \Gs be a \vari and $\alpha: A^*\to M$ a morphism into a finite monoid. The \Gs-kernel of $\alpha$ is a submonoid of $M$. Its strict \Gs-kernel is a subsemigroup~of~$M$.
\end{fct}

\begin{proof}
  We write $N$ for the \Gs-kernel of $\alpha$. Lemma~\ref{lem:morker} yields a \Gs-morphism $\eta: A^* \to G$ such that for every $w \in A^*$, if $\eta(w) = 1_G$, then $\alpha(w) \in N$. Since $\eta(\veps) = 1_G$, it follows that $1_M = \alpha(\veps) \in N$. Moreover, if $s,t \in N$, then the first assertion in Lemma~\ref{lem:morker} yields $u,v \in A^*$ such that $\eta(u) = \eta(v) = 1_G$, $\alpha(u) = s$ and $\alpha(v) = t$. Hence, $\eta(uv) = 1_G$ and we get $st = \alpha(uv) \in N$ by definition of $\eta$.
\end{proof}

We conclude the presentation with a lemma, specific to the case where \Gs is a \vari of group languages (which is the only case that we shall consider in practice).

\begin{fct}\label{fct:idemker}
  Let \Gs be a \vari of group languages and let $\alpha: A^* \to M$ be a surjective morphism into a finite monoid. Then, all idempotents in $E(M)$ belongs to the \Gs-kernel of~$\alpha$.
\end{fct}

\begin{proof}
  Let $N$ be the \Gs-kernel of $\alpha$ and let $e\in E(M)$. We prove that $e\in N$. By Lemma~\ref{lem:morker}, there exists a \Gs-morphism $\eta: A^* \to G$ such that for every $w \in A^*$, if $\eta(w) = 1_G$, then $\alpha(w) \in N$. By Lemma~\ref{lem:gmorph}, we know that $G$ is a group. Since $\alpha$ is surjective, there exists $w \in A^*$ such that $\alpha(w) = e$. Let $k = \omega(G)$. Since $G$ is a group, we get $\eta(w^k) = 1_G$ by definition of $k$. Moreover, since $e$ is idempotent, we have $\alpha(w^k)=e$. Hence, $e \in N$ by definition of $\eta$.
\end{proof}

\subsection{\Cs-pairs}

Consider a class \Cs and a morphism $\alpha: A^* \to M$ into a finite monoid.  We define a relation on $M$: the \Cs-pairs (for $\alpha$). Let $(s,t) \in M^2$. We say that,
\begin{equation}\label{eq:cpairs}
  \text{$(s,t)$ is a \emph{\Cs-pair} (for $\alpha$) if and only if $\alpha\inv(s)$ is \emph{not} \Cs-separable from $\alpha\inv(t)$}.
\end{equation}

\begin{rem}
  While we often make this implicit, being a \Cs-pair depends on $\alpha$.
\end{rem}

By definition, the set of \Cs-pairs for $\alpha$ is finite: it is a subset of $M^2$. Moreover, having a \Cs-separation algorithm in hand is clearly enough to compute all \Cs-pairs associated to an input morphism $\alpha$. We state this simple, yet crucial property in the following~lemma.

\begin{lem}\label{lem:septopairs}
  Let \Cs be a class of languages with decidable separation. Then, given a morphism $\alpha: A^*\to M$ into a finite monoid as input, one can compute all \Cs-pairs for $\alpha$.
\end{lem}

We complete the definition with some properties of \Cs-pairs. A simple and useful one is that the \Cs-pair relation is reflexive provided that the morphism $\alpha$ is surjective (which is always the case in practice). Moreover, it is symmetric when \Cs is closed under complement. On the other hand, the \Cs-pair relation is \emph{not} transitive in general (see Example~\ref{exa:cpairnottrans} below).

\begin{lem}\label{lem:pairsreflex}
  Consider a class \Cs and a morphism $\alpha: A^* \to M$ into a finite monoid. The following properties hold:
  \begin{itemize}
	\item If $\alpha$ is surjective, the \Cs-pair relation is reflexive: for every $s \in M$, $(s,s)$ is a \Cs-pair.
	\item If \Cs is closed under complement, the \Cs-pair relation is symmetric: for every \Cs-pair $(s,t) \in M^2$, $(t,s)$ is a \Cs-pair as well.
  \end{itemize}
\end{lem}

\begin{proof}
  For first assertion, assume that $\alpha$ is surjective. Given $s \in M$, we have $\alpha\inv(s) \neq \emptyset$ since $\alpha$ is surjective. Hence, $\alpha\inv(s) \cap \alpha\inv(s) \neq \emptyset$, which implies that $\alpha\inv(s)$ is not \Cs-separable from $\alpha\inv(s)$. This exactly says that $(s,s)$ is a \Cs-pair. For the second assertion, assume that~\Cs is closed under complement and consider a \Cs-pair $(s,t) \in M^2$. It follows that $\alpha\inv(s)$ is not \Cs-separable from $\alpha\inv(t)$. Since \Cs is closed under complement, this implies that $\alpha\inv(t)$ is not \Cs-separable from $\alpha\inv(s)$. Therefore, $(t,s)$ is a \Cs-pair.
\end{proof}

\begin{exa}\label{exa:cpairnottrans}
  The \Cs-pair relation is not transitive in general. Indeed, let $\Cs=\at$ and let~$M$ be the monoid $\{1,a,b,0\}$ where $1$ acts as an identity element, $0$ as an absorbing element, and the rest of the multiplication is given by $aa=ab=ba=bb=0$. Let $A=\{a,b\}$ and $\alpha:A^*\to M$ be the morphism defined by $\alpha(a)=a$ and $\alpha(b)=b$. We have $\alpha\inv(a)=\{a\}$. Therefore, any language of \at containing $\alpha\inv(a)$ also contains $a^{+}$, and intersects $\alpha\inv(0)$ (which is the set of words of length at least~2). Hence, $(a,0)$ is an \at-pair. Likewise, $(0,b)$ is an \at-pair. However, $(a,b)$ is not an \at-pair, since the language $a^+\in\at$ separates $\alpha\inv(a)=\{a\}$ from $\alpha\inv(b)=\{b\}$.
\end{exa}

We now provide a useful characterization of \Cs-pairs via \Cs-morphisms in the special case where \Cs is a \pvari, which is the only case that we consider in practice. It is similar to Lemma~\ref{lem:morker}.

\begin{lem}\label{lem:cmorph}
  Let \Cs be a \pvari and let $\alpha: A^* \to M$ be a morphism into a finite monoid. The two following properties hold:
  \begin{itemize}
	\item For every \Cs-morphism $\eta: A^* \to (N,\leq)$ and every \Cs-pair $(s,t) \in M^2$ for $\alpha$, there exist $u,v \in A^*$ such that $\eta(u) \leq \eta(v)$, $\alpha(u) = s$ and $\alpha(v) = t$.
	\item There exists a \Cs-morphism $\eta: A^* \to (N,\leq)$ such that for all $u,v \in A^*$, if $\eta(u) \leq \eta(v)$, then $(\alpha(u),\alpha(v))$ is a \Cs-pair for $\alpha$.
  \end{itemize}
\end{lem}

\begin{proof}
  Let us start with the first assertion. Let $\eta: A^* \to (N,\leq)$ be a \Cs-morphism and let $(s,t) \in M^2$ be a \Cs-pair for $\alpha$. Let $F \subseteq N$ be the set of all elements $r \in N$ such that $\eta(u) \leq r$ for some $u \in \alpha\inv(s)$. By definition, $F$ is an upper set for the ordering $\leq$ on $N$. Hence, we have $\eta\inv(F) \in \Cs$ since $\eta$ is a \Cs-morphism. Moreover, it is immediate from the definition of $F$ that $\alpha\inv(s) \subseteq \eta\inv(F)$. Therefore, since $(s,t)$ is a \Cs-pair (which means that $\alpha\inv(s)$ cannot be separated from $\alpha\inv(t)$ using a language in \Cs), it follows that $\eta\inv(F) \cap \alpha\inv(t) \neq \emptyset$. This yields $v \in A^*$ such that $\eta(v) \in F$ and $\alpha(v) = t$. Finally, since $v\in F$, the definition of $F$ yields $u \in A^*$ such that $\eta(u) \leq \eta(v)$ and $\alpha(u) = s$, concluding the proof of the first~assertion.

  We turn to the second assertion. Let $P \subseteq M^2$ be the set of all pairs $(s,t) \in M^2$ which are \emph{not} \Cs-pairs. For every $(s,t) \in P$, there exists $K_{s,t} \in \Cs$ separating $\alpha\inv(s)$ from $\alpha\inv(t)$. Proposition~\ref{prop:genocm} yields a \Cs-morphism $\eta: A^* \to (N,\leq)$ such that every language $K_{s,t}$ for $(s,t) \in P$ is recognized by $\eta$. It remains to prove that for every $u,v \in A^*$, if $\eta(u) \leq \eta(v)$, then $(\alpha(u),\alpha(v))$ is a \Cs-pair. We prove the contrapositive. Assuming that $(\alpha(u),\alpha(v))$ is \emph{not} a \Cs-pair, we show that $\eta(u)\not\leq \eta(v)$. By hypothesis, $(\alpha(u),\alpha(v)) \in P$, which means that $K_{\alpha(u),\alpha(v)} \in \Cs$ is defined and separates $\alpha\inv(\alpha(u))$ from $\alpha\inv(\alpha(v))$. Thus, $u \in K_{\alpha(u),\alpha(v)}$ and $v \not\in K_{\alpha(u),\alpha(v)}$. Since $K_{\alpha(u),\alpha(v)}$ is recognized by $\eta$, this implies $\eta(u) \not\leq \eta(v)$.
\end{proof}

We now prove that when \Cs is a \pvari of regular languages (which is the only case that we shall consider), the \Cs-pair relation is compatible with multiplication. This result is similar to Fact~\ref{fct:kernelmono} for \Gs-kernels.

\begin{lem}\label{lem:mult}
  Let \Cs be a \pvari and let $\alpha: A^* \to M$ be a morphism into a finite monoid. If $(s_1,t_1),(s_2,t_2) \in M^2$ are \Cs-pairs, then $(s_1s_2,t_1t_2)$ is a \Cs-pair as well.
\end{lem}

\begin{proof}
  Lemma~\ref{lem:cmorph} yields a \Cs-morphism $\eta: A^* \to (N,\leq)$ such that for all $u,v \in A^*$, if $\eta(u) \leq \eta(v)$, then $(\alpha(u),\alpha(v))$ is a \Cs-pair. Let $(s_1,t_1),(s_2,t_2) \in M^2$ be \Cs-pairs. Since $\eta$ is a \Cs-morphism, it follows from Lemma~\ref{lem:cmorph} that there exist $u_i,v_i \in A^*$ for $i = 1,2$ such that $\eta(u_i) \leq \eta(v_i)$, $\alpha(u_i) = s_i$ and $\alpha(v_i) = t_i$. Clearly, it follows that $\eta(u_1u_2) \leq \eta(v_1v_2)$, $\alpha(u_1u_2) = s_1s_2$ and $\alpha(v_1v_2) = t_1t_2$. Hence, $(s_1s_2,t_1t_2)$ is a \Cs-pair by definition of $\eta$.
\end{proof}

We conclude the presentation with a result connecting \Gs-pairs with the \Gs-kernel, when~\Gs is a \vari of group languages.

\begin{lem}\label{lem:ptoker}
  Let \Gs be a \vari of group languages and let $\alpha: A^* \to M$ be a surjective morphism into a finite monoid. Let $e \in E(M)$ and $s,t \in M$ such that $(e,s)$ is a \Gs-pair. Then, $es$ and $t(est)^{2\omega-1}$ both belong to the \Gs-kernel of $\alpha$.
\end{lem}

\begin{proof}
  Let $N$ be the \Gs-kernel of $\alpha$. Lemma~\ref{lem:morker} yields a \Gs-morphism $\eta: A^* \to G$ such that for every $w \in A^*$, if $\eta(w) = 1_G$, then $\alpha(w) \in N$.  By Lemma~\ref{lem:gmorph}, $G$ is a group. We first show that $es \in N$. Since $(e,s)$ is a \Gs-pair, Lemma~\ref{lem:cmorph} yields $u,v \in A^*$ such that $\eta(u) = \eta(v)$, $\alpha(u) = e$ and $\alpha(v) = s$. Let $p = \omega(G) \times \omega(M)$. Since $\eta(u) = \eta(v)$, we have $\eta(u^{p-1}v) = \eta(u^p) = 1_G$ (recall that $G$ is a group). Hence, by definition of $\eta$, we have $\alpha(u^{p-1}v) \in N$. Since $e$ is an idempotent, this yields $es \in N$, as desired.

  We now prove that $t(est)^{2\omega-1} \in N$. Let $w = u^{p-1}v$. We know that $\eta(w) = 1_G$ and $\alpha(w) =es$. Since $\alpha$ is surjective, we get $x \in A^*$ such that $\alpha(x) = t$. Since $\eta(w) = 1_G$ and $p$ is a multiple of $\omega(G)$, we have $\eta(x(wx)^{2p-1}) = \eta(x^{2p}) = 1_G$. Moreover, since $p$ is a multiple of $\omega(M)$, we have $\alpha(x(wx)^{2p-1}) =  t(est)^{2\omega-1}$. By definition of $\eta$, we get $t(est)^{2\omega-1} \in N$.
\end{proof}

\subsection{Canonical \Cs-preorder and \Cs-equivalence}

Consider a \emph{lattice} \Cs. Moreover, let $\alpha: A^*\to M$ be a morphism into a finite monoid. We define two relations on $M$: a preorder ``\canoac'' which we call the \emph{canonical \Cs-preorder} of $\alpha$ and an equivalence ``\canaeq'' which we call the \emph{canonical \Cs-equivalence} of $\alpha$. Consider a pair $(s,t) \in M^2$. We define,
\begin{alignat}{3}
  &s \canoac t &\ \ \ & \text{if and only if}&\ \ \ &\text{$s \in F \Rightarrow t \in F$ for all $F \subseteq M$ such that $\alpha\inv(F) \in \Cs$.} \label{eq:trans:canoc}\\
  &s \canaeq t &\ \ \ & \text{if and only if}&\ \ \ & \text{$s \in F \Leftrightarrow t \in F$ for all $F \subseteq M$ such that $\alpha\inv(F) \in \Cs$.}\label{eq:trans:canec}
\end{alignat}
It is immediate by definition that \canoac is indeed a preorder on $M$. Moreover, \canaeq is exactly the equivalence induced by \canoac: for every $s,t \in M$, we have $s \canaeq t$ if and only if $s\canoac t$ and $t\canoac s$. From now on, for the sake of avoiding clutter, we shall abuse terminology when the morphism $\alpha$ is understood and we write \canoc for \canoac and \canec for \canaeq. Additionally, for every element $s\in M$, we write $\ctype{s} \in {M}/{\canec}$ for the $\canec$-class of $s$. Finally, note that since \canec is the equivalence induced by \canoc, the preorder \canoc also induces an order on the quotient set ${M}/{\canec}$, which we write \qanoc: for every $s,t \in M$ we have $\ctype{s} \qanoc \ctype{t}$ if and only if $s \canoc t$.

Observe that an immediate key property of these relations is that having an algorithm for \Cs-membership suffices to compute both \canoc and \canec. Indeed, with such a procedure in hand, it is possible to compute all subsets $F \subseteq M$ such that $\alpha\inv(F) \in \Cs$. One may then decide whether $s \canoc t$ for some $s,t \in M$ by checking whether the implication $s \in F \Rightarrow t \in F$ holds for every such subset $F$. We state this in the following lemma.

\begin{lem}\label{lem:membtopairs}
  Let \Cs be a class of languages with decidable membership. Given a morphism $\alpha: A^*\to M$ into a finite monoid as input, one can compute the relations \canoc and \canec on~$M$.
\end{lem}

\subparagraph{Properties.} First, we show that given a morphism $\alpha: A^* \to M$, we are able to characterize the languages simultaneously recognized by $\alpha$ and belonging to \Cs as those which are the inverse image of an \emph{upper set} for the preorder \canoc. Recall that by definition, $F \subseteq M$ is an upper set for \canoc if and only if for every $s \in F$ and $t \in M$ such that $s\canoc t$, we have $t\in F$.

\begin{lem}\label{lem:lcanoc}
  Consider a lattice \Cs and a morphism $\alpha: A^* \to M$ into a finite monoid. For every $F \subseteq M$, we have $\alpha\inv(F) \in \Cs$ if and only if $F$ is an upper set for \canoc.
\end{lem}

\begin{proof}
  We fix $F \subseteq M$ for the proof. Assume first that $\alpha\inv(F) \in \Cs$ we prove that $F$ is an upper set for \canoc. Let $s \in F$ and $t \in M$ such that $s \canoc t$. By definition of \canoc and since $\alpha\inv(F) \in \Cs$, this implies that $t \in F$ as desired. Conversely assume that $F$ is an upper set for \canoc. We prove that $\alpha\inv(F) \in \Cs$. Consider $s\in F$. By definition, we know that for every element $r\not\in F$, we have $s \not\canoc r$. Hence by definition of \canoc, there exists a set $F_{s,r} \subseteq M$ such that $\alpha\inv(F_{s,r}) \in \Cs$, $s \in F_{s,r}$ and $r \not\in F_{s,r}$. It is now immediate that,
  \[
    F = \bigcup_{s \in F} \Bigl(\bigcap_{r \not\in F} F_{s,r}\Bigr).
  \]
  Since inverse image commutes with Boolean operations, we obtain,
  \[
    \alpha\inv(F) = \bigcup_{s \in F} \Bigl(\bigcap_{r \not\in F} \alpha\inv(F_{s,r})\Bigr).
  \]
  We conclude that $\alpha\inv(F) \in \Cs$, since \Cs is a lattice.
\end{proof}

We turn to a key property of the preorder \canoc. We connect it to the \Cs-pair relation associated to every morphism $\alpha: A^* \to M$. Specifically, we show that \canoc is the reflexive transitive closure of the \Cs-pair relation.

\begin{lem}\label{lem:transclos}
  Consider a lattice \Cs and a morphism $\alpha: A^* \to M$ into a finite monoid. Then, the relation \canoc on $M$ is reflexive transitive closure of the \Cs-pair relation.
\end{lem}

\begin{proof}
  We first prove that \canoc contains the reflexive transitive closure of the \Cs-pair relation. Since \canoc is a preorder, it suffices to show that for every \Cs-pair $(q,r) \in M^2$, we have $q \canoc r$. Hence, we fix $F \subseteq M$ such that $\alpha\inv(F) \in \Cs$ and $q \in F$. We have to show that $r \in F$. Clearly, $q \in F$ implies that $\alpha\inv(q) \subseteq \alpha\inv(F)$. Since $(q,r)$ is a \Cs-pair and $\alpha\inv(F) \in \Cs$, we know that $\alpha\inv(r) \cap \alpha\inv(F) \neq \emptyset$, whence $r \in F$, as desired.

  We turn to the converse implication. Let $s,t \in M$ be such that $s \canoc t$. Moreover, let $F \subseteq M$ be the least subset of $M$ containing $s$ and such that for every \Cs-pair $(q,r) \in M^2$, if $q \in F$, then $r \in F$ as well. We have to show that $t \in F$. Since $s \canoc t$, it suffices to show that $\alpha\inv(F) \in \Cs$ by definition of \canoc. For every $q \in F$, we may build a language $H_q \in \Cs$ such that $\alpha\inv(q) \subseteq H_q \subseteq \alpha\inv(F)$. Indeed, for any $r \not\in F$, we know that $(q,r)$ is not a \Cs-pair by definition of $F$. Therefore, there exists $H_{q,r} \in \Cs$  separating $\alpha\inv(q)$ from $\alpha\inv(r)$. We let,
  \[
    H_q = \bigcap_{r \not \in F} H_{q,r}.
  \]
  Clearly $H_q \in \Cs$ since \Cs is a lattice. It now suffices to observe that,
  \[
    \alpha\inv(F) = \bigcup_{q \in F} \alpha\inv(q) \subseteq \bigcup_{q \in F} H_q \subseteq \alpha\inv(F).
  \]
  Therefore, $\alpha\inv(F) =\bigcup_{q \in F} H_q$ belongs to \Cs since \Cs is lattice. This concludes the proof.
\end{proof}

We now consider the case where \Cs is a \emph{Boolean algebra}. We show that in this case the preorder \canoc coincides with the equivalence \canec for every morphism $\alpha: A^*\to M$.

\begin{lem}\label{lem:bcanoc}
  Consider a Boolean algebra \Cs and a morphism $\alpha: A^* \to M$ into a finite monoid. For every $s,t \in M$, we have $s \canoc t$ if and only if $s \canec t$.
\end{lem}

\begin{proof}
  We fix $s,t \in M$ for the proof. It is immediate by definition that if $s \canec t$, then we have $s \canoc t$. We prove the converse implication. Assume that $s \canoc t$. For $F \subseteq M$ such that $\alpha\inv(F) \in \Cs$, we need to show that $s \in F \Leftrightarrow t \in F$. The implication $s \in F \Rightarrow t \in F$ is immediate from $s \canoc t$. For the converse direction, since $\alpha\inv(F) \in \Cs$ and \Cs is a Boolean algebra, we have $\alpha\inv(M \setminus F) = A^* \setminus \alpha\inv(F) \in \Cs$. Hence, $s \canoc t$ yields $s \in M \setminus F \Rightarrow t\in M \setminus F$. The contrapositive exactly says that $t \in F \Rightarrow s \in F$, which concludes the proof.
\end{proof}

Finally, we consider the particular case where \Cs is additionally closed under quotients (\emph{i.e.}, when \Cs is a \pvari). We show that in this case \canoc and \canec are compatible with the multiplication of $M$, provided that the morphism $\alpha: A^* \to M$ is~surjective.

\begin{lem}\label{lem:quot}
  Let \Cs be a \pvari and let $\alpha: A^* \to M$ be a surjective morphism into a finite monoid. The two following properties hold for every $s_1,t_1,s_2,t_2 \in M$:
  \begin{itemize}
	\item if $s_1 \canoc t_1$ and $s_2 \canoc t_2$, then $s_1s_2 \canoc t_1t_2$.
	\item if $s_1 \canec t_1$ and $s_2 \canec t_2$, then $s_1s_2 \canec t_1t_2$.
  \end{itemize}
\end{lem}

\begin{proof}
  Recall that by definition, $s \canec t$ if and only if $s \canoc t$ and $t \canoc s$ for every $s,t \in M$. Hence, it suffices to prove the first assertion as the second is an immediate consequence. We fix $s_1,t_1,s_2,t_2 \in M$ such that $s_1 \canoc t_1$ and $s_2\canoc t_2$. We prove that $s_1s_2 \canoc t_1t_2$. Let $F \subseteq M$ such that $\alpha\inv(F) \in \Cs$ and $s_1s_2 \in F$. We show that $t_1t_2 \in F$. Let $u,v \in A^*$ be two words such that $\alpha(u) = s_1$ and $\alpha(v) = t_2$ (this is where we use the hypothesis that $\alpha$ is surjective). Since $s_1s_2 \in F$, we have $s_2 \in s_1\inv F$. Moreover, $\alpha\inv(s_1\inv F) = u\inv \alpha\inv(F)$ belongs to \Cs by closure under quotients. Since $s_2 \canoc t_2$, we get $t_2 \in s_1\inv F$ . It follows that $s_1t_2 \in F$, which means that $s_1 \in Ft_2\inv$. Moreover, $\alpha\inv(Ft_2\inv) =  \alpha\inv(F)v\inv$ belongs to \Cs by closure under quotients. Since $s_1 \canoc t_1$, we get $t_1 \in Ft_2\inv $. Altogether, we obtain $t_1t_2 \in F$, which concludes the proof.
\end{proof}

\subparagraph{Connection with \Cs-morphisms.} When \Cs is a \pvari, it follows from Lemma~\ref{lem:quot} that the canonical \Cs-preorder \canoc of a surjective morphism $\alpha: A^* \to M$ is a \emph{precongruence} on the monoid $M$ and the induced equivalence \canec is a \emph{congruence}. Consequently, the pair $({M}/{\canec},\qanoc)$ is an ordered monoid (recall that \qanoc denotes the order induced by \canoc on the quotient set $M/{\canec}$). Moreover, the map $\ctype{\cdot}: M \to M/{\canec}$ (which associates its \canec-class to every element in $M$) is a morphism. In the following lemma, we prove that the morphism $\ctype{\cdot} \circ \alpha: A^*\to ({M}/{\canec},\qanoc)$ is a \Cs-morphism.

\begin{lem}\label{lem:smult}
  Let \Cs be a \pvari and let $\alpha: A^* \to M$ be a surjective morphism into a finite monoid. The languages recognized by $\ctype{\cdot}  \circ \alpha$ are exactly those in \Cs which are recognized by $\alpha$. In particular, $\ctype{\cdot} \circ \alpha: A^*\to ({M}/{\canec},\qanoc)$ is a \Cs-morphism.
\end{lem}

\begin{proof}
  It is immediate from the definitions that the languages recognized by $\ctype{\cdot}  \circ \alpha$ are exactly those of the form $\alpha\inv(F)$ where $F \subseteq M$ is an upper set for \canoc. Hence, Lemma~\ref{lem:lcanoc} implies that these are exactly the languages in \Cs which are recognized by $\alpha$, concluding the proof.
\end{proof}

\section{Algebraic characterizations for polynomial closure}
\label{sec:caracpol}
We present an algebraic characterization for the classes built with polynomial closure.  The statements are based on Proposition~\ref{prop:synmemb}: given a class \Ds, rather than directly characterizing the languages in \Ds, we characterize the \Ds-morphisms. Presenting the results in this form is more convenient for reusing them later. We first consider the classes \pol{\Cs} when \Cs is a \pvari. We recall their algebraic characterization proved in~\cite{PZ:generic_csr_tocs:18}. Then, we use it to prove generic characterizations of the classes \capol{\Cs} and \capol{\bpol{\Cs}} when \Cs is a \pvari.

\subsection{Polynomial closure}

We first present the generic characterization of the \pol{\Cs}-morphisms proved in~\cite{PZ:generic_csr_tocs:18}.  The statement is based on the \Cs-pair relation defined in Section~\ref{sec:cano}. Since classes of the form \pol{\Cs} are \emph{not} closed under complement in general, it is important to consider morphisms into finite \emph{ordered} monoids here.

\begin{theorem}[\cite{PZ:generic_csr_tocs:18}]\label{thm:polcarac}
  Let \Cs be a \pvari and let $\alpha: A^* \to (M,\leq)$ be a surjective morphism into a finite ordered monoid. The following properties are equivalent:
  \begin{enumerate}
    \item $\alpha$ is a \pol{\Cs}-morphism.
    \item $\alpha$ satisfies the following property:
          \begin{equation}\label{eq:polc}
            s^{\omega+1} \leq s^\omega ts^\omega \quad \text{for every \Cs-pair $(s,t) \in M^2$.}
          \end{equation}
  \end{enumerate}
\end{theorem}

Given a \pvari \Cs, it follows from Lemma~\ref{lem:septopairs} that the \Cs-pairs associated to a morphism can be computed provided that \Cs-separation is decidable. Hence, Theorem~\ref{thm:polcarac} implies that in this case, one can decide whether an input morphism $\alpha: A^* \to (M,\leq)$ is a \pol{\Cs}-morphism. In view of Proposition~\ref{prop:synmemb}, this implies that \pol{\Cs}-membership is decidable. Indeed, given as input a regular language $L$, one can compute its syntactic morphism and since \pol{\Cs} is a \pvari by Theorem~\ref{thm:polc}, the proposition yields that $L \in \pol{\Cs}$ if and only if this syntactic morphism is a \pol{\Cs}-morphism. Altogether, we obtain the following corollary.

\begin{cor}\label{cor:polcarac}
  Let \Cs be a \pvari with decidable separation. Then \pol{\Cs}-membership is decidable.
\end{cor}

\begin{rem}
  It is also known that \pol{\Cs}-covering and \pol{\Cs}-separation are decidable when \Cs is a finite \vari~\cite{pseps3j}. We do not detail this result as we shall not need it.
\end{rem}

\subsection{Intersected polynomial closure} With the generic characterization of \pol{\Cs} in hand, it is straightforward to deduce another generic characterization for \capol{\Cs}. In particular, it shows that membership for \capol{\Cs} boils down to separation for~\Cs (when \Cs is a \pvari). However, it turns out that this result can be strengthened: for any \pvari \Cs, membership for \capol{\Cs} reduces to \emph{membership} for \Cs. This is a simple consequence of Theorem~\ref{thm:polcarac}. It was first observed by Almeida, Bartonov{\'{a}}, Kl{\'{\i}}ma and Kunc~\cite{AlmeidaBKK15}. Indeed, we have the following theorem. It is based on the canonical preorder \canoc associated to a morphism $\alpha$ (note that since classes of the form \capol{\Cs} are \varis, it suffices to consider morphisms into unordered monoids here, see Remark~\ref{rem:cmbool}).

\begin{theorem}\label{thm:pcopolcarac}
  Let \Cs be a \pvari and let $\alpha: A^* \to M$ be a surjective morphism into a finite monoid. The following properties are equivalent:
  \begin{enumerate}
    \item $\alpha$ is a $(\capol{\Cs})$-morphism.
    \item $\alpha$ satisfies the following property:
          \begin{equation}\label{eq:capolc}
            s^{\omega+1} = s^\omega ts^\omega \quad \text{for every \Cs-pair $(s,t) \in M^2$.}
          \end{equation}
    \item $\alpha$ satisfies the following property:
          \begin{equation}\label{eq:capolc2}
            s^{\omega+1} = s^\omega ts^\omega \quad \text{for every $s,t \in M$ such that $s \canoc t$.}
          \end{equation}
  \end{enumerate}
\end{theorem}

\begin{proof}
  We first prove $(1) \Leftrightarrow (2)$. Recall that we view $\alpha$ as the morphism $\alpha: A^* \to (M,=)$ into the ordered monoid $(M,=)$. Therefore, we get from Theorem~\ref{thm:polcarac} that $\alpha$ is a \pol{\Cs}-morphism if and only if~\eqref{eq:capolc} is satisfied. It remains to prove that $\alpha$ is a $(\capol{\Cs})$-morphism if and only if it is a \pol{\Cs}-morphism. The left to right implication is immediate. For the converse one, assume that $\alpha$ is a \pol{\Cs}-morphism and let $F \subseteq M$. We show that $\alpha\inv(F) \in \capol{\Cs}$. By hypothesis, we know that $\alpha\inv(F) \in \pol{\Cs}$ and $A^* \setminus \alpha\inv(F) = \alpha\inv(M \setminus F) \in \pol{\Cs}$. Hence, $\alpha\inv(F) \in \capol{\Cs}$, as desired.

  The implication $(3) \Rightarrow (2)$ is immediate from Lemma~\ref{lem:transclos}, which entails that for every \Cs-pair $(s,t) \in M^2$, we have $s \canoc t$. It remains to show the implication $(2) \Rightarrow (3)$. Assuming that $\alpha$ satisfies~\eqref{eq:capolc}, we show it satisfies~\eqref{eq:capolc2} as well. Let $s,t \in M^2$ such that $s \canoc t$. We show that $s^{\omega+1} = s^{\omega}ts^{\omega}$. Lemma~\ref{lem:transclos} yields $k \in \nat$ and $r_0,\dots,r_k \in M$ such that $r_0 =s$, $r_k = t$ and $(r_i,r_{i+1})$ is a \Cs-pair for all $i < k$. Using induction, we show that for all $i \geq 1$, $s^{\omega+1} = s^{\omega}r_is^\omega$. The case $i = k$ then yields the desired result, as $t = r_k$. When $i = 0$, it is immediate that $s^{\omega+1} = s^{\omega}r_0s^{\omega}$ since $r_0 = s$. We now assume that $i \geq 1$. Using induction, we get that $s^{\omega+1} = s^{\omega}r_{i-1}s^{\omega}$. Therefore, we obtain $s^{\omega} = (s^{\omega+1})^\omega  = (s^{\omega}r_{i-1}s^{\omega})^\omega$. Since $(r_{i-1},r_i)$ is a \Cs-pair, it follows from Lemma~\ref{lem:mult} that $(s^{\omega}r_{i-1}s^{\omega},s^{\omega}r_{i}s^{\omega})$ is a \Cs-pair as well. Thus, it follows from~\eqref{eq:capolc} that,
  \[
    (s^{\omega}r_{i-1}s^{\omega})^{\omega+1} = (s^{\omega}r_{i-1}s^{\omega})^{\omega}s^{\omega}r_{i}s^{\omega}(s^{\omega}r_{i-1}s^{\omega})^{\omega}.
  \]
  Since $s^{\omega+1} = s^{\omega}r_{i-1}s^{\omega}$ and $s^{\omega} = (s^{\omega}r_{i-1}s^{\omega})^\omega$, this yields,
  \[
    s^{\omega+1} = (s^{\omega+1})^{\omega+1} = s^{\omega}s^{\omega}r_{i}s^{\omega}s^{\omega} = s^{\omega}r_{i}s^{\omega}.
  \]
  This concludes the proof.
\end{proof}

Theorem~\ref{thm:pcopolcarac} has the announced consequence: if \Cs is a \pvari with decidable \emph{membership}, then $(\capol{\Cs})$-membership is decidable as well. Indeed, given as input a regular language $L$, we can compute its syntactic morphism $\alpha: A^* \to M$ together with the relation \canoc on $M$ (this is possible by Lemma~\ref{lem:membtopairs} since \Cs-membership is decidable). Moreover, we get from Proposition~\ref{prop:synmemb} that $L \in \capol{\Cs}$ if and only if $\alpha$ is a $(\capol{\Cs})$-morphism since \capol{\Cs} is a \vari by Corollary~\ref{cor:polc}. This property can be decided using~\eqref{eq:capolc2} in Theorem~\ref{thm:pcopolcarac} since we have \canoc in~hand.

\begin{cor}\label{cor:pcopolcarac}
  Let \Cs be a \pvari with decidable membership. Then, membership is decidable for \capol{\Cs} as well.
\end{cor}

This is surprising: \Cs-membership suffices to handle $(\capol{\Cs})$-membership while we are only able to handle \pol{\Cs}-membership when \Cs-separation is decidable. Intuitively, this means that when a regular language $L$ belongs to \capol{\Cs}, the basic languages in \Cs which are required to construct $L$ are all recognized by its syntactic morphism. On the other hand, this is not the case for \pol{\Cs}.

\medskip

We turn to the classes \capol{\bpol{\Cs}}. As seen in Theorem~\ref{thm:capoldel2}, they are interesting because they correspond to the levels $\dec{2}(\infsigc)$ in quantifier alternation hierarchies. We present a specialized characterization of them, which we prove as a corollary of Theorem~\ref{thm:polcarac} and Theorem~\ref{thm:pcopolcarac}. We start with a preliminary result. A key point is that we do not consider \bpolo in the proof: we bypass it using Lemma~\ref{lem:elimbool}, which implies that $\capol{\bpol{\Cs}}$ and $\capol{\copol{\Cs}}$ are the same class. Consequently, in order to apply Theorem~\ref{thm:pcopolcarac}, we need a description of the canonical preorder \canpol associated to a morphism. We use Theorem~\ref{thm:polcarac} to obtain one.

\begin{lem}\label{lem:ubp}
  Let \Cs be a \pvari and let $\alpha: A^* \to M$ be a surjective morphism into a finite monoid. The relation \canpol on $M$ is the least preorder such that for every $x,y,s \in M$ and $e \in E(M)$, if $(e,s) \in M^2$ is a \Cs-pair, then $xesey \canpol xey$.
\end{lem}

\begin{proof}
  One can check from the definition that \canpol is the reverse of the canonical preorder \canppol associated to $\alpha$. Therefore, it suffices prove that \canppol is the least preorder such that for every $x,y,s \in M$ and $e \in E(M)$, if $(e,s) \in M^2$ is a \Cs-pair, then $xey \canppol xesey$.

  By hypothesis on \Cs, Theorem~\ref{thm:polc} implies that \pol{\Cs} is a \pvari of regular languages. We first prove that for all $x,y,s \in M$ and $e \in E(M)$, if $(e,s) \in M^2$ is a \Cs-pair, then $xey \canppol xesey$. Let $f = \pctype{e}$ and $t = \pctype{s}$. Since $\alpha$ is surjective, Lemma~\ref{lem:smult} yields that the map $\pctype{\cdot} \circ \alpha: A^* \to ({M}/{\canep},\qanop)$ is a \pol{\Cs}-morphism. Moreover, since $e$ is an idempotent of $M$, it is immediate that $f$ is an idempotent of ${M}/{\canep}$. Finally, since $(e,s)$ is a \Cs-pair for $\alpha$, one can verify that $(f,t) = (\pctype{e},\pctype{s})$ is a \Cs-pair for $\pctype{\cdot}\circ\alpha$. Consequently, Theorem~\ref{thm:polcarac} yields $f^{\omega+1} \qanop f^\omega t f^\omega$, \emph{i.e.}, $f \qanop ftf$ since $f$ is idempotent. This exactly says that $\pctype{e} \qanop \pctype{ese}$. Therefore, we also have $\pctype{xey} \qanop \pctype{xesey}$ and we conclude that $xey \canppol xesey$, as desired.

  Conversely, let $\preccurlyeq$ be the least preorder on $M$ such that for every $x,y,s \in M$ and $e \in E(M)$, if $(e,s) \in M^2$ is a \Cs-pair, then $xey \preccurlyeq xesey$. We show that ${\canppol} \subseteq {\preccurlyeq}$. Let~$\cong$ be equivalence on $M$ induced by $\preccurlyeq$ and let $\leqslant$ be the ordering induced by $\preccurlyeq$ on the quotient set ${M}/{\cong}$. It is straightforward that $\preccurlyeq$ is compatible with multiplication (this is where the elements $x,y \in M$ in the definition are important). Hence, the pair $({M}/{\cong},\leqslant)$ is an ordered monoid and the map $s \mapsto [s]_{\cong}$, sending every element $s \in M$ to its  $\cong$-class, is a~morphism.

  We use Theorem~\ref{thm:polcarac} to prove that $[\cdot]_{\cong} \circ \alpha: A^* \to ({M}/{\cong},\leqslant)$ is a \pol{\Cs}-morphism. We have to show that given $r,t \in {M}/{\cong}$ such that $(r,t)$ is a \Cs-pair for $[\cdot]_{\cong} \circ \alpha$, we have $r^{\omega+1} \leqslant r^{\omega} t r^{\omega}$. Let $(r,t)$ be a \Cs-pair for $[\cdot]_{\cong} \circ \alpha$. One can verify that there exists a \Cs-pair $(q,s) \in M^2$ for $\alpha$ such that $[q]_{\cong} = r$ and $[s]_{\cong} = t$. Let $k = \omega(M)$. We know from Lemma~\ref{lem:mult} that $(q^k,s(q)^{k-1})$ is also a \Cs-pair for $\alpha$. Since $q^k \in E(M)$ by definition of $k$, we have $q^k \preccurlyeq q^k s q^{k-1} q^k$ by definition of $\preccurlyeq$ which yields $q^{\omega+1} \preccurlyeq q^\omega sq^\omega$ since $k = \omega(M)$. It follows that $[q^{\omega+1}]_{\cong} \leqslant [q^{\omega} s q^{\omega}]_{\cong}$. Therefore, since $[q]_{\cong} = r$ and $[s]_{\cong} = t$, we obtain $r^{\omega+1} \leqslant r^{\omega} t r^{\omega}$, as desired, and we conclude that $[\cdot]_{\cong} \circ \alpha: A^* \to ({M}/{\cong},\leqslant)$ is a \mbox{\pol{\Cs}-morphism}.

  We are ready to prove that ${\canppol} \subseteq {\preccurlyeq}$. Let $(q,r) \in M$ such that $q \canppol r$. We show that $q \preccurlyeq r$. Let $F = \{p \in M \mid q \preccurlyeq p\}$. By definition, $\alpha\inv(F)$ is recognized by the \pol{\Cs}-morphism $[\cdot]_{\cong} \circ \alpha: A^* \to ({M}/{\cong},\leqslant)$. Hence, we have $\alpha\inv(F) \in \pol{\Cs}$ and since we have $q \in \alpha\inv(F)$ and $q \canppol r$, we obtain $r \in F$ by definition of \canppol. This exactly says that $q \preccurlyeq r$ by definition of $F$, which concludes the proof.
\end{proof}

We are ready to present and prove the generic algebraic characterization of the classes \capol{\bpol{\Cs}}.

\begin{theorem}\label{thm:ubp}
  Let \Cs be a \pvari and let $\alpha:A^* \to M$ be a surjective morphism into a finite monoid. The two following properties are equivalent:
  \begin{enumerate}
    \item $\alpha$ is a $(\capol{\bpol{\Cs}})$-morphism.
    \item $\alpha$ satisfies the following property:
          \begin{equation}\label{eq:cabpol}
            \begin{array}{c}
              (eset)^{\omega+1} = (eset)^{\omega}et(eset)^{\omega} \\
              \text{for every $s,t \in M$ and $e \in E(M)$ such that $(e,s) \in M^2$ is a \Cs-pair}.
            \end{array}
          \end{equation}
  \end{enumerate}
\end{theorem}

\begin{proof}
  First, recall that by Lemma~\ref{lem:elimbool}, we have the equality $\pol{\bpol{\Cs}} = \pol{\copol{\Cs}}$. Thus, the first assertion in the theorem states that $\alpha$ is a $(\capol{\copol{\Cs}})$-morphism. As \copol{\Cs} is a \pvari by Corollary~\ref{cor:polc}, it follows from Theorem~\ref{thm:pcopolcarac} that this property holds if and only if $\alpha$ satisfies the following condition:
  \begin{equation}\label{eq:bproof}
    q^{\omega+1} = q^{\omega}rq^{\omega} \quad \text{for every $q,r \in M$ such that $q \canpol r$}.
  \end{equation}
  Consequently, it suffices to prove that $\alpha$ satisfies~\eqref{eq:bproof} if and only if it satisfies~\eqref{eq:cabpol}.

  We first assume that $\alpha$ satisfies~\eqref{eq:bproof}. We prove that it satisfies~\eqref{eq:cabpol} as well. Given $s,t \in M$ and $e \in E(M)$ such that $(e,s) \in M^2$ is a \Cs-pair, we have to prove that $(eset)^{\omega+1} = (eset)^{\omega}et(eset)^{\omega}$. It is immediate from Lemma~\ref{lem:ubp} that $ese \canpol e$, which yields $eset \canpol et$. Hence,~\eqref{eq:bproof} implies that $(eset)^{\omega+1} = (eset)^{\omega}et(eset)^{\omega}$, as desired.

  Conversely, we assume that $\alpha$ satisfies~\eqref{eq:cabpol} and prove that it satisfies~\eqref{eq:bproof} as well. Given $q,r \in M$ such that $q \canpol r$, we have to prove that $q^{\omega+1} = q^{\omega}rq^{\omega}$. By Lemma~\ref{lem:ubp}, there exists $q_0,\dots,q_n \in M$ such that $q =q_0$, $r = q_n$ and, for every $i \leq n$, there exist $x,y,s \in M$ and $e \in E(M)$ such that $(e,s) \in M^2$ is a \Cs-pair, $q_{i-1} = xesey$ and $q_{i} = xey$. We use induction on $i$ to prove that $q^{\omega+1} = q^{\omega}q_iq^{\omega}$ for every $i \leq n$. Since $q_n = r$, the case $i = n$ yields the desired result. When $i = 0$, it is immediate that $q^{\omega+1} = q^{\omega}q_0q^{\omega}$ since $q_0 = q$. Assume now that $i \geq 1$. By induction hypothesis, we have $q^{\omega+1} = q^{\omega} q_{i-1} q^{\omega}$, which implies that $q^{\omega+2} = (q^{\omega} q_{i-1} q^{\omega})^{\omega+2}$. Moreover, by definition, we have $x,y,s \in M$ and $e \in E(M)$ such that $(e,s) \in M^2$ is a \Cs-pair, $q_{i-1} = xesey$ and $q_{i} = xey$. It follows that,
  \[
    \begin{array}{lll}
      q^{\omega+2} & = &  (q^{\omega} xesey q^{\omega})^{\omega+2} \\
                   & = & q^{\omega} x (esey q^{\omega}q^\omega x)^{\omega+1} esey q^{\omega}.
    \end{array}
  \]
  Since $(e,s) \in M^2$ is a \Cs-pair, it now follows from~\eqref{eq:cabpol} applied with $t=yq^{\omega}q^{\omega}x$ that,
  \[
    \begin{array}{lll}
      q^{\omega+2} & = & q^{\omega} x (esey q^{\omega}q^\omega x)^{\omega} ey q^{\omega}q^\omega x (esey q^{\omega}q^\omega x)^{\omega} esey q^{\omega} \\
                   & = & (q^{\omega} xesey q^{\omega})^{\omega} q^\omega xey q^{\omega} (q^\omega xesey q^{\omega})^{\omega+1} \\
                   & = & (q^{\omega} q_{i-1} q^{\omega})^{\omega} q^\omega q_i q^{\omega} (q^\omega q_{i-1} q^{\omega})^{\omega+1} \\
                   & = & q^\omega q_i q^{\omega+1}.
    \end{array}
  \]
  Multiplying by $q^{\omega-1}$ on the right yields $q^{\omega+1} = q^{\omega}q_{i}q^{\omega}$, which concludes the proof.
\end{proof}

Theorem~\ref{thm:ubp} implies that membership is decidable for \capol{\bpol{\Cs}} when \Cs is a \pvari with decidable separation. The argument is based on Proposition~\ref{prop:synmemb}. Note that we already had this result. Indeed, since \capol{\bpol{\Cs}} is exactly \capol{\copol{\Cs}} by Lemma~\ref{lem:elimbool}, it also follows from Corollary~\ref{cor:polcarac} and Corollary~\ref{cor:pcopolcarac} (clearly, membership for \copol{\Cs} reduces to the same problem for \pol{\Cs}). Yet, we shall need Theorem~\ref{thm:ubp} when proving the logical characterizations of \upolo in~Section~\ref{sec:logcar}.

\begin{cor}\label{cor:ubp}
  Let \Cs be a \pvari with decidable separation. Then, membership is decidable for \capol{\bpol{\Cs}}.
\end{cor}

\section{Algebraic characterizations for unambiguous polynomial closure}
\label{sec:caracupol}
We present the generic algebraic characterization of \emph{unambiguous} polynomial closure. It holds for every input \emph{\vari} \Cs and implies that \upol{\Cs}-membership reduces to the same problem for \Cs. We also use this characterization to complete the proofs of Theorem~\ref{thm:polcopol} and Theorem~\ref{thm:apol}: when \Cs is a \vari, we have $\upol{\Cs} = \adet{\Cs} = \capol{\Cs}$.

\begin{rem}
There exists an independent characterization by Pin~\cite{Pin80}. Yet, it involves algebraic notions that we have not presented and does not yield the reduction to \Cs-membership.
\end{rem}

Again, the characterization is based on the canonical equivalence \canec associated to morphisms into finite monoids (see Section~\ref{sec:cano} for the definition).

\begin{theorem}\label{thm:caracupol}
  Let \Cs be a \vari and let $\alpha: A^* \to M$ be a surjective morphism into a finite monoid. The following properties are equivalent:
  \begin{enumerate}
    \item $\alpha$ is a \upol{\Cs}-morphism.
    \item $\alpha$ is a \adet{\Cs}-morphism.
    \item $\alpha$ is a \wadet{\Cs}-morphism.
    \item $\alpha$ satisfies the following property:
    \begin{equation*}
      s^{\omega+1} = s^{\omega}ts^{\omega} \quad \text{for every \Cs-pair $(s,t) \in M^2$.} \tag{\ref{eq:capolc}}
    \end{equation*}
    \item $\alpha$ satisfies the following property:
    \begin{equation}
      s^{\omega+1} = s^{\omega}ts^{\omega} \quad \text{for every $s,t \in M$ such that $s \canec t$}. \label{eq:cupol}
    \end{equation}
  \end{enumerate}
\end{theorem}

Note that by Theorem~\ref{thm:pcopolcarac}, Equation~\eqref{eq:capolc} characterizes $(\capol{\Cs})$-morphisms.
Before proving Theorem~\ref{thm:caracupol}, let us discuss its consequences. Importantly, we cannot directly apply it to reduce \upol{\Cs}-membership to \Cs-membership. Indeed, given an arbitrary class \Ds, the connection between \Ds-membership and \Ds-morphisms is based on Proposition~\ref{prop:synmemb}, which requires \Ds to be a (positive) \vari. Hence, we need Theorem~\ref{thm:comp}: \upol{\Cs} is a \vari when \Cs is one. Yet, we cannot use Theorem~\ref{thm:comp} at this point,~since we obtained it as the corollary of a result that we have not proved yet: Theorem~\ref{thm:polcopol} (which states that $\upol{\Cs} = \capol{\Cs}$ when \Cs is a \vari). Therefore, we first use Theorem~\ref{thm:caracupol} to provide a proof of Theorem~\ref{thm:polcopol}. In fact, we present a more general statement involving also \wadet{\Cs} and \adet{\Cs}. It combines Theorem~\ref{thm:polcopol} and Theorem~\ref{thm:apol}.

\begin{cor}\label{cor:cequiv}
  For any \vari \Cs, we have the following equalities:
  \[\wadet{\Cs} = \adet{\Cs} = \upol{\Cs} = \capol{\Cs}.\]
\end{cor}

\begin{proof}
  The inclusions $\wadet{\Cs} \subseteq \adet{\Cs} \subseteq\upol{\Cs}$ are immediate from Lemma~\ref{lem:detuna}: left/right deterministic marked concatenations are unambiguous. Hence, we need to prove that $\upol{\Cs} \subseteq \wadet{\Cs}$, $\upol{\Cs} \subseteq \capol{\Cs}$ and $\capol{\Cs} \subseteq \upol{\Cs}$. We start with the latter. Consider a language $L \in \capol{\Cs}$ and let $\alpha: A^* \to M$ be the syntactic morphism of $L$. Since \capol{\Cs} is a \vari by Corollary~\ref{cor:polc}, it follows from Proposition~\ref{prop:synmemb} that $\alpha$ is a $(\capol{\Cs})$-morphism. It then follows from Theorem~\ref{thm:pcopolcarac} that $\alpha$ satisfies~\eqref{eq:capolc}. Consequently, Theorem~\ref{thm:caracupol} yields that $\alpha$ is a \upol{\Cs}-morphism, and $L \in \upol{\Cs}$ since it is recognized by its syntactic morphism.

  We turn to the inclusions $\upol{\Cs} \subseteq \wadet{\Cs}$ and $\upol{\Cs} \subseteq \capol{\Cs}$,  which we treat simultaneously. Let $L\in\upol{\Cs}$ and $\alpha: A^* \to M$ be the syntactic~morphism of $L$. We prove that $\alpha$ satisfies~\eqref{eq:capolc}. It will follow from Theorem~\ref{thm:pcopolcarac}  that $\alpha$ is a \capol{\Cs}-morphism and from Theorem~\ref{thm:caracupol} that it is a \wadet{\Cs}-morphism. Therefore, since $L$ is recognized by its syntactic morphism, we get $L \in \capol{\Cs}$ and $L \in \wadet{\Cs}$.

  In order to prove~\eqref{eq:capolc}, let $(s,t)\in M^2$ be a \Cs-pair for $\alpha$. We have to show that $s^{\omega+1} = s^\omega t s^\omega$. Since $L \in \upol{\Cs}$, Proposition~\ref{prop:carprop} yields a \Cs-morphism $\eta: A^* \to N$ and $k \in \nat$ such that for all words $u,v,v',x,y \in A^*$ satisfying $\eta(u) = \eta(v) = \eta(v')$, we have $xu^{k}vu^{k}y \in L \Leftrightarrow xu^{k} v' u^{k} y \in L$. Since $(s,t)$ is a \Cs-pair and $\eta$ is a \Cs-morphism, it follows from Lemma~\ref{lem:cmorph} that there exist $u,v \in A^*$ such that $\eta(u) = \eta(v)$, $\alpha(u) = s$ and $\alpha(v) = t$. By definition of $\eta$, it follows that for every $x,y \in A^*$, we have $xu^{k}vu^{k}y \in L \Leftrightarrow xu^{k} u u^{k} y \in L$. This exactly says that $u^k vu ^k$ and $u^k u u^k$ are equivalent for the syntactic congruence of $L$. Hence, they have the same image under $\alpha$ by definition of the syntactic morphism, and we get $s^k t s^k = s^k s s^k$. It now suffices to multiply by the appropriate number of copies of $s$ on the left and on the right to get  $s^{\omega+1} = s^\omega t s^\omega$, which completes the proof.
\end{proof}

\begin{rem}
  There is a subtle difference between Theorem~\ref{thm:caracupol} above and Theorem~\ref{thm:pcopolcarac} (which applies to \capol{\Cs}), although the algebraic characterizations are the same: the latter holds for \emph{every \pvari} \Cs, while the former requires \Cs to be a \emph{\vari}. This is important, as Corollary~\ref{cor:cequiv} fails when \Cs is only a \pvari (see Remark~\ref{rem:eqfails}).
\end{rem}

With Corollary~\ref{cor:cequiv} in hand, we may now use Theorem~\ref{thm:comp}: if \Cs is a \vari, then so is \upol{\Cs}. Consequently, in this case, it follows from Proposition~\ref{prop:synmemb} that deciding whether an input regular language $L$ belongs to \upol{\Cs} amounts to testing whether its syntactic morphism is a \upol{\Cs}-morphism. The syntactic morphism of a regular language can be computed. Moreover, Theorem~\ref{thm:caracupol} states that an arbitrary surjective morphism is a \upol{\Cs}-morphism if and only if it satisfies~\eqref{eq:cupol}. This can be checked by testing every possible combination provided that we have the equivalence \canec in hand. Finally, by Lemma~\ref{lem:membtopairs}, \canec can be computed as soon as \Cs-membership is decidable. Altogether, we obtain that \upolo preserves the decidability of membership when it is applied to a \vari.

\begin{cor}\label{cor:upol:upolreduc}
  Let \Cs be a \vari. Assume that \Cs has decidable membership. Then, \upol{\Cs}-member\-ship is decidable as well.
\end{cor}

\noindent
It remains to establish Theorem~\ref{thm:caracupol}. We devote the rest of this section to this proof.

\begin{proof}[Proof of Theorem~\ref{thm:caracupol}]
  We fix a \vari \Cs and a surjective morphism $\alpha: A^* \to M$. We prove $(3) \Rightarrow (2) \Rightarrow (1) \Rightarrow (4) \Rightarrow (5) \Rightarrow (3)$. The implications $(3) \Rightarrow (2)$ and $(2) \Rightarrow (1)$ are clear: we have the inclusions $\wadet{\Cs} \subseteq \adet{\Cs} \subseteq \upol{\Cs}$ since left/right deterministic marked concatenations are necessarily unambiguous by Lemma~\ref{lem:detuna}.

  \smallskip

  We show $(1) \Rightarrow (4)$: let $\alpha$ be a \upol{\Cs}-morphism, we prove that it~satisfies~\eqref{eq:capolc}. Given a \Cs-pair $(s,t) \in M^2$, we show that $s^{\omega+1} = s^\omega t s^\omega$. Let $q = s^{\omega}$. By hypothesis, we have $\alpha\inv(q) \in \upol{\Cs}$. Hence, Proposition~\ref{prop:carprop} yields a \Cs-morphism $\eta: A^* \to N$ and $k \in \nat$ such that for every $u,v,v' \in A^*$, if $\eta(u) = \eta(v) = \eta(v')$, then $u^k v u^k \in \alpha\inv(q) \Leftrightarrow u^k v' u^k \in \alpha\inv(q)$. Let $p = \omega(M)$. Since $(s,t) \in M^2$ is a \Cs-pair, it follows from Lemma~\ref{lem:mult} that $(s^p,ts^{p-1})$ is also a \Cs-pair. Hence, since $\eta$ is a \Cs-morphism, Lemma~\ref{lem:cmorph} yields $u,v \in A^*$ such that $\eta(u) = \eta(v)$, $\alpha(u) = s^p$ and $\alpha(v) = ts^{p-1}$. Clearly, $\alpha(u^k uu^k) = q^{kp}q^pq^{kp} = q = s^\omega$ since $p = \omega(M)$. Hence $u^k uu^k \in \alpha\inv(q)$ and since $\eta(u) = \eta(v)$, the definition of $\eta$ yields $u^k vu^k \in \alpha\inv(q)$. We get $\alpha(u^k v u^k) = q = s^{\omega}$. By definition of $u$ and $v$, this exactly says that $s^{\omega}ts^{2\omega -1} = s^{\omega}$. It now suffices to multiply by $s$ on the right to get $s^{\omega+1} = s^\omega t s^\omega$, as~desired.

  \smallskip

  We now prove $(4) \Rightarrow (5)$. Assume that $\alpha$ is surjective and satisfies~\eqref{eq:capolc}, we need to prove that it satisfies~\eqref{eq:cupol} as well.  Since \Cs is a \vari, we know that ${\canec} = {\canoc}$ by Lemma~\ref{lem:bcanoc}. Hence, the property is immediate from the implication $(2) \Rightarrow (3)$ in  Theorem~\ref{thm:pcopolcarac}.

  \smallskip

  We proved that $(3) \Rightarrow (2) \Rightarrow (1) \Rightarrow (4) \Rightarrow (5)$. Therefore, it remains to prove $(5) \Rightarrow (3)$. Assuming that $\alpha$ satisfies~\eqref{eq:cupol}, we show that it is a \wadet{\Cs}-morphism, \emph{i.e.}, that every language recognized by $\alpha$ belongs to \wadet{\Cs}. Let $N = M/{\canec}$. Recall that $N$ is a monoid since \canec is a congruence by Lemma~\ref{lem:quot}. We define $\eta$ as the map $\eta = \ctype{\cdot} \circ \alpha: A^* \to N$. We know from Lemma~\ref{lem:smult} that $\eta$ is a \Cs-morphism. Given a finite set of languages \Kb and $s,t \in M$, we say that \Kb is \emph{$(s,t)$-safe} if for all $K \in \Kb$ and $w,w' \in K$, we have $s\alpha(w)t = s\alpha(w')t$. Finally, a \emph{\wadet{\Cs}-partition} of a language $H$ is a finite partition of $H$ into languages of \wadet{\Cs}. The implication $(5)\Rightarrow(3)$ follows from the following lemma.

  \begin{lem}\label{lem:upol:adetsuf}
    For any $x \in N$ and $s,t \in M$, there exists an $(s,t)$-safe \wadet{\Cs}-partition of~$\eta\inv(x)$.
  \end{lem}

  We first apply Lemma~\ref{lem:upol:adetsuf} to complete the main argument. For every $x \in N$, Lemma~\ref{lem:upol:adetsuf} yields a \wadet{\Cs}-partition $\Kb_x$ of $\eta\inv(x)$ which is $(1_M,1_M)$-safe. Hence, $\Kb = \bigcup_{x \in N} \Kb_x$ is a \wadet{\Cs}-partition of $A^*$ which is $(1_M,1_M)$-safe. The latter property implies that for every $K \in \Kb$, we have $s \in M$ such that $K \subseteq \alpha\inv(s)$. Since \Kb is a partition of $A^*$, it follows that every language recognized by $\alpha$ is a disjoint finite union of languages in \Kb and therefore belongs to \wadet{\Cs}, by closure under disjoint union. This concludes the main argument.

  \smallskip

  It remains to prove Lemma~\ref{lem:upol:adetsuf}. Let $s,t \in M$ and $x \in N$. We build an  $(s,t)$-safe \wadet{\Cs}-partition \Kb of $\eta\inv(x)$. The proof proceeds by induction on the following three parameters, which depend on Green relations on $M$ and $N$, listed by order of importance:
  \begin{enumerate}
    \item The \emph{rank of $(s,x,t)$}: the number of elements $y \in N$ such that $\ctype{s}x\ctype{t} \Jord y$.
    \item The \emph{\Rrel-index of $s$}: the number of elements $r \in M$ such that $r \Rord s$
    \item The \emph{\Lrel-index of $t$}: the number of elements $r \in M$ such that $r \Lord t$.
  \end{enumerate}
  There are three cases, depending on whether $\ctype{s}x\ctype{t}\Jrel x$ and on two properties of $s,t$ and~$x$:
  \begin{itemize}
    \item We say that $s$ is \emph{right $x$-stable} if there exists $q \in M$ such that $\ctype{q} \Rrel x$ and $sq \Rrel s$.
    \item We say that $t$ is \emph{left $x$-stable} if there exists $r \in M$ such that $\ctype{r} \Lrel x$ and $rt \Lrel t$.
  \end{itemize}
  In the base case, we assume that all three properties hold. Otherwise, we distinguish two inductive cases depending on whether $\ctype{s}x\ctype{t} \Jrel x$ holds or not.

\subparagraph{Base case: $\ctype{s}x\ctype{t} \Jrel x$, $s$ is right $x$-stable and $t$ is left $x$-stable.} In this case, let $\Kb = \{\eta\inv(x)\}$, which is clearly a \wadet{\Cs}-partition of $\eta\inv(x)$: this is even a \emph{\Cs-partition} of $\eta\inv(x)$ since $\eta$ is a \Cs-morphism. It remains to show that \Kb is $(s,t)$-safe. Given $w,w' \in \eta\inv(x)$, show that $s\alpha(w)t = s\alpha(w')t$. We let $p = \alpha(w)$ and $p' = \alpha(w')$. By hypothesis, we have $\ctype{p} = \ctype{p'} = x$ and we have to show that $spt = sp't$. We have $\ctype{s}x\ctype{t} \Jrel x$, which implies that $\ctype{s}x \Jrel x$ and $x\ctype{t} \Jrel x$. Hence, since $\ctype{s}x \Lord x$ and $x\ctype{t} \Rord x$, Lemma~\ref{lem:jlr} yields $\ctype{s}x \Lrel x$ and $x\ctype{t} \Rrel x$. We use this property to prove the following fact.

  \begin{lem}\label{lem:sub1}
    We have $spt \Rrel s$ and $spt \Lrel t$.
  \end{lem}

  \begin{proof}
    We prove that $spt \Rrel s$ using the hypotheses that $x\ctype{t} \Rrel x$ and that $s$ is right $x$-stable (that $spt \Lrel t$ is proved symmetrically using the hypotheses that $\ctype{s}x \Lrel x$ and that $t$ is left $x$-stable). Since $s$ is right $x$-stable, there exists $q \in M$ such that $\ctype{q} \Rrel x$ and $sq \Rrel s$. Since $x = \ctype{p}$ and  $x\ctype{t} \Rrel x$, we get $\ctype{pt} \Rrel \ctype{q}$. We get $r \in M$ such that $\ctype{q} = \ctype{ptr}$. Moreover, since $sq \Rrel s$, we have $q' \in M$ such that $s =sqq'$. Consequently, we get $s = s(qq')^{\omega} = s(qq')^{\omega+1}$. We have $\ctype{qq'} = \ctype{ptrq'}$ which implies that $qq' \canec ptrq'$. Therefore,  Equation~\eqref{eq:cupol} yields $(qq')^{\omega+1} = (qq')^{\omega}ptrq'(qq')^{\omega}$. Finally, we may multiply by $s$ on the left to obtain,
    \[
      s = s(qq')^{\omega+1}  = s(qq')^{\omega}ptrq'(qq')^{\omega} = s ptrq'(qq')^{\omega}.
    \]
    This implies that $s \Rord spt$. Since it is clear that $spt \Rord s$, we get $spt \Rrel s$, as desired.
  \end{proof}

  Lemma~\ref{lem:sub1} yields $q_1,q_2 \in M$ such that $t = q_1spt$ and $s = sptq_2$. Let $r = q_1 spt q_2$. We may combine the above equalities to obtain $s = spr$ (indeed, $r=tq_2$, whence $spr=sptq_2=s$). Therefore, $s = s(pr)^\omega$. Similarly, $t = rpt  = (rp)^{\omega+1} t$. By hypothesis, $\ctype{p} = \ctype{p'} = x$. Hence, we obtain $\ctype{pr} = \ctype{p'r}$, so that $pr \canec p'r$. Therefore, Equation~\eqref{eq:cupol} applies and~yields,
  \[
    (pr)^{2\omega+1} = (pr)^{\omega+1} = (pr)^{\omega} p' r (pr)^{\omega}.
  \]
  We may now multiply by $s$ on the left and by $pt$ on the right to obtain,
  \[
    s (pr)^{\omega} p (rp)^{\omega+1} t = s (pr)^{\omega} p' (rp)^{\omega+1} t.
  \]
  Since we already established that $s = s (pr)^\omega$ and $t =  (rp)^{\omega+1} t$, we get as desired that $s p r = s p' t$, which concludes the proof for this base case.

\subparagraph{First inductive case: $\ctype{s}x\ctype{t} \Jords x$.} In this case, it is immediate that the rank of $(1_M,x,1_M)$ is strictly smaller than the one of $(s,x,t)$. Hence, induction on our first and main parameter in Lemma~\ref{lem:upol:adetsuf} yields a \wadet{\Cs}-partition \Kb of $\eta\inv(x)$ which is $(1_M,1_M)$-safe (and therefore $(s,t)$-safe as well). This concludes the first inductive case.

\subparagraph{Second inductive case: $\ctype{s}x\ctype{t} \Jrel x$ and either $s$ is not right $x$-stable or $t$ is not left $x$-stable.} This involves two symmetrical arguments depending on which property holds. We treat the case where $t$ is not left $x$-stable using induction on the first and third parameter (the second one is used in the symmetrical case). Let $T \subseteq N \times A \times N$ be the set of all triples $(y,a,z) \in N \times A \times N$ such that $x \Lrel \eta(a)z \Lords z$ and $y\eta(a)z = x$. For each triple $(y,a,z) \in T$, we use induction to build auxiliary \wadet{\Cs}-partitions of $\eta\inv(y)$ and $\eta\inv(z)$, which we then combine to construct the desired $(s,t)$-safe \wadet{\Cs}-partition \Kb of $\eta\inv(x)$. We fix $(y,a,z) \in T$ for the definition of these auxiliary \wadet{\Cs}-partitions.

  By definition of $T$, we have $x \Lords z$. By Lemma~\ref{lem:jlr}, this implies $x \Jords z$, so that we also have $\ctype{s}x\ctype{t} \Jords z$. Hence, the rank of $(1_M,z,1_M)$ is strictly smaller than the one of $(s,x,t)$. Since this is our most important induction parameter, we obtain a \wadet{\Cs}-partition $\Vb_z$ of $\eta\inv(z)$  which is $(1_M,1_M)$-safe.

  We now use $\Vb_z$ to build several \wadet{\Cs}-partitions of $\eta\inv(y)$, one for each language $V \in \Vb_z$. We fix a language $V\in\Vb_z$ for the definition. Since \Vb is $(1_M,1_M)$-safe, there exists $r_V \in M$ such that $\alpha(v) = r_V$ for every $v \in V$. Moreover, $\ctype{r_V} = \eta(v)  = z$ since $v \in \eta\inv(z)$. We show that $(s,y,\alpha(a)r_{V}t)$ has a strictly smaller induction parameter than $(s,x,t)$. Since
 $(y,a,z) \in T$, we have $x \Lrel \eta(a)z$, which means that $x \Lrel \ctype{\alpha(a)r_V}$. Since $t$ is not left $x$-stable, it follows that $\alpha(a)r_{V}t \Lords t$: the \Lrel-index of $\alpha(a)r_{V}t$ is strictly smaller than the one of $t$, meaning that our third induction parameter has decreased.  Moreover, since $y\eta(a)z = x$, we also have $\ctype{s}y\ctype{\alpha(a)r_{V}t} = \ctype{s}x\ctype{t}$ which means that $(s,x,t)$ and $(s,y,\alpha(a)r_{V}t)$ have the same rank: our first induction parameter is unchanged. Finally, the second induction parameter remains unchanged as well, since it only depends on $s$. Hence, induction on our third parameter yields a \wadet{\Cs}-partition $\Ub_{(y,a,z),V}$ of $\eta\inv(y)$ which is $(s,\alpha(a)r_{V}t)$-safe.

 \smallskip
 We are ready to build the desired $(s,t)$-safe \wadet{\Cs}-partition \Kb of $\eta\inv(x)$. We define,
  \[
    \Kb = \bigcup_{(y,a,z) \in T}\{UaV \mid V \in \Vb_z \text{ and } U \in \Ub_{(y,a,z),V}\}.
  \]
  It remains to prove that \Kb satisfies the desired properties. First, we show that \Kb is indeed a \wadet{\Cs}-partition of $\eta\inv(x)$. We start by proving that for every $w \in \eta\inv(x)$, there exists a \emph{unique} $K \in \Kb$ such that $w \in K$. Let $v'$ be the least suffix of $w$ such that $\eta(v') \Lrel x$ ($u'$ exists since $w$ is such a suffix). Observe that $v' \neq \veps$. Indeed, otherwise $\ctype{1_M} \Lrel x$ and $1_{M}t=t \Lrel t$, contradicting our hypothesis that $t$ is not left $x$-stable. Thus, $v' = av$ for some $v \in A^*$ and $a \in A$. Finally, we let $u \in A^*$ be such that $w = uav$. Let then $y = \eta(u)$ and $z = \eta(v)$. By definition of $v' = av$, we know that $x \Lrel \eta(a)z \Lords z$ and $y\eta(a)z = x$. Hence, $(y,a,z) \in T$. Moreover, there exists $V \in \Vb_z$ such that $v \in V$ and $U \in \Ub_{(y,a,z),V}$ such that $u \in U$. Thus $w = uav \in UaV$ and $UaV \in \Kb$. Moreover, one can verify from the definition that this is the only language in \Kb containing $w$.

  We now prove that for every $K \in \Kb$, we have $K\subseteq\eta\inv(x)$ and $K \in \wadet{\Cs}$. By definition, $K = UaV$ with $V \in \Vb_y$ and $U \in \Ub_{(y,a,z),V}$ for some triple $(y,a,z)\in T$. Moreover, by definition of $\Vb_y$ and $\Ub_{(y,a,z),V}$, we have $U \subseteq \eta\inv(y)$ and $V \subseteq \eta\inv(z)$. Thus, since we have $y\eta(a)z= x$ by definition of $T$, we get $K = UaV \subseteq \eta\inv(x)$. Moreover, $U,V \in \wadet{\Cs}$. Therefore, it suffices to show that $UaV$ is right \Cs-deterministic to prove that $UaV \in \wadet{\Cs}$ (note that left \Cs-deterministic marked concatenations are used in the symmetrical case). By definition, we have $V \subseteq \eta\inv(y)$ and $U \subseteq \eta\inv(z)$. Hence, since $\eta\inv(y),\eta\inv(z) \in \Cs$ since $\eta$ is a \Cs-morphism, it suffices to prove that $\eta\inv(y)a\eta\inv(z)$ is right deterministic. This is immediate from Lemma~\ref{lem:greendet}: since $(y,a,z) \in T$ we have $\eta(a)z \Lords z$. Altogether, we conclude that \Kb is indeed a \wadet{\Cs}-partition of $\eta\inv(x)$.

  Finally, we show that \Kb is $(s,t)$-safe. Let $K \in \Kb$ and $w,w' \in K$. We have to prove that $s\alpha(w)t = s\alpha(w')t$. By definition $K = UaV$ with $V \in \Vb_z$ and $U \in \Ub_{(y,a,z),V}$ for some triple $(y,a,z) \in T$. Thus, we get $u,u' \in U$ and $v, v'\in V$ such that $w = uav$ and $w' = u'av'$. By definition $v,v' \in \Vb$ imply that $\alpha(v)=\alpha(v')=r_V$. Therefore, we obtain $s\alpha(w)t = s\alpha(u)\alpha(a)r_V t$ and $s\alpha(w')t = s\alpha(u')\alpha(a)r_{V}t$. Finally, recall that $\Ub_V$ is $(s,\alpha(a)r_{V}t)$-safe. Consequently, we get $s\alpha(w)t=s\alpha(u)\alpha(a)r_{V}t = s\alpha(u')\alpha(a)r_V t=s\alpha(w')t$, which concludes the~proof.
\end{proof}

\section{Covering and Separation for unambiguous polynomial closure}
\label{sec:upolcov}
In this section, we look at separation and covering for classes of the form \upol{\Cs}. We prove that both problems are decidable for \upol{\Cs} when the input class \Cs is a \emph{finite} \vari. The algorithm is based on a generic framework introduced in~\cite{pzcovering2} to handle separation and covering.  We first recall this framework and then use it to present the~algorithm.

\subsection{Semirings} A \emph{semiring} is a tuple $(R,+,\cdot)$ where $R$ is a set and ``$+$'' and ``$\cdot$'' are two binary operations called addition and multiplication, which satisfy the following axioms:
\begin{itemize}
  \item $(R,{+})$ is a commutative monoid, whose identity element is denoted by $0_R$.
  \item $(R,{\cdot})$ is a monoid, whose identity element is denoted by $1_R$.
  \item Multiplication distributes over addition: for $r,s,t \in R$, $r \cdot (s + t) = (r \cdot s) + (r \cdot t)$  and $(r + s) \cdot t = (r \cdot t) + (s \cdot t)$.
  \item $0_R$ is a zero for $(R,{\cdot})$: $0_R \cdot r = r \cdot 0_R = 0_R$ for every $r \in R$.
\end{itemize}
A semiring $R$ is \emph{idempotent} when $r + r = r$ for every $r \in R$, \emph{i.e.}, when the additive monoid $(R,+)$ is idempotent (there is no additional constraint on the multiplicative monoid $(R,\cdot)$). Given an idempotent semiring $(R,+,\cdot)$, one can define a canonical ordering $\leq$ over $R$:
\[\text{For all }
  r, s\in R,\quad r\leq s \text{ when } r+s=s.
\]
It is easy to check that $\leq$ is a partial order which is compatible with both addition and multiplication. Moreover, every morphism between two such commutative and idempotent monoids is increasing for this ordering.

\begin{exa}\label{ex:bgen:semiring}
  A key example of idempotent semiring is the set of all languages $2^{A^*}$. Union is the addition and language concatenation is the multiplication (with $\{\varepsilon\}$ as the identity element). Observe that in this case, the canonical ordering is inclusion. More generally, if $M$ is a monoid, then $2^M$ is an idempotent semiring whose addition is union, and whose multiplication is obtained by pointwise lifting that of $M$ to subsets.
\end{exa}

When dealing with subsets of an idempotent semiring $R$, we shall often apply a \emph{downset operation}. Given $S \subseteq R$, we write:
\[
  \dclosr S = \{r \in R \mid r \leq s \text{ for some $s \in S$}\}.
\]
We extend this notation to Cartesian products of arbitrary sets with $R$. Given some set $X$ and $S \subseteq X \times R$, we write,
\[
  \dclosr S = \{(x,r) \in X \times R \mid \text{$\exists s \in R$ such that $r \leq s$ and $(x,s) \in S$}\}.
\]

\subparagraph{\Mratms} We define a \emph{\mratm} as a semiring morphism $\rho: (2^{A^*},\cup,\cdot) \to (R,+,\cdot)$ where $(R,+,\cdot)$ is a \emph{finite} idempotent semiring, called the \emph{rating set of $\rho$}. That is, $\rho$ is a map from $2^{A^{*}}$ to $R$ satisfying the following properties:
\begin{enumerate}
  \item\label{itm:bgen:fzer} $\rho(\emptyset) = 0_R$.
  \item\label{itm:bgen:ford} For every $K_1,K_2 \subseteq A^*$, we have $\rho(K_1\cup K_2)=\rho(K_1)+\rho(K_2)$.
  \item\label{itm:bgen:funit} $\rho(\{\varepsilon\}) = 1_R$.
  \item\label{itm:bgen:fmult} For every $K_1,K_2 \subseteq A^*$, we have $\rho(K_1K_2) = \rho(K_1) \cdot \rho(K_2)$.
\end{enumerate}

For the sake of improved readability, when applying a \mratm $\rho$ to a singleton set $\{w\}$, we shall write $\rho(w)$ for $\rho(\{w\})$. Additionally, we write $\rho_*: A^* \to R$ for the restriction of $\rho$ to $A^*$: for every $w \in A^*$, we have $\rho_*(w) = \rho(w)$ (this notation is useful when referring to the language $\rho_*\inv(r) \subseteq A^*$, which consists of all words $w \in A^*$ such that $\rho(w) = r$). Note that $\rho_*$ is a morphism into the finite monoid $(R,\cdot)$.

\begin{rem}
  As the adjective ``\tame'' suggests, there exists a more general notion of ``\ratm'' introduced in~\cite{pzcovering2}. These are morphisms of idempotent and commutative monoids (\emph{i.e.}, it is not required that $R$ be equipped with a multiplication). However, we shall not use this notion in the paper.
\end{rem}

Most of the theory makes sense for arbitrary \mratms. Yet, in the paper, we work with special \mratms satisfying an additional property.

\subparagraph{\Nice and \fratms.} We say that a \mratm $\rho: 2^{A^*} \to R$ is \emph{\nice} when, for every language $K \subseteq A^*$, there exist finitely many words $w_1,\dots,w_n \in K$ such that $\rho(K) = \rho(w_1) + \cdots + \rho(w_k)$.

If a \ratm is \emph{simultaneously} \nice and \tame, we say that it is \emph{\full}. \Fratms are specially important.
This is because any \fratm $\rho: 2^{A^*} \to R$ is (fully) characterized by the canonical monoid morphism $\rho_*: A^* \to R$. Indeed, for $K \subseteq A^*$, we may consider the sum of all elements $\rho(w)$ for $w \in K$: while it may be infinite, this sum boils down to a finite one since $R$ is commutative and idempotent for addition. The hypothesis that $\rho$ is \nice implies that $\rho(K)$ is equal to this sum. The key point here is that \fratms are finitely representable: clearly, a \fratm $\rho$ is characterized by the morphism $\rho_*: A^* \to R$, which is finitely representable since it is a morphism into a finite monoid. Hence, we may speak about algorithms taking \fratms as input.

\subparagraph{Canonical \mratm associated to a monoid morphism.}  Finally, one can associate a particular \fratm $\rho_\alpha$ to every monoid morphism $\alpha: A^* \to M$ into a finite monoid. Its rating set is the idempotent semiring $(2^M,\cup,\cdot)$, whose multiplication is obtained by lifting the one of $M$ to subsets of~$M$. Moreover, for every language $K\subseteq A^*$, we let $\rho_\alpha(K)$ be the direct image $\alpha(K) \subseteq M$. In other words, we~define:
\[
  \begin{array}{llll}
    \rho_\alpha: & 2^{A^*} & \to     & 2^M                          \\
                 & K       & \mapsto & \{\alpha(w) \mid w \in K\}.
  \end{array}
\]
Clearly, $\rho_\alpha$ is a \fratm.

\subparagraph{Optimal \imprints.} Now that we have \mratms, we turn to \imprints. Consider a \mratm $\rho: 2^{A^*} \to R$. Given any finite set of languages~\Kb, we define the $\rho$-\imprint of~\Kb. Intuitively, when \Kb is a cover of some language $L$, this object measures the ``quality'' of \Kb. The  \emph{$\rho$-\imprint of \Kb} is the subset of~$R$ defined by:
\[
  \prin{\rho}{\Kb} = \dclosr \big\{\rho(K) \mid K \in\Kb\big\}.
\]
We now define optimality. Consider an arbitrary \mratm $\rho: 2^{A^*} \to R$ and a lattice~\Ds. Given a language $L$, an \emph{optimal} \Ds-cover of $L$ for $\rho$ is a \Ds-cover \Kb of $L$ having the smallest possible imprint among all \Ds-covers, \emph{i.e.}, which satisfies the following~property:
\[
  \prin{\rho}{\Kb} \subseteq \prin{\rho}{\Kb'} \quad \text{for every \Ds-cover $\Kb'$ of $L$}.
\]
In general, there can be infinitely many optimal \Ds-covers for a given \mratm $\rho$. The key point is that there always exists at least one, provided that \Ds is a lattice. We state this property in the following lemma (proved in~\cite{pzcovering2}).

\begin{lem}\label{lem:bgen:opt}
  Let \Ds be a lattice. For every language $L$ and every \mratm~$\rho$, there exists an optimal \Ds-cover of $L$ for $\rho$.
\end{lem}

Clearly, given a lattice \Ds, a language $L$ and a \mratm $\rho$, all optimal \Ds-covers of $L$ for $\rho$ have the same $\rho$-\imprint. Hence, this unique $\rho$-\imprint is a \emph{canonical} object for \Ds, $L$ and $\rho$. We call it the \emph{\Ds-optimal $\rho$-\imprint on $L$} and we write it $\opti{\Ds}{L,\rho}$:
\[
  \opti{\Ds}{L,\rho} = \prin{\rho}{\Kb} \quad \text{for any optimal \Ds-cover \Kb of $L$  for $\rho$}.
\]
An important special case is when $L = A^*$. In this case, we write \opti{\Ds}{\rho} for \opti{\Ds}{A^*,\rho}. Let us present a few properties of optimal \imprints (all proved in~\cite{pzcovering2}). First, we have the following useful fact, which is immediate from the definitions.

\begin{fct}\label{fct:linclus}
  Let $\rho$ be a \mratm and consider two languages $H,L$ such that $H \subseteq L$. Then, $\opti{\Ds}{H,\rho} \subseteq \opti{\Ds}{L,\rho}$.
\end{fct}

Additionally, the following lemma describes optimal \imprints of a union of languages.

\begin{lem}\label{lem:optunion}
  Let \Ds be a lattice and let $\rho: 2^{A^*} \to R$ be a \mratm. Given two languages $H,L$, we have $\opti{\Ds}{H \cup L,\rho} = \opti{\Ds}{H,\rho} \cup \opti{\Ds}{L,\rho}$
\end{lem}

We complete Lemma~\ref{lem:optunion} with a similar statement for language concatenation instead of union. Note that it requires more hypotheses on the class \Ds: we need closure under~quotient.

\begin{lem}\label{lem:mratmult}
  Let \Ds be a \pvari and let $\rho: 2^{A^*} \to R$ be a \mratm. Given two languages $H,L \subseteq A^*$, we have $\opti{\Ds}{H,\rho}\cdot\opti{\Ds}{L,\rho}\subseteq \opti{\Ds}{HL,\rho}$.
\end{lem}

\subparagraph{Connection with covering.} We may now connect these definitions to the covering problem. The key idea is that solving \Ds-covering for a fixed class \Ds boils down to finding an algorithm that computes \Ds-optimal \imprints from \fratms given as inputs. In~\cite{pzcovering2}, two statements are presented. The first is simpler but it only applies to Boolean algebras, while the second is more involved and applies to all lattices. Since all classes investigated in the paper are Boolean algebras, we only present the first statement.

\begin{prop}\label{prop:breduc}
  Let \Ds be a Boolean algebra. There exists an effective reduction from \Ds-covering to the following problem:

  \begin{tabular}{ll}
    {\bf Input:} & A \fratm $\rho: 2^{A^*} \to R$ and $F \subseteq R$. \\
    {\bf Question:} & Is it true that $\opti{\Ds}{\rho} \cap F = \emptyset$?
  \end{tabular}
\end{prop}

\begin{proof}[Proof sketch]
  We briefly describe the reduction (we refer the reader to~\cite{pzcovering2} for details). Consider an input pair $(L_0,\{L_1,\dots,L_n\})$ for \Ds-covering. Since the languages $L_i$ are regular, for every $i \leq n$, one can compute a morphism $\alpha_i: A^* \to M_i$ into a finite monoid recognizing~$L_i$ together with the set $F_i \subseteq M_i$ such that $L_i = \alpha_i\inv(F_i)$. Consider the associated \fratms $\rho_{\alpha_i} : 2^{A^*} \to 2^{M_i}$. Moreover, let $R$ be the idempotent semiring $2^{M_0} \times \cdots \times 2^{M_n}$ equipped with the componentwise addition and multiplication. We define a \fratm $\rho: 2^{A^*} \to R$ by letting $\rho(K) = (\rho_{\alpha_0}(K),\dots,\rho_{\alpha_n}(K))$ for every $K \subseteq A^*$. Finally, let $F \subseteq R$ be the set of all tuples $(X_0,\dots,X_n) \in R$ such that $X_i \cap F_i \neq \emptyset$ for every $i \leq n$. One can now verify that $(L_0,\{L_1,\dots,L_n\})$ is \Ds-coverable if and only if $\opti{\Ds}{\rho} \cap F = \emptyset$. Let us point out that this equivalence is only true when \Ds is a Boolean algebra. When \Ds is only a lattice, one has to handle the language $L_0$ separately.
\end{proof}

\subparagraph{Pointed optimal \imprints.} In view of Proposition~\ref{prop:breduc}, given a Boolean algebra \Ds, an algorithm computing \opti{\Ds}{\rho} from a \fratm $\rho$ yields a procedure for \Ds-covering. We consider the case where $\Ds = \upol{\Cs}$ for some finite \vari \Cs. We present a least fixpoint procedure for computing \opti{\upol{\Cs}}{\rho}. Yet, implementing it requires considering an object which carries \emph{more information} than \opti{\upol{\Cs}}{\rho}, which we now define.

Let \Ds be a Boolean algebra, $\eta: A^* \to N$ be a morphism into a finite monoid and $\rho: 2^{A^*} \to R$ be a \mratm. The \emph{$\eta$-pointed \Ds-optimal $\rho$-\imprint} is defined as the following set $\popti{\Ds}{\eta}{\rho} \subseteq N \times R$:
\[
  \popti{\Ds}{\eta}{\rho} = \bigl\{(t,r) \in N \times R \mid r \in \opti{\Ds}{\eta\inv(t),\rho}\bigl\}.
\]
In view of this definition, we shall manipulate Cartesian products $N \times R$ where $N$ is a finite monoid and $R$ is a finite idempotent semiring. In this context, it will be convenient to use the following functional notation. Given a subset $S \subseteq N \times R$ and $t \in N$, we write $S(t) \subseteq R$ for the set $S(t) = \{r \in R \mid (t,r) \in S\}$.

The set \popti{\Ds}{\eta}{\rho} encodes the \Ds-optimal $\rho$-\imprint on $A^*$, that is, the subset $\opti{\Ds}{\rho}$ of $R$. Indeed, since $\opti{\Ds}{\rho} = \opti{\Ds}{A^*,\rho}$, we have the following immediate corollary of Lemma~\ref{lem:optunion}.

\begin{cor}\label{cor:pointgen}
  Let \Ds be a Boolean algebra, $\eta: A^* \to N$ be a morphism into a finite monoid and $\rho: 2^{A^*} \to R$ be a \mratm. Then,
  \[
    \opti{\Ds}{\rho} = \bigcup_{t\in N}  \opti{\Ds}{\eta\inv(t),\rho} = \bigcup_{t\in N}  \popti{\Ds}{\eta}{\rho}(t).
  \]
\end{cor}

In the sequel, we consider the case where $\Ds = \upol{\Cs}$, for some finite \vari~\Cs. We present an algorithm for computing $\popti{\upol{\Cs}}{\etac}{\rho} \subseteq \canc \times R$ from a \fratm $\rho$. Recall that $\etac: A^* \to \canc$ is the canonical \Cs-morphism, which is well-defined since \Cs is a \emph{finite} \vari. This is actually the reason why the algorithm depends on the finiteness of~\Cs.

\subsection{Characterization of \upol{\Cs}-optimal \imprints}

Let us first describe the property characterizing \upol{\Cs}-optimal \imprints for a finite \vari \Cs. Consider a morphism $\eta: A^* \to N$ into a finite monoid (as explained above, in the statement, $\eta$ will be the canonical \Cs-morphism \etac) and a \mratm $\rho: 2^{A^*} \to R$. We say that a subset $S \subseteq N \times R$ is \emph{\upolo-saturated for $\eta$ and $\rho$} if satisfies the four following properties:
\begin{enumerate}[leftmargin=*,topsep=1ex]
  \item \emph{Trivial elements}: for every $w \in A^*$, we have $(\eta(w),\rho(w)) \in S$.
  \item \emph{Closure under downset}: we have $\dclosr S = S$.
  \item \emph{Closure under multiplication}: for every $(s_1,r_1),(s_2,r_2) \in S$, we have $(s_1s_2,r_1r_2) \in S$.
  \item \emph{\upolo-closure}: if $(e_1,f_1),(e_2,f_2) \in S$ are pairs of multiplicative idempotents and we have some $s \in N$ such that $e_1 \Rord se_2$ and $e_2 \Lord e_1s$, then $(e_1se_2,f_1 \rho(\eta\inv(s)) f_2) \in S$.
\end{enumerate}

We are ready to state the main theorem of the section. Recall that when \Cs is a finite \vari, we write $\etac: A^* \to \canc$ for the canonical \Cs-morphism (see Section~\ref{sec:prelims}). We prove that for every \mratm $\rho: 2^{A^*} \to R$, the $\etac$-pointed \upol{\Cs}-optimal $\rho$-\imprint \pupolopti is the  least \upolo-saturated subset of $\canc \times R$ (for inclusion) for $\etac$ and $\rho$.

\begin{theorem}\label{thm:upolopt}
  Let \Cs be a finite \vari and $\rho: 2^{A^*} \to R$ be a \mratm. Then, \pupolopti is the least \upolo-saturated subset of $\canc \times R$ for $\etac$ and $\rho$.
\end{theorem}

\begin{rem}
  When \Cs is finite, Theorem~\ref{thm:upolopt} characterizes the least \upolo-saturated subset for the canonical \Cs-morphism and a multiplicative rating map. However, note that in order to obtain the decidability of \upol\Cs-covering, we need the input to be finitely representable, \emph{i.e.}, we need the \ratm to be \emph{\full}. In this case, it is clear that one can use a least fixpoint procedure to compute the least \upolo-saturated subset of $\canc \times R$ from an input \fratm $\rho: 2^{A^*} \to R$. Therefore, Theorem~\ref{thm:upolopt} yields a procedure for computing \pupolopti from an input \fratm $\rho: 2^{A^*} \to R$. By Corollary~\ref{cor:pointgen}, it follows that we may compute $\upolopti \subseteq R$ as well. Hence, Proposition~\ref{prop:breduc} yields that \upol{\Cs}-covering is decidable. This extends to \upol{\Cs}-separation by Lemma~\ref{lem:covsep}.
\end{rem}

From the above remark, we obtain the following corollary of Theorem~\ref{thm:upolopt}.

\begin{cor}\label{cor:upolopt}
  For every finite \vari \Cs, \upol{\Cs}-covering and \upol{\Cs}-separation are decidable.
\end{cor}

A key application of Corollary~\ref{cor:upolopt} is when \Cs is the finite \vari \at of alphabet testable languages: we get the decidability of covering and separation for $\upol{\at} = \upol{\bpol{\stzer}}$ (we proved this equality in Lemma~\ref{lem:upolat}). Let us point out that this result was already known: it is proved in~\cite{pzcovering2} using a specialized argument based on the logical characterization of \upol{\at} in terms of two-variable first-order logic (we present this logical characterization in Section~\ref{sec:logcar}). Actually, specializing the above generic algorithm yields exactly the procedure of~\cite{pzcovering2}.

\smallskip

We turn to the proof of Theorem~\ref{thm:upolopt}. We present two independent statements, which correspond to soundness and completeness of the least fixpoint procedure that computes \pupolopti. We start with the former.

\begin{prop}[Soundness]\label{prop:upsound}
  Let \Cs be a finite \vari and $\rho: 2^{A^*} \to R$ be a \mratm. Then, $\pupolopti \subseteq \canc \times R$ is \upolo-saturated for $\etac$ and $\rho$.
\end{prop}

\begin{proof}
  Recall that \upol{\Cs} is a \vari by Theorem~\ref{thm:comp}. There are four properties to verify. We start with the first three which are standard.  For the trivial elements, consider $w \in A^*$ and let \Kb be an optimal \upol{\Cs}-cover of $\etac\inv(\etac(w))$. Since $w \in \etac\inv(\etac(w))$, there exists $K \in \Kb$ such that $w\in K$. Thus, $\rho(w) \leq \rho(K)$ which yields $\rho(w) \in \prin{\rho}{\Kb} = \opti{\upol{\Cs}}{\etac\inv(\etac(w)),\rho}$. This implies that $(\etac(w),\rho(w)) \in \pupolopti$. Closure under downset is immediate by definition of \imprints. Finally, for closure under multiplication, consider $(s_1,r_1),(s_2,r_2) \in \pupolopti$. We have $r_i \in \opti{\upol{\Cs}}{\etac\inv(s_i),\rho}$ for $i = 1,2$. Since \upol{\Cs} is a \vari, Lemma~\ref{lem:mratmult} yields $r_1r_2 \in \opti{\upol{\Cs}}{\etac\inv(s_1)\etac\inv(s_2),\rho}$. Clearly, $\etac\inv(s_1)\etac\inv(s_2) \subseteq \etac\inv(s_1s_2)$. Thus, Fact~\ref{fct:linclus} yields $r_1r_2 \in \opti{\upol{\Cs}}{\etac\inv(s_1s_2),\rho}$. By definition, this exactly says that $(s_1s_2,r_1r_2) \in  \pupolopti$.

  It remains to prove that \pupolopti satisfies \upolo-closure. Consider two pairs of multiplicative idempotents $(e_1,f_1),(e_2,f_2) \in \pupolopti$ and $s \in \canc$ such that $e_1 \Rord se_2$ and $e_2 \Lord e_1s$. We prove that $(e_1se_2,f_1 \rho(\etac\inv(s)) f_2) \in \pupolopti$. By definition, we have to show that,
  \[
    f_1 \rho(\etac\inv(s)) f_2 \in \opti{\upol{\Cs}}{\etac\inv(e_1se_2),\rho}.
  \]
  By definition, this boils down to proving that given an arbitrary \upol{\Cs}-cover \Kb of $\etac\inv(e_1se_2)$, we have $f_1 \rho(\etac\inv(s)) f_2 \in \prin{\rho}{\Kb}$. We fix \Kb for the proof. Proposition~\ref{prop:genocm} yields a \upol{\Cs}-morphism $\alpha: A^* \to M$ recognizing every language $K \in \Kb$. Let $k = \omega(M)$.

  Let $i \in \{1,2\}$. By hypothesis, we know that $(e_i,f_i) \in \pupolopti$. By definition, this means that $f_i \in \opti{\upol{\Cs}}{\etac\inv(e_i),\rho}$. Clearly, the set $\{\alpha\inv(t) \mid t\in M \text{ and } \alpha\inv(t) \cap \etac\inv(e_i) \neq \emptyset\}$ is a \upol{\Cs}-cover of $\etac\inv(e_i)$. Therefore, since $f_i \in \opti{\upol{\Cs}}{\etac\inv(e_i),\rho}$, we get $t_i \in M$ such that $\etac\inv(e_i) \cap \alpha\inv(t_i) \neq \emptyset$ and $f_i \leq \rho(\alpha\inv(t_i))$. We fix $w_i \in \etac\inv(e_i) \cap \alpha\inv(t_i)$ for the proof. Moreover, we let $u \in \etac\inv(s)$. We define,
  \[
    w = w_1^k u w_2^k.
  \]
  Since $e_1,e_2\in\canc$ are idempotents, we have $\etac(w) = e_1se_2$. Since \Kb is a cover of $\etac\inv(e_1se_2)$, we get $K \in \Kb$ such that $w \in K$. We prove that,
  \begin{equation}\label{eq:upol:eqsound}
    (\alpha\inv(t_1))^k \etac\inv(s) (\alpha\inv(t_2))^k \subseteq K.
  \end{equation}
  Let us first explain how to use~\eqref{eq:upol:eqsound} to conclude the proof. Since  $f_i \leq \rho(\alpha\inv(t_i))$ for $i=1,2$, the inclusion given by~\eqref{eq:upol:eqsound} implies that $f_1^k \rho(\etac\inv(s)) f_2^k \leq \rho(K)$. Moreover, since $f_1,f_2 \in R$ are multiplicative idempotents, this yields $f_1 \rho(\etac\inv(s)) f_2 \leq \rho(K)$. Finally, since $K \in \Kb$, we get $f_1 \rho(\etac\inv(s)) f_2 \in \prin{\rho}{\Kb}$, as desired.

  It remains to prove that~\eqref{eq:upol:eqsound} holds. We fix $w' \in (\alpha\inv(t_1))^k \etac\inv(s) (\alpha\inv(t_2))^k$ for the proof and show that $w' \in K$. Since $K$ is recognized by $\alpha$ and $w \in K$, it suffices to prove that $\alpha(w') = \alpha(w)$. For $i = 1,2$, we write $g_i = t_i^\omega$. By definition of $w$ and of $k=\omega(M)$, we have $\alpha(w) = g_1\alpha(u)g_2$. By hypothesis on $w'$, there exists $u' \in \etac\inv(s)$ such that $\alpha(w') = g_1 \alpha(u') g_2$. Hence, it remains to show that $g_1 \alpha(u)g_2 = g_1 \alpha(u') g_2$. First, we use our hypothesis on $e_1,e_2,s \in \canc$: since $e_1 \Rord se_2$ and $e_2 \Lord e_1s$, there exist $q_1,q_2 \in \canc$ such that $e_1 = se_2q_1$ and $e_2 = q_2e_1s$. Let $p = q_2e_1se_2q_1 \in \canc$. Since $e_1$ and $e_2$ are idempotents, we have,
  \[
    \begin{array}{lllllllll}
      e_1 & = & se_2q_1 & = & se_2e_2e_2q_1 & = & se_2q_2e_1se_2q_1 & = & se_2 p, \\
      e_2 & = & q_2e_1s & = & q_2e_1e_1e_1s & = & q_2e_1se_2q_1e_1s & = & pe_1s. \\
    \end{array}
  \]
  Let $v \in \etac\inv(p)$. Since $\etac(w_i^k) = e_i$ for  $i = 1,2$ and $\etac(u) = \etac(u') = s$, the above yields the equalities $\etac(w_1^k) = \etac(u'w_2^{k}v)$ and $\etac(w_2^{k}) =\etac(vw_1^{k}u)$. Since \etac is the canonical \Cs-morphism and $\alpha(w_i^k) =g_i$ for $i = 1,2$, one  obtain, using Lemma~\ref{lem:cmdiv} and Lemma~\ref{lem:smult}, that $g_1 \canec \alpha(u')g_2\alpha(v)$ and $g_2 \canec \alpha(v)g_1\alpha(u)$. Finally, since $\alpha$ is a \upol{\Cs}-morphism, Theorem~\ref{thm:caracupol} implies that it satisfies~\eqref{eq:cupol}. Since $g_1$ and $g_2$ are idempotents, this yields: $g_1 = g_1\alpha(u')g_2\alpha(v)g_1$ and $g_2 = g_2\alpha(v)g_1\alpha(u)g_2$. It follows that,
  \[
    g_1 \alpha(u) g_2  = g_1\alpha(u')g_2\alpha(v)g_1\alpha(u)g_2 = g_1\alpha(u') g_2.
  \]
  This concludes the proof.
\end{proof}

We turn to the completeness direction in Theorem~\ref{thm:upolopt}.

\begin{prop}[Completeness]\label{prop:maincov}
  Let \Cs be a \vari, $\eta: A^* \to N$ be a \Cs-morphism, $\rho: 2^{A^*} \to R$ be a \mratm and $S \subseteq N \times R$ be \upolo-saturated for $\eta$ and~$\rho$. For every $t \in N$, there exists a \upol{\Cs}-cover $\Kb_t$ of $\eta\inv(t)$ such that $\prin{\rho}{\Kb_t} \subseteq S(t)$.
\end{prop}

\begin{proof}
  We fix \Cs, $\eta$, $\rho$ and $S$ as in the statement. Since $S$ is \upolo-saturated, it is closed under multiplication, which implies that $S$ is a monoid for the componentwise multiplication (the identity element is the trivial element $(1_N,1_R)=(\eta(\veps),\rho(\veps))$). The proposition is a corollary of the following lemma, which we prove by induction.

  \begin{lem}\label{lem:upol:maincov}
    Let $(x,p),(y,q) \in S$ and $t \in N$. There exists a \upol{\Cs}-partition \Kb of $\eta\inv(t)$ such that $(xty,p\rho(K)q) \in S$ for every $K \in \Kb$.
  \end{lem}

  We first apply the lemma to complete the proof of Proposition~\ref{prop:maincov}. We apply it for $(x,p) = (y,q) = (1_{N},1_R) \in S$. For every $t \in N$, this yields a \upol{\Cs}-partition $\Kb_t$ of $\eta\inv(t)$ such that $(t,\rho(K)) \in S$ for every $K \in \Kb_t$. Since $S$ is closed under downset, the fact that $(t,\rho(K)) \in S$ for every $K \in \Kb_t$ implies that $\prin{\rho}{\Kb_t} \subseteq S(t)$, as desired.

  \smallskip

  It remains to prove Lemma~\ref{lem:upol:maincov}. The proof is reminiscent of that of Lemma~\ref{lem:upol:adetsuf}. Let $(x,p),(y,q) \in S$ and $t \in N$. We build a \upol{\Cs}-partition \Kb of $\eta\inv(t)$ such that $(xty,p\rho(K)q) \in S$ for every $K \in \Kb$ by induction on the following three parameters listed by order of importance (they depend on the Green relations of the finite monoids $N$ and~$S$).
  \begin{enumerate}
    \item The \emph{rank of $xty \in N$}: the number of elements $t' \in N$ such that $xty \Jord t'$.
    \item The \emph{\Rrel-index of $(x,p)$}: the number of pairs $(x',p') \in S$ such that $(x',p') \Rord (x,p)$.
    \item The \emph{\Lrel-index of $(y,q)$}: the number of pairs $(y',q') \in S$ such that $(y',q') \Lord (y,q)$.
  \end{enumerate}
  We distinguish three cases depending on whether $xty \Jrel t$ and on which of the two following properties of $t$, $(x,p)$ and $(y,q)$ are fulfilled:
  \begin{itemize}
    \item $(x,p)$ is \emph{right $t$-stable} when there is $(z,r) \in S$ such that $(xz,pr) \Rrel (x,p)$ and $z \Rrel t$.
    \item $(y,q)$ is \emph{left $t$-stable} when there is $(z,r) \in S$ such that $(zy,rq) \Lrel(y,q)$ and $z \Lrel t$.
  \end{itemize}
  In the base case, we assume that all three properties holds and conclude directly. Otherwise, we consider two distinct inductive cases. In the first one, we assume that $xty \Jords t$ and in the second one, that either $(x,p)$ is not right $t$-stable or $(y,q)$ is not left $t$-stable.

\subparagraph{Base case: $xty \Jrel t$, $(x,p)$ is right $t$-stable and $(y,q)$ is left $t$-stable.} We let $\Kb = \{\eta\inv(t)\}$. Clearly, this is a \upol{\Cs}-partition of $\eta\inv(t)$ since $\eta$ is a \Cs-morphism. Hence, it remains to prove that $(xty,p \rho(\eta\inv(t))q) \in S$. We first use our hypothesis to prove the following fact.

  \begin{fct}\label{fct:upol:stableupol}
    There exist two pairs of multiplicative idempotents $(e_1,f_1),(e_2,f_2) \in S$ satisfying the following conditions: $(xe_1,pf_1) = (x,p)$, $(e_2 y,f_2 q) = (y,q)$, $e_1 \Rord te_2$ and $e_2 \Lord e_1t$.
  \end{fct}

  \begin{proof}
    We exhibit $(e_1,f_1)$ using the right $t$-stability of $(x,p)$. Since $(x,p)$ is right $t$-stable, there exists  $(z,r) \in S$ such that $(xz,pr) \Rrel (x,p)$ and $z \Rrel t$. Hence, we obtain $(z',r') \in S$ such that $(x,p) = (xzz',prr')$. Let $k=\omega(S) \geq 1$. Then, $(e_1,f_1) =  ((zz')^k,(rr')^k)$  is a pair of multiplicative idempotents. It is now immediate that $(xe_1,pf_1) = (x,p)$. The existence of $(e_2,f_2)$ is proved symmetrically using the left $t$-stability of $(y,q)$. We now show that $e_1 \Rord te_2$ (the proof that $e_{2}\Lord e_1t$ is symmetrical and left to the reader). Since $z \Rrel t$, it is clear that $e_1 = (zz')^k \Rord t$. Moreover, since $xty=xe_1te_2y \Jrel t$ by hypothesis, we have $te_2 \Jrel t$. Since it is clear that $te_2 \Rord t$, Lemma~\ref{lem:jlr} yields $te_2 \Rrel t$. Altogether, we get $e_1 \Rord te_2$.
  \end{proof}

  We let $(e_1,f_1),(e_2,f_2) \in S$ be given by Fact~\ref{fct:upol:stableupol}. Since $(e_1,f_1),(e_2,f_2) \in S$ are multiplicative idempotents such that $e_1 \Rord te_2$ and $e_2 \Lord e_1t$, and $S$ is \upolo-saturated, we get from \upolo-closure that $(e_1te_2, f_1\rho(\eta\inv(t))f_2) \in S$. Finally, since we have $(x,p),(y,q) \in S$, $(xe_1,pf_1) = (x,p)$ and $(e_2 y,f_2 q) = (y,q)$, we obtain from closure under multiplication, that $(x,p)(e_1te_2, f_1\rho(\eta\inv(t))f_2)(y,q)=(xty,p\rho(\eta\inv(t))q) \in S$, which concludes this case.

\subparagraph{First inductive case: $xty \Jords t$.} In this case, the rank of $t = 1_{N}t1_{N}$ is strictly smaller than the one of $xty$. Consequently, induction on our first parameter in Lemma~\ref{lem:upol:maincov} (applied for $(x,q) = (y,r) = (1_N,1_R) \in S$) yields a \upol{\Cs}-partition \Kb of $\eta\inv(t)$ such that $(t,\rho(K)) \in S$ for every $K \in \Kb$. Since $(x,q),(y,r) \in S$ and $S$ is closed under multiplication, it then follows that $(xty,q\rho(K)r) \in S$ for every $K \in \Kb$, concluding this~case.

\subparagraph{Second inductive case: either $(x,p)$ is not right $t$-stable or $(y,q)$ is not left $t$-stable.} There are two symmetrical cases depending on which property holds. We handle the case where $(y,q)$ is not left $t$-stable and leave the other, which follows by symmetry, to the reader. Let $T$ be the set of all triples $(t_1,a,t_2) \in  N\times A \times N$ such that $t_1 \eta(a) t_2 = t$ and $t \Lrel \eta(a)t_2 \Lords t_2$. In the next fact, we use induction to build \upol{\Cs}-partitions of $\eta\inv(t_2)$ and $\eta\inv(t_1)$ for every triple $(t_1,a,t_2) \in T$. We shall then combine them to construct the desired \upol\Cs-partition \Kb of $\eta\inv(t)$.

  \begin{fct}\label{fct:upol:cupolind}
    Consider a triple $(t_1,a,t_2) \in T$. There exists a \upol{\Cs}-partition $\Vb_{t_2}$ of $\eta\inv(t_2)$ such that $(t_2,\rho(V)) \in S$ for every $V \in \Vb_{t_2}$. Moreover, for every $V \in \Vb_{t_2}$, there exists a \upol{\Cs}-partition $\Ub_{(t_1,a,t_2),V}$ of $\eta\inv(t_1)$ such that $(xty,p\rho(UaV)q) \in S$ for every $U \in \Ub_{(t_1,a,t_2),V}$.
  \end{fct}

  \begin{proof}
    By definition of $T$, we know that $t\Lords t_2$. This implies $t\Jords t_2$ by Lemma~\ref{lem:jlr} and since $xty \Jord t$, we get $xty \Jords t_2$. Hence, the rank of $t_2 = 1_{N}t_{2}1_{N}$ is strictly smaller than the one of $xty$. Therefore, by induction on our first and main parameter in Lemma~\ref{lem:upol:maincov} (applied for $(x,p) = (y,q) = (1_N,1_R)$) we obtain a \upol{\Cs}-partition $\Vb_{t_2}$ of $\eta\inv(t_2)$ such that $(t_2,\rho(V)) \in S$ for every $V \in \Vb_{t_2}$.

    We now fix $V \in \Vb_{t_2}$ and build $\Ub_{(t_1,a,t_2),V}$. Since $(t_2,\rho(V)) \in S$ by definition of $\Vb_{t_2}$ and $(\eta(a),\rho(a)) \in S$ (this is a trivial element), we have $(\eta(a)t_2,\rho(aV))  \in S$. Hence, since $t \Lrel \eta(a)t_2$ by definition of $T$, the hypothesis that $(y,q)$ is not left $t$-stable yields $(\eta(a)t_2y,\rho(aV)q) \Lords (y,q)$. Consequently, the \Lrel-index of $(\eta(a)t_2y,\rho(aV)q)$ is strictly smaller than the one of $(y,q)$. We may apply induction to  $(x,p)$, $(\eta(a)t_2y,\rho(aV)q)$ and $t_1$. Indeed, while the third parameter has decreased, the first has not increased since $xt_1\eta(a)t_2y = xty$ by definition of $T$. The second one has not increased as well since we did not change the first pair $(x,p)$. Hence, induction on our third parameter in Lemma~\ref{lem:upol:maincov} applied when $(y,q)$ has been replaced by $(\eta(a)t_2y,\rho(aV)q)$ and $t$ by $t_1$ yields a \upol{\Cs}-partition $\Ub_{(t_1,a,t_2),V}$ of $\eta\inv(t_1)$ such that $(xt_1\eta(a)t_2y,p\rho(U)\rho(aV)q) \in S$ for every $U \in \Ub_{(t_1,a,t_2),V}$.  Finally, since $xt_1\eta(a)t_2y = xty$, we have $(xty,p\rho(UaV)q) \in S$ for every $U \in \Ub_{(t_1,a,t_2),V}$.
  \end{proof}

  We are ready to construct our \upol{\Cs}-partition \Kb of $\eta\inv(t)$. We define,
  \[
    \Kb = \bigcup_{(t_1,a,t_2) \in T} \big\{UaV \mid V \in \Vb_{t_2} \text{ and } U \in \Ub_{(t_1,a,t_2),V}\big\}.
  \]
  It is immediate from the definition and Fact~\ref{fct:upol:cupolind} that $(xty,p\rho(K)q) \in S$ for every $K \in \Kb$. It remains to verify that \Kb is a \upol{\Cs}-partition of $\eta\inv(t)$. We first prove that it is a partition of $\eta\inv(t)$: let $w \in \eta\inv(t)$, we show that there is a \emph{unique} $K \in \Kb$ such that~$w \in K$. Let $v' \in A^*$ be the least suffix of $w$ such that $t = \eta(w) \Lrel \eta(v')$ and let $u\in A^*$ such that $w = uv'$. Observe that $v' \neq \veps$ (otherwise, we would have $t \Lrel 1_N$, contradicting the hypothesis that $(y,q)$ is not left $t$-stable, since $(1_Ny,1_Rq) \Lrel (y,q)$). Hence, there exists $v\in A^*$ and $a \in A$ such that $v' = av$. Let $t_1 = \eta(u)$ and $t_2 = \eta(v)$. By definition $t_1\eta(a)t_2 = \eta(uav) = \eta(w) = t$. Moreover, $t \Lrel \eta(v') = \eta(a)t_2 \Lords \eta(v)$ by definition of $v'$ as the least prefix of $w$ such that $t = \eta(w) \Lrel \eta(v')$. Hence, $(t_1,a,t_2) \in T$ and we may consider the partition $\Vb_{t_2}$ of $\eta\inv(t_2)$. In particular, there exists $V \in \Vb_{t_2}$ such that $v\in V$ since~$t_2=\eta(v)$. Moreover, since $\Ub_{(t_1,a,t_2),V}$ is a partition of $\eta\inv(t_1)$ and $t_1=\eta(u)$, there exists $U \in \Ub_{(t_1,a,t_2),V}$ such that $u \in U$. Hence, $w = uav \in UaV$ and $UaV \in \Kb$ is the unique language of \Kb containing $w$ by definition.

  Finally, let us verify that for every $K \in \Kb$, we have $K \subseteq \eta\inv(t)$ and $K \in \upol{\Cs}$ which completes the proof that \Kb is a \upol{\Cs}-partition of $\eta\inv(t)$. Let $K\in\Kb$. By definition, $K = UaV$ where $V \in \Vb_{t_2}$ and $U \in \Ub_{(t_1,a,t_2),V}$ for some $(t_1,a,t_2) \in T$. Since $\Vb_{t_2}$ and $\Ub_{(t_1,a,t_2),V}$ are \upol{\Cs}-partitions of $\eta\inv(t_2)$ and $\eta\inv(t_1)$ respectively, we know that $U,V \in \upol{\Cs}$, $U \subseteq\eta\inv(t_1)$ and $V \subseteq \eta\inv(t_2)$. It follows that $K \subseteq \eta\inv(t_1)a\eta\inv(t_2)$. Hence, since $t_1\eta(a)t_2 = t$ by definition of $T$, we get $K \subseteq \eta\inv(t)$. Moreover, $U,V \in \upol{\Cs}$ and $UaV$ is right deterministic by Lemma~\ref{lem:greendet} since $V \subseteq  \eta\inv(t_2)$ and $t_2\eta(a) \Lords t_2$ (as $(t_1,a,t_2) \in T$). This concludes the proof.
\end{proof}

We are now ready to prove Theorem~\ref{thm:upolopt}.

\begin{proof}[Proof of Theorem~\ref{thm:upolopt}]
  Let \Cs be a finite \vari and $\rho: 2^{A^*} \to R$ be a \mratm. By Proposition~\ref{prop:upsound}, $\upolopti \subseteq \canc \times R$ is \upolo-saturated for $\etac$ and $\rho$. It remains to show that it is the least such set. Thus, let $S \subseteq \canc \times R$ be \upolo-saturated for \etac and $\rho$. We show that $\pupolopti \subseteq S$. Clearly, it suffices to show that $\pupolopti(t)\subseteq S(t)$ for every $t \in \canc$. We fix $t$ for the proof.

  Since $\etac$ is a \Cs-morphism, Proposition~\ref{prop:maincov} yields a \upol{\Cs}-cover \Kb of $\etac\inv(t)$ such that $\prin{\rho}{\Kb_t} \subseteq S(t)$. Moreover, recall that $\pupolopti(t) = \opti{\upol{\Cs}}{\etac\inv(t),\rho}$. Therefore, since \Kb is a \upol{\Cs}-cover of $\etac\inv(t)$, we get $\pupolopti(t) \subseteq \prin{\rho}{\Kb}$. Altogether, we obtain $\pupolopti(t)\subseteq S(t)$, as desired.
\end{proof}

\section{Unary temporal logic}
\label{sec:logic}
We introduce \emph{unary temporal logic}. We use a definition that generalizes the standard one: for each class \Cs, we define a particular variant of unary temporal logic that we denote by \tla{\Cs} and associate a class \tlc{\Cs} to it. The standard definitions of unary temporal logic found in the literature correspond either to \tlc{\stzer} or to \tlc{\stzer^+}.

We prove two key results in the section. First, we establish a connection with two-variable first-order logic. For every Boolean algebra \Cs, we prove that $\tlc{\Cs}=\fod(\infsigc)$. This result generalizes a well-known theorem by Etessami, Vardi and Wilke~\cite{evwutl}. We also compare these logical classes to those built with unambiguous polynomial closure. More precisely, we prove the given an arbitrary \vari \Cs, we have the inclusion $\upol{\bpol{\Cs}} \subseteq \tlc{\Cs}$. While it is strict in general, we shall prove in the next section that when \Cs is a \vari of group languages or a \wsuit extension thereof, the converse inclusion holds as well.

\subsection{Definition and properties}

We actually define two distinct sets of temporal formulas, which we denote by \tlxs and \tls. Then, we explain how these sets can be restricted depending on some class of languages \Cs. This yields two new classes built from \Cs, which we write \tlxc{\Cs} and \tlc{\Cs}. Let us first define the \tlxs formulas, which are more general.

A \tlxs formula is built from atomic formulas using Boolean connectives and temporal operators. The atomic formulas are $\top$, $\bot$, $min$, $max$ and ``$a$'' for every letter $a \in A$. All Boolean connectives are allowed: if $\psi_1$ and $\psi_2$ are \tlxs formulas, then so are $(\psi_{1} \vee \psi_{2})$, $(\psi_{1} \wedge \psi_{2})$ and $(\neg \psi_1)$. There are two kinds of \emph{temporal operators}, which are both unary. First, one may use $\textup{X}$ and $\textup{Y}$: if $\psi$ is a \tlxs formula, then so are $(\nex{\psi})$ and $(\nexm{\psi}$). Moreover, we associate \emph{two temporal operators} to every language $L \subseteq A^*$, which we write $\textup{F}_L$ and $\textup{P}_L$: if $\psi$ is a \tlxs formula, then so are $(\finallyl{\psi})$ and $(\finallyml{\psi})$. For the sake of improved readability, we omit parentheses when there is no ambiguity. Finally, we define \tls as a syntactical restriction: a \tls formula is a \tlxs formula that does \emph{not} contain $\textup{X}$, nor $\textup{Y}$.

We now turn to the semantics. Since \tls formulas are particular \tlxs formulas, it suffices to define the semantics of the latter. Evaluating a \tlxs formula $\varphi$ requires a word $w \in A^*$ and a position $i \in \pos{w}$. We use structural induction on $\varphi$ to define what it means for \emph{$(w,i)$ to satisfy $\varphi$}. We denote this property by $w,i \models \varphi$:
\begin{itemize}
  \item {\bf Atomic formulas:} $w,i \models \top$ always holds and $w,i \models \bot$ never holds. Additionally, for every $\ell \in A \cup \{min,max\}$, $w,i \models \ell$ holds when $\ell = \wpos{w}{i}$.
  \item {\bf Disjunction:} $w,i \models \psi_1 \vee \psi_2$ when $w,i \models \psi_1$ or $w,i \models \psi_2$.
  \item {\bf Conjunction:} $w,i \models \psi_1 \wedge \psi_2$ when $w,i \models \psi_1$ and $w,i \models \psi_2$.
  \item {\bf Negation:} $w,i \models \neg \psi$ when $w,i \models \psi$ does not hold.
  \item {\bf Next:} $w,i \models \nex{\varphi}$ when $i+1$ is a position of $w$ and $w,i+1 \models \varphi$.
  \item {\bf Preceding:} $w,i \models \nexm{\varphi}$ when $i-1$ is a position of $w$ and  $w,i-1 \models \varphi$.
  \item {\bf Finally:} for $L \subseteq A^*$, we let $w,i \models \finallyl{\psi}$ when there exists a position $j > i$ of $w$ such that $w,j \models \psi$ and $\infix{w}{i}{j} \in L$.
  \item {\bf Previously:} for $L \subseteq A^*$, we let $w,i \models \finallyml{\psi}$ when there exists a position $j < i$ of $w$ such that $w,j \models \psi$ and $\infix{w}{j}{i} \in L$.
\end{itemize}
When no distinguished position is specified, we evaluate formulas at the \emph{leftmost unlabeled position}. More precisely, given a \tlxs formula $\varphi$ and a word $w\in A^*$, we write $w\models\varphi$  and say that \emph{$w$ satisfies $\varphi$} if and only if $w,0 \models \varphi$. Finally, the \emph{language defined by $\varphi$} is $L(\varphi) = \{w \in A^* \mid w \models \varphi\}$. Of course, considering \emph{all} formulas is not really interesting, as every language $L$ is defined by ``$\finallyl{\mathit{max}}$''.

Consider a class \Cs. We first explain how \Cs can be used to restrict the sets of \tls and \tlxs formulas. A \tla{\Cs} (resp. \tlxa{\Cs}) formula  is a \tls (resp. \tlxs) formula $\varphi$  such that every temporal operator $\textup{F}_L$ or $\textup{P}_L$ occurring in $\varphi$ satisfies $L \in \Cs$. With this definition in hand, we associate two classes to \Cs, which we denote by \tlc{\Cs} and \tlxc{\Cs}. They consist of all languages that can be defined by a \tla{\Cs} and a \tlxa{\Cs} formula, respectively. Observe that by definition, Boolean connectives can be used freely in \tls and \tlxs formulas, Hence, it is immediate that the two associated classes are Boolean algebras. They both contain \Cs: a language $L\in\Cs$ is defined by the \tla{\Cs} formula ``$\finallyl{\mathit{max}}$''.  One may also prove that if \Cs is a \vari, then so are \tlc{\Cs} and \tlxc{\Cs}. Yet, we shall not use this property.

\begin{rem} \label{rem:parcase}
  It is immediate that for an arbitrary \tlxs formula $\varphi$, $w \in A^*$ and $i \in \pos{w}$, we have $w,i \models \finallyp{A^*}{\varphi}$ (resp. $w,i \models \finallymp{A^*}{\varphi}$) if and only if there exists $j \in \pos{w}$ such that $i < j$ (resp. $j < i$) and $w,j \models \varphi$. Hence, $\textup{F}_{A^*}$ and $\textup{P}_{A^*}$ correspond to standard operators in unary temporal logic, usually denoted by $\textup{F}$ and $\textup{F}\inv$. In particular for the input class $\stzer = \{\emptyset,A^*\}$, the classes \tlc{\stzer} and \tlxc{\stzer} correspond to the standard variants of unary temporal logic, which are often denoted by $\textup{F}+\textup{F}\inv$ and  $\textup{F}+ \textup{X}+\textup{F}\inv+\textup{X}\inv$. For the sake of consistency, we write  $\textup{P}$ (for ``past'') instead of $\textup{F}\inv$ and $\textup{Y}$ (for ``yesterday'') instead of~$\textup{X}\inv$.
\end{rem}

We now prove that when \Cs is a \vari, the operator  $\Cs \mapsto \tlc{\Cs}$ is ``more fundamental'' than $\Cs \mapsto \tlxc{\Cs}$.  More precisely, we show that if \Cs is a \vari, then $\tlxc{\Cs} = \tlc{\Cs^+}$. The takeaway is that while considering \tlc{\Cs} and \tlxc{\Cs} independently is natural when presenting statements, it suffices to consider \tlc{\Cs} in proof arguments.

\begin{lem} \label{lem:utlx}
  Let \Cs be a \vari. Then, $\tlxc{\Cs} = \tlc{\Cs^+}$.
\end{lem}

\newcommand{\bfor}[1]{\ensuremath{\langle #1\rangle}\xspace}

\begin{proof}
  We start with $\tlxc{\Cs} \subseteq \tlc{\Cs^+}$. By definition, it suffices to show that $\textup{X}$ and $\textup{Y}$ can be expressed with the operators available in \tla{\Cs^+} formulas. This is immediate since $\{\veps\} \in \Cs^+$ and the operators $\textup{F}_{\{\veps\}}$ and $\textup{P}_{\{\veps\}}$ have the same semantics as $\textup{X}$ and $\textup{Y}\inv$.

  We turn to the inclusion $\tlc{\Cs^+} \subseteq \tlxc{\Cs}$. Given an arbitrary \tla{\Cs^+} formula $\varphi$, we explain how to inductively construct an equivalent \tlxa{\Cs} formula $\bfor{\varphi}$ (\emph{i.e.}, for every word $w \in A^*$ and every position $i \in \pos{w}$, we have $w,i \models \varphi \Leftrightarrow w,i \models \bfor{\varphi}$). By definition, this yields $\tlc{\Cs^+} \subseteq \tlxc{\Cs}$, as desired. If $\varphi$ is atomic, it is already a \tlxa{\Cs} formula and we may define $\bfor{\varphi} := \varphi$. Boolean combinations are handled in the natural way. By definition of \tla{\Cs^+} formulas, it remains to handle the case where $\varphi$ is of the form \finallyl{\psi} or \finallyml{\psi} for some $L \in \Cs^+$. By symmetry, we only consider the case where $\varphi := \finallyl{\psi}$. By definition of $\Cs^+$, there exists $K \in \Cs$ such that either $L = \{\veps\} \cup K$ or $L=A^+\cap K$. In the former case, it suffices to define $\bfor{\varphi} := (\nex{\bfor\psi}) \vee (\finallyp{K}{\bfor\psi})$, which is a \tlxa{\Cs} formula by definition. Otherwise, we have $L = A^+ \cap K$. Since \Cs is a \vari, we know that $a\inv K \in \Cs$ for every $a \in A$. Hence, we may define $\bfor{\varphi}$ as the following \tlxa{\Cs} formula:
  \[
    \bfor{\varphi} := \nex{\Big(\bigvee_{a \in A}\left(a \wedge \finallyp{a\inv K}{\bfor\psi}\right)\Big)}.
  \]
  This concludes the proof.
\end{proof}

We complete the definition with a property of the classes \tlc{\Cs} (when \Cs is a \vari) which we shall need in Section~\ref{sec:logcar} to establish the correspondence with unambiguous polynomial closure. This involves quite a bit of work as we require some machinery to present it. First, we define equivalence relations that we shall use to formulate the property.

\subparagraph{Canonical equivalences.} We start with some terminology needed for the definition. Given a morphism $\eta: A^*\to N$ into a finite monoid $N$, a \tla{\eta} formula is a \tls formula $\varphi$  such that for every operator $\textup{F}_L$ or $\textup{P}_L$ occurring in $\varphi$, the language $L \subseteq A^*$ is recognized by $\eta$. The following simple fact connects this notion to the classes \tlc{\Cs}.

\begin{fct} \label{fct:esuit}
  Let \Cs be a \vari. For every \tla{\Cs} formula $\varphi$, there exists a \Cs-morphism $\eta: A^* \to N$ such that $\varphi$ is a \tla{\eta} formula.
\end{fct}

\begin{proof}
  Let \Lb be the finite set consisting of all languages $L \subseteq A^*$ such that either $\textup{F}_L$ or $\textup{P}_L$ occurs in $\varphi$. Since $\varphi$ is a \tla{\Cs} formula,  we know that every $L \in \Lb$ belongs to~$\Cs$. By hypothesis on \Cs, it follows from Proposition~\ref{prop:genocm} that there exists a \Cs-morphism $\eta: A^* \to N$ recognizing every $L \in \Lb$. By definition, $\varphi$ is a \tla{\eta} formula.
\end{proof}

We now associate a number called \emph{rank} to every \tls formula $\varphi$ (this is a standard notion in unary temporal logic). As expected, the rank of $\varphi$ is defined as the length of the longest sequence of nested temporal operators within its parse tree. More precisely,
\begin{itemize}
  \item Any atomic formula has rank $0$.
  \item The rank of $\neg \varphi$ is the same as the rank of $\varphi$.
  \item The rank of $\varphi \vee \psi$ and $\varphi \wedge \psi$ is the maximum between the ranks of $\varphi$ and $\psi$.
  \item For every language $L \subseteq A^*$, the rank of \finallyl{\varphi} and \finallyml{\varphi} is the rank of $\varphi$ plus $1$.
\end{itemize}

Two \tls formulas $\varphi$ and $\psi$ are \emph{equivalent} if they have the same semantics. That is, for every $w \in A^*$ and every position $i \in \pos{w}$, we have $w,i \models \varphi \Leftrightarrow w,i \models \psi$. The following key lemma is immediate from a simple induction on the rank of \tls formulas.

\begin{lem}\label{lem:rank}
  Let $\eta: A^* \to N$ be a morphism into a finite monoid and let $k \in \nat$. There are only finitely many non-equivalent \tla{\eta} formulas with rank at most $k$.
\end{lem}

We may now define the equivalences associated to \tls. They relate pairs $(w,i)$, where $w \in A^*$ and $i \in \pos{w}$. Let $\eta: A^* \to N$ be a morphism into a finite monoid and let $k \in \nat$. Given $w,w' \in A^*$, $i \in \pos{w}$ and $i'\in \pos{w'}$, we write, $w,i \tleqk w',i'$ when:
\[
  \text{For every \tla{\eta} formula $\varphi$ with rank at most $k$,} \quad w,i \models \varphi \Longleftrightarrow w',i' \models \varphi.
\]
Clearly, the relations \tleqk are equivalences. Moreover, it is immediate from the definition and Lemma~\ref{lem:rank}, that they have finite index. Finally, we also introduce equivalences which compare single words in $A^*$. Abusing terminology, we also write them \tleqk. Given $w,w' \in A^*$, we write $w \tleqk w'$ if $w,0 \tleqk w',0$. Clearly, the relation $\tleqk$ is an equivalence on $A^{*}$.

We complete this definition with a useful lemma. It presents an alternative definition of the equivalences \tleqk, which is convenient in proof arguments, using induction on $k$.

\begin{lem}\label{lem:efgame}
  Let $\eta: A^* \to N$ be a morphism into a finite monoid, $k \in \nat$, $w,w' \in A^*$, $i \in \pos{w}$ and $i' \in \pos{w'}$. Then, $w,i \tleqk w',i'$ if and only the five following conditions~hold:
  \begin{itemize}
    \item We have $\wpos{w}{i} = \wpos{w'}{i'}$.
    \item If $k \geq 1$, then for every $j \in \pos{w}$ such that $i < j$, there exists $j' \in \pos{w'}$ such that $i' < j'$, $\eta(\infix{w}{i}{j}) = \eta(\infix{w'}{i'}{j'})$ and $w,j \tleqp{k-1} w',j'$.
    \item If $k \geq 1$, then for every $j' \in \pos{w'}$ such that $i'<j'$, there exists $j \in \pos{w}$ such that $i < j$, $\eta(\infix{w}{i}{j}) = \eta(\infix{w'}{i'}{j'})$ and $w,j \tleqp{k-1} w',j'$.
    \item If $k \geq 1$, then for every $j \in \pos{w}$ such that $j < i$, there exists $j' \in \pos{w'}$ such that $j' < i'$, $\eta(\infix{w}{j}{i}) = \eta(\infix{w'}{j'}{i'})$ and $w,j \tleqp{k-1} w',j'$.
    \item If $k \geq 1$, then for every $j' \in \pos{w'}$ such that $j'<i'$, there exists $j \in \pos{w}$ such that $j < i$, $\eta(\infix{w}{j}{i}) = \eta(\infix{w'}{j'}{i'})$ and $w,j \tleqp{k-1} w',j'$.
  \end{itemize}
\end{lem}

\begin{proof}
  We start with the ``only if'' implication. Assume that $w,i \tleqk w',i'$. We show that the five conditions in the lemma hold. For the first one, we know that for every $\ell \in A \cup \{min,max\}$, ``$\ell$'' is an atomic \tla{\eta} formula, hence of rank~0. Therefore, our hypothesis implies that $w,i \models \ell \Leftrightarrow w,i' \models \ell$ for every $\ell \in A \cup \{min,max\}$. This exactly says that $\wpos{w}{i} = \wpos{w'}{i'}$ and the first condition is proved. We turn to the four remaining conditions. By symmetry, we only detail one of them: we consider the second condition in the lemma. Hence we assume that $k \geq 1$ and consider $j \in \pos{w}$ such that $i < j$. We have to exhibit $j' \in \pos{w'}$ such that $i' < j'$, $\eta(\infix{w}{i}{j}) = \eta(\infix{w'}{i'}{j'})$ and $w,j \tleqp{k-1} w',j'$. We use $w$ and $j$ to build a \tla{\eta} formula. Lemma~\ref{lem:rank} yields a finite set $S$ of \tla{\eta} formulas of rank at most $k-1$ such that every \tla{\eta} formula $\psi$ of rank at most $k-1$ is equivalent to some formula in $S$. Consider the set $T = \{\psi \in S \mid w,j \models \psi\}$. We define,
  \[
    \varphi = \bigwedge_{\psi \in T} \psi.
  \]
  Moreover, let $s = \eta(\infix{w}{i}{j}) \in N$ (recall that $i < j$) and let $L = \eta\inv(s)$. Clearly, $w,j \models \varphi$ and it follows that $w,i \models \finallyl{\varphi}$. Moreover, $\varphi$ has rank at most $k-1$ by definition, which means that $\finallyl{\varphi}$ has rank at most $k$. Since $w,i \tleqk w',i'$, it follows that $w',i' \models \finallyl{\varphi}$. This yields $j' \in \pos{w'}$ such that $i' < j'$, $\infix{w'}{i'}{j'} \in L$, \emph{i.e.}, $\eta(\infix{w'}{i'}{j'}) = \eta(\infix{w}{i}{j})$ and $w',j' \models \varphi$. It remains to show that $w,j \tleqp{k-1} w',j'$.  Let $\psi$ be an $\eta$-formula of rank at most $k-1$. We have to show that $w,j \models \psi \Leftrightarrow w',j'  \models \psi$. For the left to right implication, if $w,j \models \psi$, then we know that $T$ contains a formula equivalent to $\psi$.  Recall that $w',j' \models \varphi$. By definition of $\varphi$, this yields $w',j' \models \psi$. For the converse implication, we prove the contrapositive. Assume that $w,j \not\models \psi$, \emph{i.e.}, $w,j \models \neg \psi$. Since $\neg \psi$ has rank at most $k-1$, we know that $T$ contains a formula equivalent to $\neg \psi$. Again by choice of $\varphi$, since $w',j' \models \varphi$, we also have $w',j' \models \neg \psi$, which implies that $w',j' \not\models \psi$, completing the proof.

  Conversely, assume that the five conditions in the lemma are satisfied. We prove that $w,i \tleqk w',i'$. Given an arbitrary \tla{\eta} formula $\varphi$ of rank at most $k$, we have to prove that $w,i \models \varphi \Leftrightarrow w',i'  \models \varphi$. We proceed by structural induction on $\varphi$. First, if $\varphi$ is an atomic formula then either  $\varphi = \top$, $\varphi = \bot$ or $\varphi = \ell$ for some $\ell \in A \cup \{min,max\}$. In the first two  cases, the result is trivial. In the last one, the first condition in the lemma states that $\wpos{w}{i} = \wpos{w'}{i'}$. Hence, we have $w,i \models \ell \Leftrightarrow w',i' \models \ell$ for every $a \in A \cup \{min,max\}$. We now consider Boolean connectives. If $\varphi = \psi_1 \vee \psi_2$, it follows from structural induction on $\varphi$ that $w,i \models \psi_1 \Leftrightarrow w',i' \models \psi_1$ and $w,i \models \psi_2 \Leftrightarrow w',i' \models \psi_2$.  It is then immediate that $w,i \models \varphi \Leftrightarrow w',i' \models \varphi$. Similarly, if $\varphi = \neg \psi$, it follows from structural induction on~$\varphi$ that we have $w,i \models \psi \Leftrightarrow w',i' \models \psi$, whence immediately, that $w,i \models \varphi \Leftrightarrow w',i' \models \varphi$. It remains to treat the temporal operators, \emph{i.e.}, the cases where $\varphi = \finallyl{\psi}$ or $\varphi = \finallyml{\psi}$ where $L \subseteq A^*$ is recognized by $\eta$ (recall that $\varphi$ is a \tla{\eta} formula). Note that since $\varphi$ has rank at most $k$, these cases may only happen when $k \geq 1$ and $\psi$ has rank at most $k-1$. By symmetry, we only detail the case where $\varphi = \finallyl{\psi}$. Moreover, we only prove the implication $w,i \models \varphi \Rightarrow w',i' \models \varphi$, as the other one is proved symmetrically. Hence, we assume that $w,i \models \varphi$. Since $\varphi = \finallyl{\psi}$, it follows that there exists a position $j$ in $w$ such that $i < j$, $\infix{w}{i}{j} \in L$ and $w,j \models \psi$. The second condition in the lemma yields a position $j'$ of $w'$ such that $i' < j'$, $\eta(\infix{w}{i}{j}) = \eta(\infix{w'}{i'}{j'})$ and $w,j \tleqp{k-1} w',j'$.  Since $L$ is recognized by $\eta$ and $\infix{w}{i}{j} \in L$, the equality $\eta(\infix{w}{i}{j}) = \eta(\infix{w}{i}{j})$ implies that $\infix{w'}{i'}{j'} \in L$. Moreover, since~$\psi$ has rank at most $k-1$ and both $w,j \models \psi$ and $w,j \tleqp{k-1} w',j'$ hold, we obtain $w',j' \models \psi$. Altogether, since $i' < j'$, it follows that $w',i'\models\finallyl{\psi}$, \emph{i.e.}, that $w',i' \models \varphi$, as desired. This concludes the~proof.
\end{proof}

It can be verified from Lemma~\ref{lem:efgame} and a straightforward induction that the equivalences~\tleqk on $A^*$ are congruences. The detailed proof is left to the reader.

\begin{lem} \label{lem:efcong}
  Consider a morphism $\eta: A^* \to N$ into a finite monoid and $k \in \nat$. For every $u,v,u',v' \in A^*$ such that $u \tleqk u'$ and $v \tleqk v'$, we have $uv \tleqk u'v'$.
\end{lem}

We are ready to present the property of the equivalences \tleqk that we shall need in Section~\ref{sec:logcar} to establish the correspondence with unambiguous polynomial closure.

\begin{prop} \label{prop:efg}
  Consider a morphism $\eta: A^* \to N$ into a finite monoid, let $e \in E(N)$ be an idempotent and let $u,v,z \in \eta\inv(e)$. For every $k \in \nat$, the following property holds:
  \[
    (z^{k}uz^{2k}vz^{k})^{k}(z^{k}uz^{2k}vz^{k})^{k} \tleqk (z^{k}uz^{2k}vz^{k})^{k} z^kvz^k(z^{k}uz^{2k}vz^{k})^{k}.
  \]
\end{prop}

\begin{proof}
  We first show that we may restrict ourselves to the special case where $u,v$ and~$z$ are single letters. To this aim, we consider a second independent alphabet $B= \{a,b,c\}$ and we define $\alpha: B^* \to A^*$ as the morphism given by $\alpha(a)= u$, $\alpha(b) = v$ and $\alpha(c) = z$. Finally, we let $\delta=\eta \circ \alpha: B^* \to N$. We have the following fact.

  \begin{fct} \label{fct:morph}
    For all $w,w' \in B^*$ and all $k \in \nat$, if $w \tldk w'$, then $\alpha(w) \tleqk \alpha(w')$.
  \end{fct}

  \begin{proof}
    Let $n = max(\{|u|,|v|,|z|\})$. For every $x \in B^*$ and every $h \leq n$, we define a \emph{partial} map $f_{x,h}: \pos{x} \to \pos{\alpha(x)}$. We fix $i \in \pos{x}$ for the definition. If $i = 0$, then $f_{x,h}(0) = 0$ and if $i = |x|+1$, then $f_{x,h}(|x|+ 1) = |\alpha(x)|+1$. Assume now that $1 \leq i \leq |x|$ and let $d = \wpos{x}{i} \in B$. If $h > |\alpha(d)|$, then $f_{x,h}(i)$ is undefined. Otherwise $h \leq |\alpha(d)|$ and position $i$ of $x$ corresponds to a factor $\alpha(d)$ in $\alpha(x)$ which is made of more than $h$ positions. We let $f_{x,h}(i)$ be $h$-th position of this factor. Formally, this boils down to defining $f_{x,h}(i) = |\alpha(\prefix{x}{i})|+h$.

    One may verify using Lemma~\ref{lem:efgame} and induction on $k \in \nat$ that for every $w,w' \in B^*$, every $i \in \pos{w}$ and every $i' \in \pos{w'}$, if $w,i \tldk w',i'$, then for all $h \leq n$, $f_{w,h}(i)$ is defined if and only if $f_{w',h}(i')$ is defined and, in this case, $\alpha(w),f_{w,h}(i) \tleqk \alpha(w'),f_{w',h}(i')$. In particular, it follows that if $w \tleqk w'$ (\emph{i.e.}, $w,0 \tleqk w',0$), then $\alpha(w) \tleqk \alpha(w')$ as~desired.
  \end{proof}

  Note that by hypothesis on $u,v$ and $z$, we have $\delta(a) = \delta(b) = \delta(c) = e \in E(N)$, which means that for every nonempty word $w \in B^+$, we have $\delta(w) = e$ (on the other hand, $\delta(\veps) = 1_N$, which might be distinct from $e$). We prove that for every $k \in \nat$, we have,
  \begin{equation} \label{eq:efg}
    (c^{k}ac^{2k}bc^{k})^{k}(c^{k}ac^{2k}bc^{k})^{k} \tldk (c^{k}ac^{2k}bc^{k})^{k} c^kbc^k(c^{k}ac^{2k}bc^{k})^{k}.
  \end{equation}
  By Fact~\ref{fct:morph} and  by definition of $\alpha$, this will yield the desired property of Proposition~\ref{prop:efg}.

  To prove~\eqref{eq:efg}, we establish by induction a more general property. We fix $k$ for the proof, and we write $x = c^kac^k$ and $y = c^kbc^k$.  Let $n \in \nat$ and consider a quadruple $(w,i,w',i')$ where $w,w' \in B^*$, $i \in \pos{w}$ and $i' \in \pos{w'}$. We say that $(w,i,w',i')$ is an \emph{$n$-candidate} if there exist $h,\ell \geq n$ and $w_1,w'_1 \in (x+y)^*$ such that $w = (xy)^h w_1(xy)^\ell$, $w' = (xy)^h w'_1(xy)^\ell$ and one of the three following properties is satisfied:
  \begin{enumerate}
    \item\label{itm:c1} We have $i \leq |(xy)^h|$, $i' \leq |(xy)^h|$ and $i = i'$.
    \item\label{itm:c2} We have $|(xy)^hw_1| < i$, $|(xy)^hw'_1| < i'$ and $i- |(xy)^hw_1| = i' - |(xy)^hw'_1|$.
    \item\label{itm:c3} We have $|(xy)^h| < i \leq |(xy)^hw_1|$,  $|(xy)^h| < i' \leq |(xy)^hw'_1|$, and the following infixes are equal: $\infix{w}{i-n-1}{i+n+1} = \infix{w'}{i'-n-1}{i'+n+1}$.
  \end{enumerate}
  We use  Lemma~\ref{lem:efgame} and induction on $n$ to prove that for every $n \leq k$ and every $n$-candidate $(w,i,w',i')$, we have $w,i \tldp{n} w',i'$. Since it is clear that $((xy)^k(xy)^k,0,(xy)^ky(xy)^k,0)$ is a $k$-candidate (as Condition~\ref{itm:c1} in the definition holds), it will follow that the equivalence $(xy)^k(xy)^k \tldk (xy)^ky(xy)^k$ also holds. By definition of $x$ and $y$, this is exactly~\eqref{eq:efg}.

  Fix $n \leq k$ and an $n$-candidate $(w,i,w',i')$ for the proof. We have to show that $w,i\tldp{n} w',i'$. It is clear from the definition of an $n$-candidate that $i$ and $i'$ are either both unlabeled or share the same label in $B$. By Lemma~\ref{lem:efgame}, this concludes the proof when $n = 0$: we get $w,i \tldp{0} w', i'$. We now assume that $n \geq 1$. In view of Lemma~\ref{lem:efgame}, there are four additional conditions to prove in this case. By symmetry, we only prove the first one and leave the others to the reader. Consider a position $j$ of $w$ such that $i < j$. We have to exhibit a position $j'$ of $w'$ such that $i' < j'$, $\eta(\infix{w}{i}{j}) = \eta(\infix{w'}{i'}{j'})$ and $w,j \tldp{n-1} w',j'$. Observe that for the latter property (\emph{i.e.}, $w,j \tldp{n-1} w',j'$), it suffices to choose $j'$ such that $(w,j,w',j')$ is an $(n-1)$-candidate: by induction on $n$, this implies that $w,j \tleqp{n-1} w',j'$. Since $(w,i,w',i')$ is an $n$-candidate, there exist $h,\ell \geq n$ and $w_1,w'_1 \in (x+y)^*$ such that $w=(xy)^h w_1(xy)^\ell$, $w' = (xy)^h w'_1(xy)^\ell$ and one of the three conditions in the definition holds. We consider several cases depending on the position $j$ inside $w$.

  Assume first that $j \leq |(xy)^h|$. Hence, we have $i < j \leq |(xy)^h|$. It follows that Condition~\ref{itm:c1} in the definition of the $n$-candidate $(w,i,w',i')$ holds. Therefore, we have  $i' = i  < |(xy)^h|$. It now suffices to define $j' = j$. Clearly,  $i' < j'$ and $\infix{w'}{i'}{j'} = \infix{w}{i}{j}$. Moreover, one can check that $(w,j,w',j')$ is an $(n-1)$-candidate (in fact, it is even an $n$-candidate in this case) as Condition~\ref{itm:c1} in the definition holds.

  We now assume that $|(xy)^h| < j \leq |(xy)^hw_1xy|$. We consider two subcases depending on whether $i+1 = j$ or $i+1 < j$. Assume first that $i+1 = j$. In this case, we define $j' = i' + 1$. Clearly, we have $i' < j'$ and $\infix{w'}{i'}{j'} = \infix{w}{i}{j} = \veps$. Hence, we have to verify that $(w,j,w',j')$ is an $(n-1)$-candidate. We show that Condition~\ref{itm:c3} holds. By definition, $h \geq n-1$, $\ell-1 \geq n-1$, $w = (xy)^hw_1xy(xy)^{\ell-1}$ and $w' = (xy)^hw'_1xy(xy)^{\ell-1}$. Moreover, $|(xy)^h| < j \leq |(xy)^hw_1xy|$ and it can be verified from the definition of $j'$ that $|(xy)^h| < j' \leq |(xy)^hw'_1xy|$. Finally, one can also check that $\infix{w}{j-n}{j+n} = \infix{w'}{j-n}{j'+n}$ from the hypothesis that $(w,i,w',i')$ is an $n$-candidate, since $j = i+1$ and $j' = i'+1$. We conclude that $(w,j,w',j')$ is an $(n-1)$-candidate as Condition~\ref{itm:c3} holds. We turn to the second subcase: assume that $i+1 < j$. By hypothesis, $i' \leq |(xy)^hw'_1|$. Hence,  by definition of $x,y$ and since $n \leq k$, one can verify that there exists a position $j' \in \pos{w'}$ such that  $|(xy)^h| < j' \leq |(xy)^hw'_1xy|$,  $i' +1 < j'$ and $\infix{w}{j-n}{j+n} = \infix{w'}{j-n}{j'+n}$. We have $i' < j'$ by definition. Moreover, since $i+1 < j$ and $i'+1 < j'$, we know that $\infix{w'}{i'}{j'}$ and $\infix{w}{i}{j}$ are nonempty, which yields $\delta(\infix{w'}{i'}{j'}) = \delta(\infix{w}{i}{j}) = e$ by definition of $\delta$. Finally, it is immediate that $(w,j,w',j')$ is an $(n-1)$-candidate since Condition~\ref{itm:c3} again holds.

  Finally, assume that $|(xy)^h w_1xy| < j$. In this case, let $j'$ be the unique position of~$w'$ such that $|(xy)^h w'_1xy| < j'$ and $j - |(xy)^hw_1xy| = j' - |(xy)^hw'_1xy|$. Since $i < j$ and $(w,i,w',i')$ is an $n$-candidate, one can check that $i' < j'$ and that either $i+1 = j$ and $i'+1 = j'$ or $i+1 < j$ and $i'+1 < j'$, which implies that $\delta(\infix{w'}{i'}{j'}) = \delta(\infix{w}{i}{j})$. Finally, it is straightforward to verify that $(w,j,w',j')$ is a $(n-1)$-candidate (in fact, it is even an $n$-candidate in this case) as Condition~\ref{itm:c2} in the definition holds, completing the~proof.
\end{proof}

\subsection{Connection with two-variable first-order logic}

We now prove the generic correspondence existing between two-variable first-order logic and unary temporal logic: we prove that $\fod(\infsigc) = \tlc{\Cs}$ for every Boolean algebra \Cs. As we announced, this generalizes a theorem by Etessami, Vardi and Wilke~\cite{evwutl}, which states that $\fodw = \fpfm$ (which is the particular case $\Cs = \stzer$) and $\fodws =\fpfmx$ (which is the particular case $\Cs= \stzer^+$).

\begin{theorem} \label{thm:temporal}
  For every Boolean algebra \Cs, we have $\fod(\infsigc) = \tlc{\Cs}$.
\end{theorem}

\begin{proof}
  We first prove the inclusion $\tlc{\Cs} \subseteq \fod(\infsigc)$. We have to prove that for every \tla{\Cs} formula, there exists a sentence of $\fod(\infsigc)$ defining the same language. Consider a \tla{\Cs} formula $\varphi$. We use structural induction to construct a formula $\bfor{\varphi}(x)$ of $\fod(\infsigc)$ with at most one free variable ``$x$'', which satisfies the following property:
  \begin{equation} \label{eq:utl1}
    \text{Given $w \in A^*$ and a position $i \in \pos{w}$,} \quad w \models \bfor{\varphi}(i)\  \Longleftrightarrow\ w,i \models \varphi.
  \end{equation}
  Clearly, the language defined by $\varphi$ is also defined by the sentence ``$\bfor{\varphi}(min)$'' of $\fod(\infsigc)$. Hence, this proves $\tlc{\Cs} \subseteq \fod(\infsigc)$, as announced.

  It remains to construct $\bfor{\varphi}(x)$ from $\varphi$, which we do using structural induction on the \tla{\Cs} formula $\varphi$. Note that we only present the construction. That it satisfies~\eqref{eq:utl1} is straightforward to verify and left to the reader. First, if $\varphi := \top$ or $\varphi := \bot$, then we define $\bfor{\varphi}(x) := \top$ and $\bfor{\varphi}(x) := \bot$ respectively. Moreover, if $\varphi := min$ or $\varphi := max$, then we define $\bfor{\varphi}(x) := (x = min)$ and $\bfor{\varphi}(x) := (x = max)$ respectively. If $\varphi := a$ for some letter $a \in A$, then we define $\bfor{\varphi}(x) := a(x)$. Assume now that $\varphi := \finallyl{\psi}$ for some $L \in \Cs$ and a smaller \tla{\Cs} formula $\psi$. In this case, we define $\bfor{\varphi}(x) := \exists y\ I_L(x,y) \wedge \bfor{\psi}(y)$. Finally, when $\varphi := \finallyml{\psi}$ for some $L \in \Cs$ and a smaller \tla{\Cs} formula $\psi$, we  define $\bfor{\varphi}(x) := \exists y\ I_L(y,x) \wedge \bfor{\psi}(y)$. This concludes the proof of the right to left inclusion in~Theorem~\ref{thm:temporal}.

  \medskip

  \newcommand{\crh}[1]{\ensuremath{[#1]}\xspace}

  We turn to the inclusion $\fod(\infsigc) \subseteq \tlc{\Cs}$. Let $\varphi(x)$ be an $\fod(\infsigc)$ formula with at most one free variable $x$. We prove that there exists a \tla{\Cs} formula \crh{\varphi} satisfying the following property:
  \begin{equation}\label{eq:translatfo2}
    \text{Given $w \in A^*$ and a position $i \in \pos{w}$,} \quad w,i \models \crh{\varphi}\  \Longleftrightarrow\ w \models \varphi(i).
  \end{equation}
  Before we present the construction, let us explain why this implies $\fod(\infsigc) \subseteq \tlc{\Cs}$. Let $L \in \fod(\infsigc)$. By definition, $L$ is defined by a sentence $\varphi$ of $\fod(\infsigc)$. Since $\varphi$ has no free variables, the \tla{\Cs} formula \crh{\varphi} given by~\eqref{eq:translatfo2} satisfies $w,i \models \crh{\varphi} \Leftrightarrow w \models \varphi$ for every $w \in A^*$ and every $i \in \pos{w}$. In particular, $\crh{\varphi}$ defines $L$ and we get $L \in \tlc{\Cs}$.

  \smallskip
  It remains to prove that for every formula $\varphi(x)$ of $\fod(\infsigc)$ with at most one free variable~$x$, there exists a \tla{\Cs} formula \crh{\varphi}  satisfying~\eqref{eq:translatfo2}. We proceed by induction on the size of $\varphi(x)$. We start with the base case: $\varphi(x)$ is an atomic formula. There are several cases. First, if $\varphi(x) := \top$ or $\varphi(x) := \bot$, then we define $\crh{\varphi} := \top$ and $\crh{\varphi} := \bot$ respectively. If $\varphi := a(x)$ for some letter $a \in A$, then we define $\crh{\varphi} := a$. If $\varphi := a(min)$ or $\varphi := a(max)$ for some letter $a \in A$, then we define $\crh{\varphi} := \bot$. We now consider the atomic formulas involving equality. Since there is only one free variable $x$ and equality is commutative, this boils down to four cases.  If $\varphi := (x = x)$, we define $\crh{\varphi} := \top$. If $\varphi := (x = min)$, we define $\crh{\varphi} := min$. If $\varphi := (x = max)$, we define $\crh{\varphi} := max$. Finally, if $\varphi := (min = max)$, we define $\crh{\varphi} := \bot$. It now remains to consider the atomic formulas involving a binary predicate $I_L$ for some $L \in \Cs$. Since $x$ is the only variable which is (possibly) free in $\varphi(x)$, we have to consider the following cases. If $\varphi$ is either $I_L(x,x)$, $I_L(max,x)$, $I_L(x,min)$ or $I_L(max,min)$, we simply define $\crh{\varphi} := \bot$. If we have $\varphi := I_L(min,x)$, we define $\crh{\varphi} := \finallyml{min}$. If we have $\varphi := I_L(x,max)$, we define $\crh{\varphi} := \finallyl{max}$. Finally, when $\varphi := I_L(min,max)$, we define $\crh{\varphi} := \left(min \wedge \finallyl{max}\right) \vee \finallymp{A^*}{\left(min \wedge \finallyl{max}\right)}$ (note that $A^* \in \Cs$ since \Cs is a Boolean algebra). This concludes the case of atomic formulas. The Boolean connectives are handled in the obvious way: we let $\crh{\psi_1 \vee \psi_2} = \crh{\psi_1} \vee \crh{\psi_2}$, $\crh{\psi_1 \wedge \psi_2} = \crh{\psi_1} \wedge \crh{\psi_2}$  and $\crh{\neg \psi} = \neg \crh{\psi}$.

  It remains to treat quantification over a second variable $y$, \emph{i.e.}, $\varphi(x) := \exists y\ \gamma(x,y)$ (since $\forall$ has the same semantics as $\neg \exists \neg$, we may assume without loss of generality that all quantifiers in $\varphi$ are existential). We cannot apply induction to $\gamma$ since it has \emph{two} free variables. To solve this problem, we first need to put $\gamma$ into normal form. First, we define \Lb as the finite set of all languages $L \in \Cs$ such that $\gamma$ contains an atomic formula $I_L(x,y)$ or $I_L(y,x)$ inside which the occurrences of $x$ and $y$ are free. Using \Lb, we define an equivalence $\sim$ over $A^*$: for $u,v \in A^*$, we write $u \sim v$ when $u \in L \Leftrightarrow v \in L$ for every $L \in \Lb$. Since \Lb is finite, $\sim$ has finite index. Moreover, since every $L \in \Lb$ belongs to $\Cs$, which is a Boolean algebra by hypothesis, the $\sim$-classes belong to $\Cs$ as well. We say that a formula is an \emph{\Lb-profile} when it is of the form $I_H(x,y)$, $I_H(y,x)$ or $x = y$ where $H \in \Cs$ is a $\sim$-class. Note that since $\sim$ has finite index, there are finitely many \Lb-profiles.

  Finally, we say that an $\fod(\infsigc)$ formula with at most two free-variables $x$ and $y$ is \Lb-\emph{normalized} if it has the following form:
  \[
    \bigwedge_{1 \leq i \leq n} (\zeta_i(x) \Leftrightarrow t_i)\ \wedge \pi(x,y) \wedge \psi(y),
  \]
  where $n \in \nat$, $\zeta_1(x),\dots,\zeta_n(x)$ are $\fod(\infsigc)$ formulas \emph{smaller than $\gamma$} whose only free variable is $x$, $t_1,\dots,t_n \in \{\top,\bot\}$, $\pi(x,y)$ is an \Lb-profile, and $\psi(y)$ is an $\fod(\infsigc)$ formula \emph{smaller than $\gamma$} whose only free variable is $y$. The argument is based on the following fact.

  \begin{fct}\label{fct:normalized}
    The formula $\gamma(x,y)$ is equivalent to a finite disjunction of\/ \Lb-normalized formulas.
  \end{fct}

  Before we prove Fact~\ref{fct:normalized}, let us use it to complete the construction of $\crh{\varphi}$. Recall that $\varphi(x) :=$ $\exists y\ \gamma(x,y)$. Since disjunctions and existential quantifications commute, it follows from Fact~\ref{fct:normalized} that $\varphi(x)$ is equivalent to a finite disjunction of the form
  \[
    \bigvee_{1 \leq j \leq m} \varphi_j(x).
  \]
  where each formula $\varphi_j(x)$ is of form $\exists y\ \xi_j(x,y)$ with $\xi_j(x,y)$ an \Lb-normalized formula. Therefore, we concentrate on the formulas $\varphi_j$: if we construct \tla{\Cs} formulas $\crh{\varphi_1},\dots,\crh{\varphi_m}$ satisfying~\eqref{eq:translatfo2} for $\varphi_1,\dots,\varphi_m$, it will then suffice to define:
  \[
    \crh{\varphi} = \bigvee_{1 \leq j \leq m} \crh{\varphi_j}.
  \]
  Consider one of the formulas $\varphi_j$. We use induction to build $\crh{\varphi_j}$. By definition of normalized formulas, $\varphi_j(x)$ has the following form:
  \[
    \varphi_j(x) = \exists y\ \bigwedge_{1 \leq i \leq n} (\zeta_i(x) \Leftrightarrow t_i)\ \wedge\pi(x,y) \wedge \psi(y).
  \]
  with $n \in \nat$, $\zeta_1(x),\dots,\zeta_n(x)$ are $\fod(\infsigc)$ formulas \emph{smaller than $\gamma$} whose only free variable is $x$, $t_1,\dots,t_n \in \{\top,\bot\}$, $\pi(x,y)$ an \Lb-profile and $\psi(y)$ is an $\fod(\infsigc)$ formula \emph{smaller than $\gamma$} whose only free variable is $y$. We may therefore apply induction to the subformulas $\zeta_1(x),\dots,\zeta_n(x)$ and $\psi(y)$, which yields $\fod(\infsigc)$ formulas $\crh{\zeta_1},\dots,\crh{\zeta_n}$ and $\crh{\psi}$ satisfying~\eqref{eq:translatfo2}. We consider three cases depending on the \Lb-profile $\pi(x,y)$. First, if $\pi(x,y) :=$ ``$I_H(x,y)$'' for some $\sim$-class $H \in \Cs$, we define,
  \[
    \crh{\varphi_j} = \bigwedge_{1 \leq i \leq n} (\crh{\zeta_i} \Leftrightarrow t_i)\ \wedge \finallyp{H}{\crh{\psi}}.
  \]
  Second, if $\pi(x,y) :=$ ``$I_H(y,x)$'' for some $\sim$-class $H \in \Cs$, we define,
  \[
    \crh{\varphi_j} = \bigwedge_{1 \leq i \leq n} (\crh{\zeta_i} \Leftrightarrow t_i)\ \wedge \finallymp{H}{\crh{\psi}}.
  \]
  Finally, if $\pi(x,y) :=$ ``$x = y$'' for some $\sim$-class $H \in \Cs$, we define,
  \[
    \crh{\varphi_j} = \bigwedge_{1 \leq i \leq n} (\crh{\zeta_i} \Leftrightarrow t_i)\ \wedge \crh{\psi}.
  \]
  One may verify that this definition satisfies~\eqref{eq:translatfo2}, which concludes the proof for the construction of \crh{\varphi}.

  \medskip

  It remains to prove Fact~\ref{fct:normalized}. By definition, $\gamma(x,y)$ is a Boolean combination of atomic formulas and existential quantifications, \emph{i.e.}, formulas of the form $\exists z\ \psi(x,y,z)$. Note that since $\gamma(x,y)$ is \fod and already contains the two variables $x$ and $y$, we know that in the latter case, the quantified variable $z$ is either $x$ or $y$, which means that there remains in the formula a \emph{single} free variable (namely, $x$ if $z=y$ and $y$ if $z=x$). The important point resulting from this discussion is that the only subformulas in the Boolean combination involving \emph{both} $x$ and $y$ as free variables are \emph{atomic}. In summary, we conclude that $\gamma(x,y)$ is a Boolean combination of three kinds of subformulas:
  \begin{enumerate}
    \item Atomic formulas involving \emph{both} $x$ and $y$: by definition, they are of the form $I_L(x,y)$ or $I_L(y,x)$ for $L \in \Lb$.
    \item Formulas $\zeta_1(x),\dots,\zeta_n(x)$ whose only free variable is $x$.
    \item Formulas $\chi_1(y),\dots,\chi_m(y)$ whose only free variable is $y$.
  \end{enumerate}
  Intuitively, the main idea is now to get rid of the first two kinds of subformulas from $\gamma(x,y)$ by making a disjunction over all possible assignments of truth values for $\zeta_1(x),\dots,\zeta_n(x)$ and all possible formulas $I_L(x,y)$ or $I_L(y,x)$ for $L \in \Lb$.

  Let $T = \{\bot,\top\}^n$. Intuitively, we view each tuple $\tau = (t_1,\dots,t_n) \in T$ as an assignment of truth values for the subformulas $\zeta_1(x),\dots,\zeta_n(x)$. Moreover, let $P$ be the set of all \Lb-profiles. By definition, an \Lb-profile $\pi(x,y)$ determines the truth value of every atomic formula $I_L(x,y)$ or $I_L(y,x)$ for $L \in \Lb$. Indeed, when $\pi(x,y) :=$ ``$x = y$''  is satisfied, then all atomic formulas $I_L(x,y)$ or $I_L(y,x)$ are false. When $\pi(x,y) :=$ ``$I_H(x,y)$'' is satisfied for some $\sim$-class $H$, then all formulas $I_L(y,x)$ are false and given $L \in \Lb$, $I_L(x,y)$ is true if and only if $H \subseteq L$ (by definition of the equivalence $\sim$ from \Lb). Symmetrically, when $\pi(x,y) :=$ ``$I_H(y,x)$'' is satisfied for some $\sim$-class $H$, then all formulas $I_L(x,y)$ are false and given $L \in \Lb$, $I_L(y,x)$ is true if and only if $H \subseteq L$. For every $\tau = (t_1,\dots,t_n) \in T$ and $\pi \in P$, we denote by  $\pi(\tau(\gamma(x,y)))$ the formula obtained from $\gamma(x,y)$ by replacing each subformula $\zeta_i(x)$ in the Boolean combination by the truth value $t_i \in \{\bot,\top\}$ and each atomic formula $I_L(x,y)$ or $I_L(y,x)$ for $L \in \Lb$ by its truth value as given by $\pi$. Observe that the only free variable in $\pi(\tau(\gamma(x,y)))$ is $y$ since we replaced all subformulas involving $x$ in the Boolean combination by truth values. One may now verify that $\gamma(x,y)$ is equivalent to the following~formula:
  \[
    \bigvee_{\tau =(t_1,\dots,t_n) \in T} \bigvee_{\pi \in P}\Biggl(\bigwedge_{i \leq n} (\zeta_i(x) \Leftrightarrow t_i)\ \wedge \pi(x,y) \wedge \pi(\tau(\gamma(x,y))) \Biggr).
  \]
  By definition, this is a disjunction of normalized formulas. Note also that each formula $\zeta_i(x)$ is smaller than $\gamma(x,y)$, as well as $\psi(y)=\pi(\tau(\gamma(x,y)))$. This concludes the proof of~Fact~\ref{fct:normalized}.
\end{proof}

\subsection{Connection with unambiguous polynomial closure}

We now prove that if \Cs is a \vari, the inclusion $\upol{\bpol{\Cs}} \subseteq \tlc{\Cs}$ holds. Let us point out that this inclusion is \emph{strict} in general. Yet, we shall prove in the next section that when \Cs is either a \vari of group languages \Gs or its \wsuit extension $\Gs^+$, the inclusion is in fact an equality: in that case, we have $\upol{\bpol{\Cs}} = \tlc{\Cs}$.

\begin{rem}
  The cases when \Cs is either \Gs or $\Gs^+$ for a \vari of group languages \Gs cover most of the situations when \tlc{\Cs} corresponds to a natural logical class. Indeed, as $\tlc{\Cs} = \fod(\infsig{\Cs})$ by Theorem~\ref{thm:temporal}, it follows from the results of Section~\ref{sec:frag} that when \Cs is either \Gs or $\Gs^+$, we obtain the variant $\fod(<,\prefsigg)$ and $\fod(<,+1,\prefsigg)$ of \fod. Yet, there is a special variant of\/ \fod that is not covered by theses cases. It was considered by Krebs, Lodaya, Pandya and Straubing~\cite{between}. It is called ``two-variable first-order logic with between predicates'' and it is simple to verify that it corresponds to $\fod(\infsig{\at})$. It is proved in \cite{between} that this variant has decidable membership but this is based on a specialized proof which is independent of the techniques that we use here.
\end{rem}

We may now prove the generic inclusion $\upol{\bpol{\Cs}}\subseteq\tlc{\Cs}$. The proof argument is based on Theorem~\ref{thm:apol} which yields $\upol{\bpol{\Cs}} = \wadet{\bpol{\Cs}}$. With this equality in hand, it suffices to prove that $\wadet{\bpol{\Cs}} \subseteq \tlc{\Cs}$.

\begin{prop} \label{prop:ubinutl}
  Let \Cs be a \vari. Then, $\upol{\bpol{\Cs}} \subseteq \tlc{\Cs}$.
\end{prop}

\begin{proof}
  By Corollary~\ref{cor:bpolc}, \bpol{\Cs} is a \vari. Hence, $\upol{\bpol{\Cs}} = \wadet{\bpol{\Cs}}$ by Theorem~\ref{thm:apol} and it suffices to prove that $\wadet{\bpol{\Cs}} \subseteq \tlc{\Cs}$. By definition of \wadeto, this boils down to proving that \tlc{\Cs} contains \bpol{\Cs} and is closed under disjoint union and left/right \bpol{\Cs}-deterministic marked concatenation. Closure under union is immediate by definition. For the other properties, we use the following fact.

  \begin{fct} \label{fct:utlg}
    Let $K \in \bpol{\Cs}$. There exist two \tla{\Cs} formulas $\xi_K^\ell$ and $\xi_K^r$ such that for every word $w \in A^*$ and every position $i \in \pos{w}$, we have,
    \[
      \begin{array}{lll}
        w,i \models \xi_K^\ell & \Longleftrightarrow & i \geq 1 \text{ and } \prefix{w}{i} \in K. \\
        w,i \models \xi_K^r & \Longleftrightarrow & i \leq |w| \text{ and } \suffix{w}{i} \in K.
      \end{array}
    \]
  \end{fct}

  \begin{proof}
    We prove the existence of $\xi_K^r$ (the argument for $\xi_K^\ell$ is symmetrical and left to the reader). By definition, $K \in \bpol{\Cs}$ is a Boolean combination of languages $K_0a_1K_1 \cdots a_nK_n$ where $a_1,\dots,a_n  \in A$ and $K_0,\dots,K_n \in \Cs$. Therefore, since we may use Boolean connectives freely in \tla{\Cs} formulas, we may assume without loss of generality that $K$ itself is of the form $K_0a_1K_1 \cdots a_nK_n$. In that case, it suffices to define,
    \[
      \xi_K^r = \finallyp{K_0}{\left(a_1 \wedge  \finallyp{K_1}{\left(a_2 \wedge  \finallyp{K_2}{\left(\quad \cdots\quad  a_n \wedge  \finallyp{K_n}{max}\right)}\right)}\right)}.
    \]
    By definition, $\xi_K^r$ is a \tla{\Cs} formula. Moreover, one can verify that it satisfies the desired property, concluding the proof.
  \end{proof}

  Clearly, Fact~\ref{fct:utlg} implies that \tlc{\Cs} contains $\bpol{\Cs}$: every language $K \in \bpol{\Cs}$ is defined by the \tla{\Cs} formula $\xi_K^r$. It remains to prove that \tlc{\Cs} is closed under left/right \bpol{\Cs}-deterministic marked concatenation. By symmetry, we concentrate on \emph{left} \bpol{\Cs}-deterministic marked concatenation.  Let $H,L \in \tlc{\Cs}$ and $a \in A$ such that $HaL$ is left \bpol{\Cs}-deterministic. We prove that $HaL \in \tlc{\Cs}$. By definition, we have $K \in \bpol{\Cs}$ such that $H \subseteq K$ and $KaL$ is left deterministic. We write $\psi$ for the \tla{\Cs} formula $a \wedge \xi_K^\ell$ (where  $\xi_K^\ell$ is given by Fact~\ref{fct:utlg}).  Since $KaL$ is left deterministic, it follows from Lemma~\ref{lem:detuna} that $KaA^*$ is unambiguous. Hence, by definition of $\psi$, for every word $w \in A^*$, there exists \emph{at most one} position $i \in \pos{w}$ such that $w,i \models \psi$. Since $H \subseteq K$, it follows that for every $w \in A^*$, we have $w \in HaL$ if and only if $w$ satisfies the three~properties:
  \begin{enumerate}
    \item there exists $i \in \pos{w}$ such that $w,i \models \psi$ (the position $i$ is unique as explained above).
    \item The prefix \infix{w}{0}{i} belongs to $H$.
    \item The suffix \infix{w}{i}{|w|+1} belongs to $L$.
  \end{enumerate}
  Therefore, it suffices to prove that these three properties may be expressed using a \tla{\Cs} formula. This is simple for the first property: it is expressed by the formula $\finally{\psi}$. It remains to prove that the other two properties can be defined as well. Let us start with the second one. Since $H \in \tlc{\Cs}$, it is defined by a \tla{\Cs} formula $\varphi_H$. We modify it to construct another formula $\varphi'_H$ expressing the second property. Given $w \in A^*$, we restrict the evaluation of $\varphi_H$ to the positions $i \in \pos{w}$ such that $i$ is either strictly smaller than the unique one satisfying $\psi$ or $i = |w|+1$. More precisely, we build $\varphi'_H$ by applying the two following modifications to $\varphi_H$:
  \begin{enumerate}
    \item We recursively replace each subformula of the form $\finallymp{V}{\zeta}$ with $V \in \Cs$ by,
    \[
      \left(max \wedge \finallym{\left(\psi \wedge \finallymp{V}{\zeta}\right)}\right) \vee \left(\finally{\psi} \wedge \finallymp{V}{\zeta}\right).
    \]
    \item We recursively replace each subformula of the form $\finallyp{V}{\zeta}$ with $V \in \Cs$ by,
    \[
      \left(\finallyp{V}{\left(\zeta \wedge \finally{\psi}\right)}\right) \vee \left(\finallyp{V}{\left(\psi \wedge \finally{\left(max \wedge \zeta\right)}\right)}\right).
    \]
  \end{enumerate}
  It remains to handle the third property. Since $L \in \tlc{\Cs}$, there exists a \tla{\Cs} formula $\varphi_L$ defining $L$.  We construct a formula $\varphi'_L$ expressing  the third property above.

  Given $w \in A^*$, we restrict the evaluation of $\varphi_L$ to the positions $i \in \pos{w}$ such that $i$ is either strictly larger than the unique one satisfying $\psi$ or $i = 0$. More precisely, we build $\varphi'_L$ by applying the two following modifications to $\varphi_L$:
  \begin{enumerate}
    \item We recursively replace each subformula of the form $\finallymp{V}{\zeta}$ with $V \in \Cs$ by,
    \[
      \left(\finallymp{V}{\left(\zeta \wedge \finallym{\psi}\right)}\right) \vee \left(\finallymp{V}{\left(\psi \wedge \finallym{\left(min \wedge \zeta\right)}\right)}\right).
    \]
    \item We recursively replace each subformula of the form $\finallyp{V}{\zeta}$ with $V \in \Cs$ by,
    \[
      \left(min \wedge \finally{\left(\psi \wedge \finallyp{V}{\zeta}\right)}\right) \vee \left(\finallym{\psi} \wedge \finallyp{V}{\zeta}\right).
    \]
  \end{enumerate}
  The language $HaL$ is now defined by the \tla{\Cs} formula $(\finally{\psi}) \wedge \varphi'_H \wedge \varphi'_L$. This concludes the proof of Proposition~\ref{prop:ubinutl}.
\end{proof}

\section{Logical characterizations of unambiguous polynomial closure}
\label{sec:logcar}
This final section details the logical characterizations of unambiguous polynomial closure. Let us first summarize what we already know. By Theorem~\ref{thm:capoldel2}, when \Cs is a \vari, we have  $\dec{2}(\infsigc) = \capol{\bpol{\Cs}}$. Moreover, since \bpol{\Cs} is a \vari by Corollary~\ref{cor:bpolc}, we get from Theorem~\ref{thm:polcopol} that,
\[
  \dec{2}(\infsigc) = \capol{\bpol{\Cs}} = \upol{\bpol{\Cs}}.
\]
Independently, we proved in Theorem~\ref{thm:temporal} that the variant $\fod(\infsigc)$ of two-variable first-order logic corresponds to the unary temporal logic \tlc{\Cs} (this holds as soon as \Cs is a Boolean algebra). Finally, we proved the inclusion $\upol{\bpol{\Cs}} \subseteq \tlc{\Cs}$ in Proposition~\ref{prop:ubinutl} for every \vari \Cs. Altogether, it follows that when \Cs is a \vari, we have
\[
  \dec{2}(\infsigc) = \capol{\bpol{\Cs}} = \upol{\bpol{\Cs}} \subseteq \tlc{\Cs} = \fod(\infsigc).
\]
In general, the inclusion is strict. Yet, we prove in this section that when \Cs is either a \vari of group languages \Gs or its \wsuit extension $\Gs^+$, the converse inclusion holds. Since the presentation of logical classes can be simplified in these cases (see Lemma~\ref{lem:gensig} and Lemma~\ref{lem:utlx}), we get the following generic results for every \vari of group languages \Gs:
\[
  \begin{array}{ccccccc}
    \upol{\bpol{\Gs}} & = & \dec{2}(<,\prefsigg) & = & \fod(<,\prefsigg) & =&  \tlc{\Gs}. \\
    \upol{\bpol{\Gs^+}} & = & \dec{2}(<,+1,\prefsigg) & = & \fod(<,+1,\prefsigg) & =&  \tlxc{\Gs}.
  \end{array}
\]
The proofs of the missing inclusions are based on Theorem~\ref{thm:ubp}, the algebraic characterization of $\capol{\bpol{\Cs}}=\dec{2}(\infsigc)$. We use Proposition~\ref{prop:efg} to prove that it is satisfied by the class \tlc{\Cs} when \Cs is a \vari of group languages \Gs or its \wsuit extension $\Gs^+$. Actually, we simplify the generic characterization presented in Theorem~\ref{thm:ubp} and present specialized characterizations for these two special cases. In particular, this yields generic characterizations of $\fod(<,\prefsigg)$ and $\fod(<,+1,\prefsigg)$. They generalize well-known results for particular instances of these logical classes. Let us point out that here, we directly characterize the languages of these classes (using a property of their syntactic morphism) rather than characterizing the morphisms associated to the class. This is because in this case, we do need to use the characterizations as subresults.

\subsection{Group languages} We first consider the classes \upol{\bpol{\Gs}} where \Gs is a \vari of group languages. The algebraic characterization is based on the class of finite monoids \davar, whose connection with unambiguous polynomial closure~\cite{schul} and two-variable first-order logic~\cite{twfo2} is well known. We use a definition based on an equation (see~\cite{Tesson02diamondsare} for details). A finite monoid $N$ belongs to \davar when it satisfies the following equation:
\begin{equation} \label{eq:da}
  (st)^\omega = (st)^\omega t (st)^\omega \quad \text{for every $s,t \in N$}.
\end{equation}
We are now ready to present the generic characterization for classes of the form \upol{\bpol{\Gs}}.

\begin{theorem}\label{thm:glogic}
  Let \Gs be a \vari of group languages and let $L$ be a regular language. The following properties are equivalent:
  \begin{enumerate}
    \item $L \in \upol{\bpol{\Gs}}$.
    \item $L \in \capol{\bpol{\Gs}}$.
    \item $L \in \dec{2}(<,\prefsigg)$.
    \item $L \in \fod(<,\prefsigg)$.
    \item $L \in \tlc{\Gs}$.
    \item The \Gs-kernel of the syntactic morphism of $L$ belongs to \davar.
  \end{enumerate}
\end{theorem}

\begin{proof}
  By Lemma~\ref{lem:gensig}, we have $\dec{2}(<,\prefsigg) = \dec{2}(\infsigg)$. Therefore, $(2) \Leftrightarrow (3)$ follows from Theorem~\ref{thm:capoldel2}. Moreover, $(1) \Leftrightarrow (2)$ is given by Theorem~\ref{thm:polcopol}. Since $\fod(<,\prefsigg) = \fod(\infsigg)$ by Lemma~\ref{lem:gensig}, the equivalence $(4) \Leftrightarrow (5)$ is immediate from Theorem~\ref{thm:temporal}. Finally, the implication $(1) \Rightarrow (5)$ follows from Proposition~\ref{prop:ubinutl}. We now prove independently that $(6) \Rightarrow (2)$ and $(5) \Rightarrow (6)$, which completes the argument. Let $\alpha: A^* \to M$ be the syntactic morphism of $L$ and let $N \subseteq M$ be its \Gs-kernel.

  Let us start with the implication $(6) \Rightarrow (2)$. We assume that $N$ belongs to \davar and show that $L \in \capol{\bpol{\Gs}}$. Since $L$ is recognized by its syntactic morphism $\alpha$, it suffices to show that $\alpha$ is a $(\capol{\bpol{\Gs}})$-morphism. By Theorem~\ref{thm:ubp},this boils down to proving that for every $s,t \in M$ and $e \in E(M)$ such that $(e,s) \in M^2$ is a \Gs-pair, we have $(eset)^{\omega+1} = (eset)^{\omega}et(eset)^{\omega}$. By Lemma~\ref{lem:ptoker}, we have $es \in N$ and $et(eset)^{2\omega-1} \in N$. Since $N$ belongs to \davar, it follows from~\eqref{eq:da} that,
  \[
    (es et(eset)^{2\omega-1})^\omega = (es et(eset)^{2\omega-1})^\omega et(eset)^{2\omega-1} (es et(eset)^{2\omega-1})^\omega.
  \]
  This exactly says that $(eset)^{\omega} = (eset)^{\omega} et (eset)^{3\omega-1}$. It now suffices to multiply by $eset$ on the right to get $(eset)^{\omega+1}=(eset)^{\omega}et(eset)^{\omega}$, as desired.

  Finally, we prove that $(5) \Rightarrow (6)$. We assume that $L \in \tlc{\Gs}$ and show that $N \in \davar$, \emph{i.e.}, that $N$ satisfies~\eqref{eq:da}. Let $s,t \in N$. We prove that $(st)^\omega = (st)^\omega t (st)^\omega$. By hypothesis,~$L$ is defined by a formula $\varphi_L$ of $\tla{\Gs}$. Let $k \in \nat$ be the rank of $\varphi_L$. Fact~\ref{fct:esuit} yields a \Gs-morphism $\eta: A^* \to G$ such that $\varphi_L$ is a \tla{\eta} formula. Moreover, since $s,t$ belong to the \Gs-kernel $N$ of $\alpha$ and $\eta$ is a \Gs-morphism, Lemma~\ref{lem:morker} yields $u,v \in A^*$ such that $\eta(u) = \eta(v) = 1_G$, $\alpha(u) = s$ and $\alpha(v) = t$. Hence, since $\varphi_L$ is an $\eta$-formula of rank $k$, it follows from Lemma~\ref{lem:efcong} and Proposition~\ref{prop:efg} (applied for $z = \veps$ which also maps to $1_G$) that $x(uv)^{k}(uv)^{k}y \in L \Leftrightarrow x(uv)^{k} v(uv)^{k}y \in L$ for all $x,y \in A^*$. In other words, $(uv)^{k}(uv)^{k}$ and $(uv)^{k} v(uv)^{k}$ are equivalent for the syntactic congruence of $L$. Since $\alpha$ is the syntactic morphism of $L$, we get $\alpha((uv)^{k}(uv)^{k})=\alpha((uv)^{k} v(uv)^{k})$. By definition of $u$ and $v$, it follows that $(st)^{k} (st)^{k} = (st)^{k} t(st)^{k}$. It now suffices to multiply by enough copies of $st$ on the left and on the right to obtain $(st)^\omega = (st)^\omega t (st)^\omega$, as desired.
\end{proof}

Given as input a regular language $L$, one can compute its syntactic morphism. Moreover, by Lemma~\ref{lem:kercomp}, the \Gs-kernel of this morphism can be computed when \Gs-separation is decidable. It is then simple to decide whether it belongs to \davar: this boils down to checking whether~\eqref{eq:da} holds by testing all possible combinations for $s$ and $t$. By Theorem~\ref{thm:glogic}, this decides if $L \in \upol{\bpol{\Gs}}$. Altogether, we obtain the following corollary.

\begin{cor}\label{cor:glogic}
  Let \Gs be a \vari of group languages with decidable separation. Then, membership is decidable for the class $\fod(<,\prefsigg) =  \tlc{\Gs} = \dec{2}(<,\prefsigg) = \upol{\bpol{\Gs}}$.
\end{cor}

Recall that separation is decidable for the four standard \varis of group languages \stzer, \md, \abg and \grp that we presented in Section~\ref{sec:group-languages}. Therefore, Corollary~\ref{cor:glogic} applies in these cases. In particular, we get the decidability of membership for the logical classes \fodw, \fodwm, \fodwam and \fodwgm.

Again, Theorem~\ref{thm:glogic} generalizes classic results. The most prominent one is for~$\Gs = \stzer$. Since $\stzer = \{\emptyset,A^*\}$, it is easy to check that the \stzer-kernel of a surjective morphism $\alpha: A^*\to M$ is the whole monoid $M$. Also by Lemma~\ref{lem:upolat}, we have $\upol{\bpol{\stzer}} = \upol{\at}$. Therefore, in this case, Theorem~\ref{thm:glogic} yields that $\upol{\at} = \dewd = \fodw =\fpfm$ and that a regular language belongs to this class if and only if its syntactic monoid belongs to \davar. Historically, this was proved by combining several independent results. First, the correspondence between \upol{\at} and \davar is due to Schützenberger~\cite{schul}. Then, it was shown by Pin and Weil~\cite{pwdelta2} that $\upol{\at} = \dewd$. As we explained in the previous section, that $\fodw = \fpfm$ was proved by Etessami, Vardi and Wilke~\cite{evwutl}. Finally, the correspondence between \fodw and \davar is due to Thérien and Wilke~\cite{twfo2}.

Another well-known application of Theorem~\ref{thm:glogic} is the case where $\Gs = \md$. In particular, the theorem implies that $\dewmd = \fodwm$ and that a regular language belongs to this class if and only if the \md-kernel of its syntactic morphism belongs to \davar. The correspondence between \fodwm and the membership of the \md-kernel of the syntactic morphism to \davar is due to Dartois and Paperman~\cite{DartoisP13}, while the equality $\dewmd = \fodwm$ is due to Kufleitner and Walter~\cite{KufleitnerW15}.

\begin{rem}
  While \textrm{membership} for \upol{\bpol{\Gs}} boils down to \Gs-separation, not much is known about \emph{separation} and \emph{covering} for \upol{\bpol{\Gs}}. The only class of this kind for which separation and covering are known to be decidable is $\upol{\bpol{\stzer}} = \fodw$. This is the class \upol{\at} by Lemma~\ref{lem:upolat}. Since \at is a \emph{finite} class, \upol{\at}-covering is decidable by Corollary~\ref{cor:upolopt}.
\end{rem}

\subsection{\Wsuit extensions} We now consider classes of the form \upol{\bpol{\Gs^+}}, where~$\Gs^+$ is the \wsuit extension of an arbitrary \vari of group languages \Gs. In this case, the characterization is based on the class of finite semigroups \ldavar. Again, we use a definition which is based on an equation. A finite semigroup $S$ belongs to \ldavar when it satisfies the following equation:
\begin{equation} \label{eq:lda}
  (esete)^\omega = (esete)^\omega ete (esete)^\omega \quad \text{for every $s,t \in S$ and $e \in E(S)$}.
\end{equation}

We now present the generic statement characterization of \upol{\bpol{\Gs^+}}.

\begin{theorem}\label{thm:wglogic}
  Let \Gs be a \vari of group languages and let $L$ be a regular language. The following properties are equivalent:
  \begin{enumerate}
    \item $L \in \upol{\bpol{\Gs^+}}$.
    \item $L \in \capol{\bpol{\Gs^+}}$.
    \item $L \in \dec{2}(<,+1,\prefsigg)$.
    \item $L \in \fod(<,+1,\prefsigg)$.
    \item $L \in \tlxc{\Gs}$.
    \item The strict \Gs-kernel of the syntactic morphism of $L$ belongs to \ldavar.
  \end{enumerate}
\end{theorem}

\begin{proof}
  In this case as well, we already proved most of the implications. By Lemma~\ref{lem:gensig}, we have $\dec{2}(<,+1,\prefsigg) = \dec{2}(\infsiggp)$. Hence, $(2) \Leftrightarrow (3)$ follows from Theorem~\ref{thm:capoldel2}. Moreover, $(1) \Leftrightarrow (2)$ is given by Theorem~\ref{thm:polcopol}. Finally, since $\fod(<,+1,\prefsigg) = \fod(\infsiggp)$ by Lemma~\ref{lem:gensig} and $\tlxc{\Gs} = \tlc{\Gs^+}$ by Lemma~\ref{lem:utlx}, the equivalence $(4) \Leftrightarrow (5)$ is immediate from Theorem~\ref{thm:temporal}. Finally, the implication $(1) \Rightarrow (5)$ follows from Proposition~\ref{prop:ubinutl}. We now prove independently that $(6) \Rightarrow (2)$ and $(5) \Rightarrow (6)$ to complete the proof. We let $\alpha: A^* \to M$ be the syntactic morphism of $L$. Let also $N$ be the \Gs-kernel of $\alpha$, and $S$ be its strict \Gs-kernel. Recall from Section~\ref{sec:cano} that we have $S = N \cap \alpha(A^+)$ by definition.

  We start with the implication $(6) \Rightarrow (2)$. Assume that $S$ belongs to \ldavar. We have to show that $L \in \capol{\bpol{\Gs^+}}$. Since $L$ is recognized by its syntactic morphism $\alpha$, it suffices to show that $\alpha$ is a $(\capol{\bpol{\Gs^+}})$-morphism. By Theorem~\ref{thm:ubp}, this boils down to proving that for every $s,t \in M$ and $e \in E(M)$ such that $(e,s) \in M^2$ is a $\Gs^+$-pair, we have $(eset)^{\omega+1} = (eset)^{\omega}et(eset)^{\omega}$. By Fact~\ref{fct:idemker}, we have $e \in N$. We consider two cases depending on whether $e \in S$ or $e \in N \setminus S$. Assume first that $e \in N \setminus S$. Since $S = N \cap \alpha(A^+)$, this implies that $e = 1_M$ and that $\alpha\inv(1_M) = \{\veps\}$. Since $(1_M,s)$ is a $\Gs^+$-pair and $\{\veps\} \in \Gs^+$, it follows that $s = 1_M$ as well. Hence, $(eset)^{\omega+1} = t^{\omega+1} =  (eset)^{\omega}et(eset)^{\omega}$ follows directly. We now consider the case where $e \in S$. By Lemma~\ref{lem:ptoker}, we have $es$ and $et(eset)^{2\omega-1} \in N$. Since $e \in S$, this implies that $es \in S$ and $et(eset)^{2\omega-1} \in S$. Since $S$ belongs to \ldavar, we get from Equation~\eqref{eq:lda}~that,
  \[
    (e es e et(eset)^{2\omega-1}e)^\omega = (e es e et(eset)^{2\omega-1}e)^\omega  eet(eset)^{2\omega-1}e (eese et(eset)^{2\omega-1}e)^\omega.
  \]
  Since $e$ is idempotent, this simplifies into $(eset)^{\omega}e = (eset)^{\omega} et (eset)^{3\omega-1}e$. It now suffices to multiply by $set$ on the right to get the desired equality $(eset)^{\omega+1}=(eset)^{\omega}et(eset)^{\omega}$.

  We finally turn to the implication $(5) \Rightarrow (6)$.  Assume that $L \in \tlxc{\Gs} = \tlc{\Gs^+}$. We show that $S$ belongs to \ldavar, \emph{i.e.}, that $S$ satisfies~\eqref{eq:lda}. Let $s,t \in S$ and $e \in E(S)$. We prove that $(esete)^\omega = (esete)^\omega ete (esete)^\omega$. By hypothesis, $L$ is defined by a formula $\varphi_L$ of $\tla{\Gs^+}$. Let $k \in \nat$ be the rank of $\varphi$. Fact~\ref{fct:esuit} yields a $\Gs^+$-morphism $\eta: A^* \to Q$ such that $\varphi_L$ is a \tla{\eta} formula. It follows from Lemma~\ref{lem:gpmorph} that $G= \eta(A^+)$ is a group in $N$ and that the morphism $\beta: A^* \to G$ defined by $\beta(w) = \eta(w)$ for every $w \in A^+$ is a \Gs-morphism. Using the hypothesis that $s,t,e \in S$, we prove that there exist $u,v,z \in A^*$ such that $\alpha(u) = s$, $\alpha(v) = t$, $\alpha(z) = e$ and $\eta(u) = \eta(v) = \eta(z) = 1_G$. By symmetry, we only prove the existence of $u \in A^*$. There are two cases, depending on whether $s = 1_M$ or not. If $s = 1_M$, then we have $1_M \in S \subseteq \alpha(A^+)$. Thus, there exists $u' \in A^+$ such that $\alpha(u') = 1_M$. We let $u = (u')^p$ for $p = \omega(G)$. Clearly, $\alpha(u) = 1_M = s$ and $\eta(u) = (\eta(u'))^p = 1_G$ since $G$ is a group. Assume now that $s \neq 1_M$. We have $s \in S \subseteq N$. Since $\beta: A^* \to G$ is a \Gs-morphism, it follows from Lemma~\ref{lem:morker} that there exists $u \in A^*$ such that $\alpha(u) = s$ and $\beta(u) = 1_G$. Since $s \neq 1_M$, we have $u \in A^+$, which also implies that $\eta(u) = \beta(u) = 1_G$, as~desired.

  We now use $u,v,z$ (satisfying $\alpha(u) = s$, $\alpha(v) = t$, $\alpha(z) = e$ and $\eta(u) = \eta(v) = \eta(z) = 1_G$) to conclude the proof. Since $\varphi_L$ is an $\eta$-formula of rank $k$ and $1_G \in E(N)$, it follows from Lemma~\ref{lem:efcong} and Proposition~\ref{prop:efg} that for all $x,y \in A^*$, we have,
  \[
    x(z^{k}uz^{2k}vz^{k})^{k}(z^{k}uz^{2k}vz^{k})^{k}y \in L \Leftrightarrow x(z^{k}uz^{2k}vz^{k})^{k} z^{k}vz^k(z^{k}uz^{2k}vz^{k})^{k}y \in L.
  \]
  Hence, $(z^{k}uz^{2k}vz^{k})^{k}(z^{k}uz^{2k}vz^{k})^{k}$ and $(z^{k}uz^{2k}vz^{k})^{k} z^{k}vz^k(z^{k}uz^{2k}vz^{k})^{k}$ are equivalent for the syntactic congruence of $L$. Since $\alpha$ is the syntactic morphism of $L$, these words have the same image under $\alpha$. This yields $(esete)^{k} (esete)^{k} = (esete)^{k} ete (esete)^{k}$ by definition of $u,v$ and $z$. It now suffices to multiply by enough copies of $esete$ on the left and on the right to obtain $(esete)^\omega = (esete)^\omega ete (esete)^\omega$, as desired.
\end{proof}

Given as input a regular language $L$, one can compute its syntactic morphism. Moreover, by Lemma~\ref{lem:kercomp}, the strict \Gs-kernel of this morphism can be computed as soon as \Gs-separation is decidable. It is then simple to decide whether it belongs to \ldavar: this boils down to checking if~\eqref{eq:lda} holds by testing all possible combinations for $s$, $t$ and $e$. By Theorem~\ref{thm:wglogic}, this decides whether $L \in\upol{\bpol{\Gs^+}}$. We state this in the following corollary.

\begin{cor}\label{cor:wglogic}
  Let \Gs be a \vari of group languages with decidable separation. Then, membership is decidable for $\fod(<,+1,\prefsigg) =  \tlxc{\Gs} = \dec{2}(<,+1,\prefsigg) =\! \upol{\bpol{\Gs^+}}$.
\end{cor}

Recall again that separation is decidable for the \varis of group languages \stzer, \md, \abg and \grp, so that Corollary~\ref{cor:wglogic} applies in these cases: membership is decidable for the logical classes \fodws, \fodwsm, \fodwsam and \fodwsgm.

Here again, Theorem~\ref{thm:wglogic} generalizes known results in the particular case $\Gs = \stzer$. Since $\stzer = \{\emptyset,A^*\}$, it is straightforward to verify that the strict \stzer-kernel of a morphism $\alpha: A^*\to M$ is the semigroup $\alpha(A^+)$. When $\alpha$ is the syntactic morphism of a language $L$, this object is called the \emph{syntactic semigroup} of $L$. Thus, Theorem~\ref{thm:wglogic} yields that $\dewsd = \fodws =\fpfmx$ and that a regular language belongs to this class if and only if its syntactic semigroup belongs to \ldavar. As seen in the previous section, the equality $\fodws = \fpfmx$ is due to Etessami, Vardi and Wilke~\cite{evwutl}. The remaining correspondences are due to Thérien and Wilke~\cite{twfo2}. In particular, the connection with \ldavar is established by relying on results of Almeida~\cite{dadalmeida}, based on a complex algebraic framework by Tilson~\cite{tilson}, which involves categories and wreath products of finite semigroups. We bypass this intricate machinery here.

\begin{rem}
  In this case as well, there exists no generic result concerning separation and covering for classes of the form \upol{\bpol{\Gs^+}}. However, both problems are known to be decidable in \emph{two} particular cases: $\Gs = \stzer$ and $\Gs = \md$. In view of Theorem~\ref{thm:wglogic}, this corresponds to the logical classes \fodws and \fodwsm. The decidability of covering for these two classes can be obtained by applying generic transfer results of~\cite{pzsucc} and~\cite{prwmodulo} to the simpler class \fodw (as we mentioned above, decidability of covering for this class was first proved in~\cite{pzcovering2}). The techniques involved in these results are orthogonal and independent from the ones of the present paper.
\end{rem}

\section{Conclusion}
\label{sec:conc}
We presented a generic algebraic characterization of unambiguous polynomial closure and used it to prove that when \Cs is a \vari, membership for \upol{\Cs} reduces to membership for~\Cs. An interesting byproduct of the proof is that $\upol{\Cs} = \adet{\Cs} = \capol{\Cs}$ in that case. Moreover, we showed that when \Cs is a finite \vari, covering and separation are decidable for \upol{\Cs}. This completes similar results of~\cite{pzbpol} for \pol{\Cs} and \bool{\pol{\Cs}} and of~\cite{pseps3j} for \pol{\bool{\pol{\Cs}}}. Finally, we presented logical characterizations of the classes built with unambiguous polynomial closure. In particular, we proved that if the input class \Gs is a \vari of \emph{group languages}, then $\upol{\bpol{\Gs}} = \dec{2}(<,\prefsigg) = \fod(<,\prefsigg) = \tlc{\Gs}$ and $\upol{\bpol{\Gs^+}} = \dec{2}(<,+1,\prefsigg) = \fod(<,+1,\prefsigg) = \tlxc{\Gs}$. This generalizes earlier results corresponding to particular examples of \varis of group languages \Gs.

A natural follow-up question is whether our result for separation and covering can be pushed further to encompass more inputs than the finite classes. In particular, in view of the above logical characterizations, it would be interesting to look at both problems for classes of the form \upol{\bpol{\Gs}} and \upol{\bpol{\Gs^+}} when \Gs is a \vari of group languages. While this is difficult, there are known results of this kind; in particular, it is known~\cite{pzconcagroup} that separation and covering are decidable for \pol{\bpol{\Gs}} as soon as separation is decidable~for~\Gs.

\printbibliography

\end{document}